\pdfoutput=1
\documentclass{article}
\usepackage[sort,authoryear]{natbib}
\usepackage{mathtools}
\usepackage{amssymb}
\usepackage{amsthm}
\usepackage{enumitem}
\usepackage[margin=1in]{geometry}
\usepackage{graphicx}
\usepackage{subcaption}
\usepackage[utf8]{inputenc}
\usepackage[english]{babel}
\usepackage{xcolor}
\usepackage[colorlinks = true, citecolor = blue, urlcolor = blue]{hyperref}
\PassOptionsToPackage{hyphens}{url}\usepackage{hyperref}
\usepackage{indentfirst}
\usepackage{apptools}
\usepackage{cleveref}
\usepackage{algorithm}
\usepackage[noend]{algpseudocode}

\graphicspath{{fig/}}

\AtAppendix{\counterwithin{lemma}{section}}
\AtAppendix{\counterwithin{proposition}{section}}

\DeclareMathOperator*{\argmin}{argmin}
\DeclareMathOperator*{\argmax}{argmax}

\newtheorem{lemma}{Lemma} 
\newtheorem{example}{Example} 
\newtheorem{definition}{Definition} 
\newtheorem{proposition}{Proposition}
\newtheorem{theorem}{Theorem}
\newtheorem{corollary}{Corollary}

\newcommand{\Var}{\textup{Var}}
\newcommand{\R}{\mathbb{R}}
\newcommand{\X}{\boldsymbol{X}}

\title{Selective inference is easier with p-values}
\author{Anav Sood\\ Stanford University}
\date{\today}

\begin{document}

\maketitle

\begin{abstract}
Selective inference is a subfield of statistics that enables valid inference after selection of a data-dependent question. In this paper, we introduce selectively dominant p-values, a class of p-values that allow practitioners to easily perform inference after arbitrary selection procedures. Unlike a traditional p-value, whose distribution must stochastically dominate the uniform distribution under the null, a selectively dominant p-value must have a post-selection distribution that stochastically dominates that of a uniform having undergone the same selection process; moreover, this property must hold simultaneously for all possible selection processes. Despite the strength of this condition, we show that all commonly used p-values (e.g., p-values from two-sided testing in parametric families, one-sided testing in monotone likelihood ratio and exponential families, $F$-tests for linear regression, and permutation tests) are selectively dominant. By recasting two canonical selective inference problems-inference on winners and rank verification-in our selective dominance framework, we provide simpler derivations, a deeper conceptual understanding, and new generalizations and variations of these methods. Additionally, we use our insights to introduce selective variants of methods that combine p-values, such as Fisher's combination test. 
\end{abstract}

\section{Introduction}

Selective inference is a sub-field of statistics that allows practitioners to make valid inferences even when the underlying statistical question is chosen by a data-driven selection process. However, many selective methods, which operate by conditioning on the selection event, can be difficult to derive, hard to implement, and exhibit counter-intuitive behavior. To statisticians outside of the sub-field, each selective procedure may seem to come from a different argument or approach.

In this paper, we provide a unifying framework for selective inference centered around p-values. So long as a statistician knows how to compute traditional  p-values, our framework enables them to easily deliver valid hypothesis tests and confidence intervals, even when the inferential question at hand arose from an arbitrary, albeit known, selection procedure. Our framework (1) can greatly simplify the process of designing new selective methods and (2) results in more natural and general derivations of some existing selective methods, allowing for a deeper understanding of their behavior as well as new variations and extensions.

\subsection{Motivation}

\begin{figure}
    \centering
    \scalebox{1}{
    \hspace{-0.025 \textwidth}
    \begin{minipage}{0.35\textwidth}
        \centering
        \includegraphics[width=\textwidth]{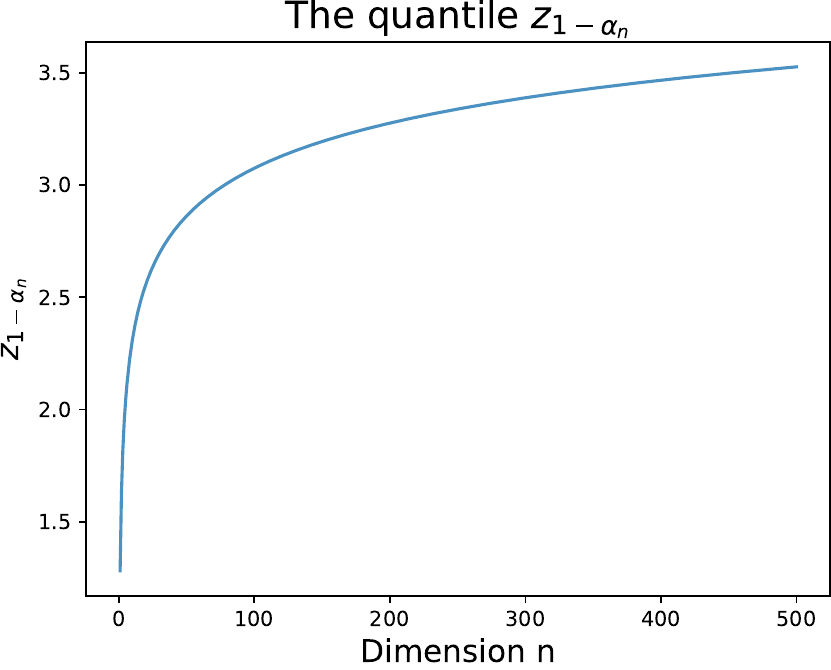}
    \end{minipage}
    \hspace{0.05 \textwidth}
    \begin{minipage}{0.35\textwidth}
        \centering
        \includegraphics[width=\textwidth]{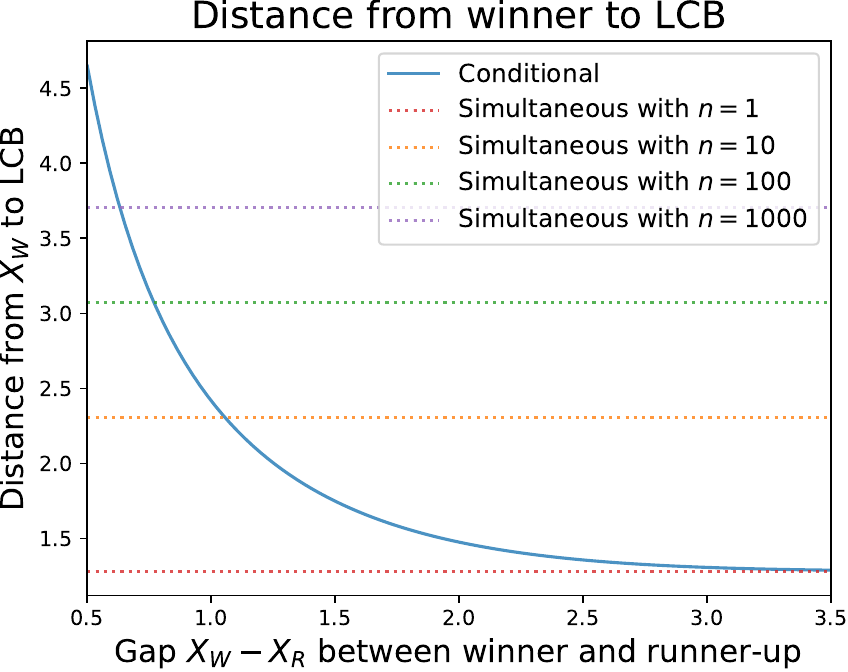}
    \end{minipage}
    }
    \caption{ The first panel (left) shows the growth of the quantile $z_{1 - \alpha_n}$ as a function of the dimension $n$. The second panel (right) gives the distance between the level $\alpha=0.1$ conditional LCB and the winner $X_W$ as a function of the gap $X_W - X_R$ between winner and the runner-up. For dimensions $n=1,10,100,1000$, it also gives the distance from $X_W$ to the level $\alpha=0.1$ simultaneous LCB.}
    \label{fig:winner}
\end{figure}

As motivation, we recall the age-old task of making inferences from the largest observation. Considering Gaussian data $X \sim N(\mu, I_n)$ with unknown mean $\mu$, the largest observation $W = \argmax_{i \in [n]} X_i$ may well correspond to the largest mean. Therefore, a natural way to verify the existence of a large $\mu_i$ is to give a lower confidence bound (LCB) for $\mu_W$, the mean of the winning value. 

Traditionally, we provide such an LCB via Sidak's simultaneous approach. Let $z_{1-\alpha}$ denote the $1-\alpha$ quantile of a standard normal distribution and define $\alpha_n = 1 - (1-\alpha)^{1/n}$. Sidak's approach tells us that the lower bounds $ \mu_i > X_i - z_{1 - \alpha_n}$ all simultaneously hold with probability $1-\alpha$. Therefore, $\hat{\mu}_{simul} = X_W - z_{1 - \alpha_n}$ is a valid lower bound for $\mu_W$ that holds with probability at least $1-\alpha$. Performing simultaneous inference on $n$ means, however, comes at a cost. As $n$ grows, the quantile $z_{1 - \alpha_n}$ grows as well (depicted in the left panel of \Cref{fig:winner}), and the distance from the winning observation to the LCB correspondingly increases. 

The more modern selective approach suggests providing an LCB that is valid conditionally on $W$. This approach provides inferences solely for the winning mean $\mu_W$, so we may hope that it avoids the simultaneous inference's curse of dimensionality. Following the recipe of \cite{Fithian2017} (see \Cref{sec:cond_appdx} for details), one finds that the conditional LCB $\hat{\mu}_{cond}$ is the infimum of the confidence region
\begin{equation}
    \label{eq:motivating_lcb}
    \{\mu_0 : \mu_0 > X_W - \text{Quantile}(1-\alpha, TN(0, 1, X_R - \mu_0, \infty))  \},
\end{equation}
where $X_R$ is the runner-up (second largest) observation and $TN(\mu, \sigma^2, a, b)$ is a $N(\mu, \sigma^2)$ distribution truncated to lie in the interval $[a, b]$. 

The conditional LCB \eqref{eq:motivating_lcb} is near impossible to parse at first glance, but it turns out to have some very interesting properties. As plotted in right panel of \Cref{fig:winner}, the distance $X_W - \hat{\mu}_{cond}$ between the winner and the conditional LCB is purely a function of the gap $X_W - X_R$ between the winner and runner-up. If the gap between $X_W$ and $X_R$ is large, then the conditional LCB for $\mu_W$ will be roughly $X_{W} - z_{1-\alpha}$, i.e., what we expect in a one-dimensional inference problem. But as the runner-up gets close to the winner, the conditional LCB explodes quickly to $-\infty$ and gives much worse inferences than the classical approach. In summary, the conditional method appears to avoid the curse of dimensionality when the winner is well separated, but when it is not, the consequences can be tremendous. 

The motivation for this article comes from the following fact: the conditional LCB becomes shockingly simple to parse once written in terms of p-values. Suppose we wanted just to check whether each LCB was non-negative, i.e., we used each LCB solely to verify the existence of a positive mean. Typically, we verify the existence of a positive mean by using the p-values $p_i = 1 - \Phi(X_i)$ to test the nulls $H_{0, i} : \mu_i \leq 0$. It turns out that the simultaneous LCB verifies the existence of a positive mean when smallest of these p-values $p_{(1)}$ is at most $\alpha_n$. In contrast, the conditional LCB does so when the ratio of the two smallest p-values $p_{(1)}/p_{(2)}$ is at most $\alpha$. 

At least when the p-values $p_i$ are uniform under the null, there appears to be a simple and intuitive proof for the conditional procedure's validity (much simpler than our derivations in \Cref{sec:cond_appdx}): under the null, the smallest p-value $p_{(1)}$ should be uniform on $[0, p_{(2)}]$, so we should control error if we reject when $p_{(1)} \leq \alpha p_{(2)}$. But the question remains, does this argument go through when p-values are not exactly uniform? If not, why does it work in the Gaussian case?

The framework we develop in this article gives a complete answer to these questions. It identifies a broad class of p-values for which the sort of argument we have outlined above goes through. Beyond inference on winners, it turns out to encompass a comprehensive set of selective inference problems, and it provides a clearer view on how selective methods behave and are derived. As we illustrate throughout the article, our framework serves to both deepen our understanding of existing methods in the selective inference literature, as well as enable to us to easily design new selective methods for a wide range of inferential problems. 

\subsection{Our contributions}

In this paper, we introduce the selective dominance framework, which we briefly summarize here. In our framework, we imagine using a valid p-value $p$ to test the null hypothesis $H_0$, but we only decide whether or not to test $H_0$ \underline{after} we have observed $p$. Our decision on whether or not to test $H_0$ is determined by a \textbf{selection function} $s(x)$: after observing that $p=x$, we use $p$ to test $H_0$ with probability $s(x)$. Due to bias from this selection, rejecting when $p \leq \alpha$ is no longer guaranteed to maintain Type I error control. If, however, $p$ is a \textbf{selectively dominant} p-value, then we can still maintain Type I error control if we reject when the adjusted p-value $\int_0^p s(x) dx/\int_0^1 s(x) dx$ is at most $\alpha$. We give a precise definition of selective dominance in \Cref{sec:dominance}, where we also show that all the p-values we commonly use are selectively dominant. For many selective inference problems, $s(x)$ turns out to be simple enough that this adjustment is computable is closed form, and the resulting procedure is correspondingly easy to implement and interpret. 

After presenting a few simple examples of how to apply our framework, the remainder of the article focuses on three specific applications, which we describe below. \newline 

\noindent \textbf{Inference on winners} (\Cref{sec:winner})\textbf{:} We characterize both the earlier conditional inference procedure \eqref{eq:motivating_lcb} and also hybrid inference \citep{Andrews2023}, an inference on winners procedure that interpolates between the conditional and simultaneous approaches, in terms of selectively dominant p-values. Once written in terms of p-values, both methods become significantly easier to derive, interpret, and implement. Doing so also provides a more general and fundamental view of these methods, enabling their application beyond the Gaussian data setting. Via our p-value viewpoint, we are able gain insight on (1) when the conditional LCB may be useful and (2) the extent to which the conditional LCB explodes to $-\infty$. Also, we are able to clearly illustrate that hybrid inference strictly dominates its naive union bound competitor, which splits the Type I error budget between the conditional and simultaneous approaches and then takes the better of the two. Unfortunately, our viewpoint also suggests that, while it results in a strict improvement, hybrid inference cannot do too much better than this naive alternative. \newline 

\noindent \textbf{Rank verification} (\Cref{sec:rank_verification})\textbf{:}  Our selective dominance framework allows us to very easily establish the following claim: if a researcher observes two Gaussian samples $X_i \sim N(\mu_i, 1)$  and rejects the null $H_0: \mu_1 = \mu_2$ via the standard two-sided test, then they can further infer that the winning sample has the strictly larger of the two means. While this result was previously known, we show, perhaps more surprisingly, that if a researcher wants just to verify that the mean of the winning sample is \underline{at least} as large as the other mean, they can run a one-sided test comparing the winner to the loser with \underline{no correction} for selection. Finally, we illustrate how \cite{Hung2019}'s rank verification for exponential families procedure fits nicely in our framework. In doing so, we show that the original method does not appropriately handle settings with ties, but that our framework naturally handles these settings properly. \newline

\noindent \textbf{Combining selective p-values } (\Cref{sec:multiple}) \textbf{:} Considering a meta-analysis problem where we have p-values from an original and replication study, we propose a new data-carving method that combines inferences from both studies while still accounting for potential publication bias in the original study. On the open science collaboration dataset, which contains 92 pairs of original replication psychology studies \citep{OSF}, our method finds that 47 pairs of studies are significant, while using solely the replication p-value results in only 34 significant findings. In the same vein, we provide variants of Fisher's combination test that are more powerful when some null p-values are conservative (i.e., super-uniform), including a version of Fisher's truncated combination test \citep{Zaykin} that allows for super-uniform p-values. Previous versions of this test required p-values to be exactly uniform under the null, preventing the test's application in settings where it is most powerful. \newline 

The code used to run the experiments and generate the figures in this article can be found at \url{github.com/AnavSood/seldom}.

\subsection{Related work}

We aim to provide a unifying view of modern conditional selective inference, a field started by the seminal work \cite{Lee2016} that has been the subject of much academic research over the last decade \citep{Panigrahi, Benjamini, Tian2018Apr, Tian2018Dec,Taylor2014, Taylor2016, Taylor2018, Tibshirani, Fithian2015, Fithian2017, Loftus, Chen2023Apr, Chen2023May, Gao, Lee2014, Hyun}. Alongside this contemporary body of work, there is an earlier line of work that studies drawing inferences from data-dependent questions, particularly as it pertains to rank verification \cite{Bofinger, Fabian, Gutmann, Maymin, Hsu}. Outside of conditional selective inference, there is a class of selective methods that avoid conditioning on the selection event (e.g, \citep{Romano, Zrnic}), but \cite{Perry} suggest empirically that such approaches may be less powerful than their conditional counterparts. We do not discuss such methods in this article. 

We point out a few articles that are most relevant to our work. \cite{Hung2020} mention that p-values coming from $t$- and $z$-tests satisfy our selective dominance condition, but only for thresholding selection functions (i.e., $s(x) = I(x \leq \alpha)$). Both \cite{Markovic, Andrews2019} consider settings where the selection probabilities are unknown, and provide either (1) methods for learning them or (2) theoretical results regarding to what extent they are identified. When we treat our procedures as global null testing procedures and close them, we get sequential selective hypothesis that are similar to those in \cite{Fithian2015}, although the procedures in \cite{Fithian2015} specifically test two sided nulls while our nulls are arbitrary.

\section{Selectively dominant p-values}
\label{sec:dominance}

In this section we define selectively dominant p-values, a class of p-values that enable us to easily do inference after selection. We give a precise characterization of when p-values are selectively dominant and illustrate that the most commonly used p-values are all selectively dominant. Finally, we provide a few examples of how one may apply our selective dominance framework. 

Throughout the remainder of the article, we will often encounter expressions like $P(A|X, B)$, where $P$ is a probability measure, $A$ and $B$ are events, and $X$ is a random variable. If $B$ is a probability zero event, i.e., $P(B) = 0$, we set $P(A|X, B)$ to be identically zero. If not, we can define the conditional probability measure $Q(\cdot) = P(\cdot| B)$ given the event $B$, and we consider $P(A | X, B)$ to be the conditional expectation $E_Q[I_A| X]$ of the indicator $I_A$ under $Q$. Also, for a null hypothesis $H_0$, the notation $P_{H_0}$ refers to an arbitrary distribution under the null. 

\subsection{Selective dominance}

\begin{figure}[]
    \centering
    \scalebox{1}{
    \hspace{-0.03\textwidth}
    \begin{minipage}{0.33\textwidth}
        \centering
        \includegraphics[width=\textwidth]{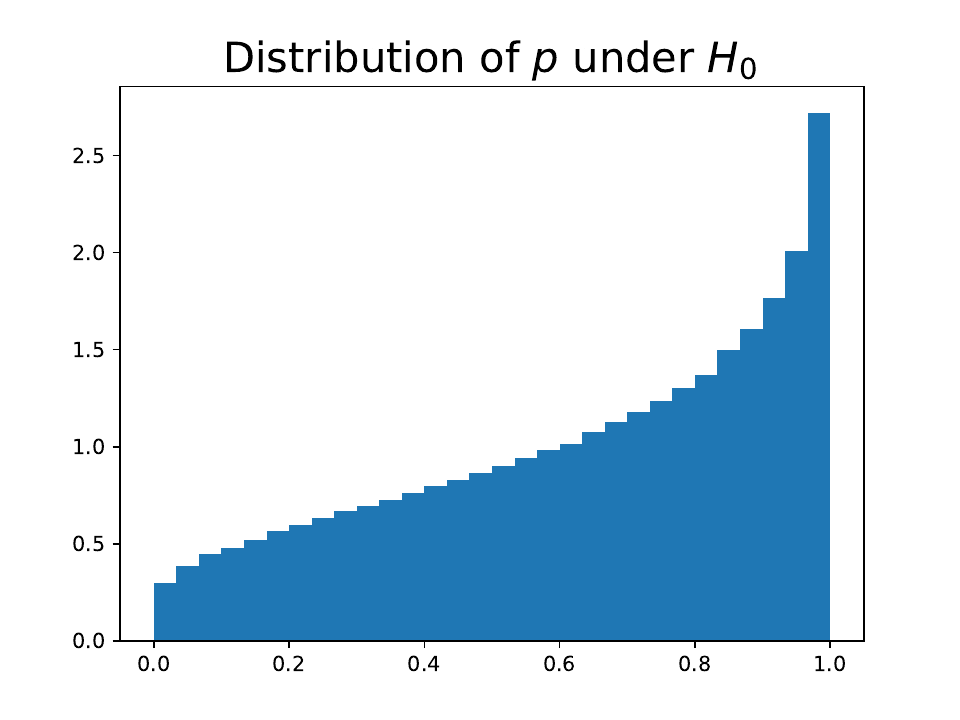}
    \end{minipage}\hfill
    \hspace{0.02\textwidth}
    \begin{minipage}{0.33\textwidth}
        \centering
        \includegraphics[width=\textwidth]{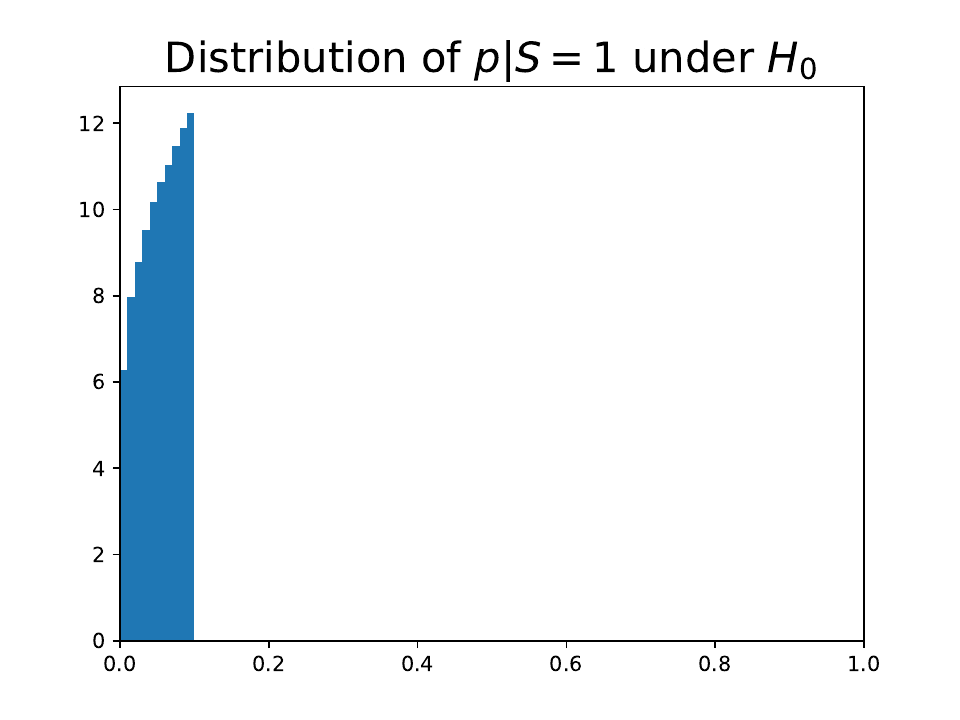}
    \end{minipage}\hfill
    \hspace{0.02\textwidth}
    \begin{minipage}{0.33\textwidth}
        \centering
        \includegraphics[width=\textwidth]{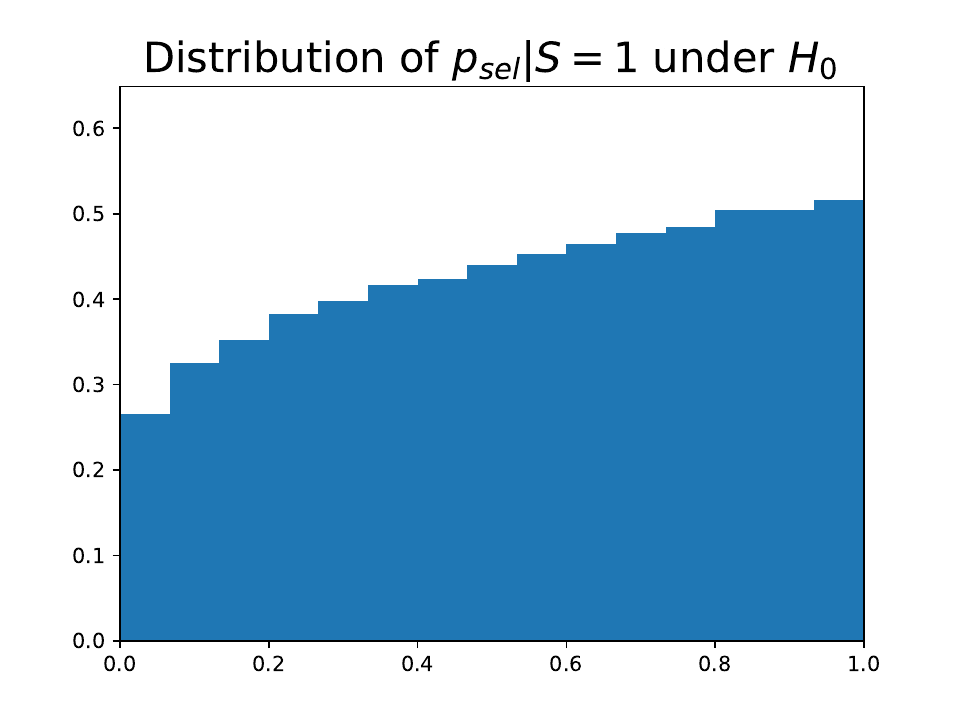}
    \end{minipage}
    }
    \caption{The first panel (left) depicts an example distribution of a p-value $p$ under the null. The distribution stochastically dominates the uniform distribution. The second panel (middle) depicts the distribution $p | S = 1$ of the same null p-value $p$ given that it was selected for being most $\alpha=0.1$. This distribution does \underline{not} stochastically dominate the uniform distribution. The third panel (right) depicts the null distribution $p_{sel}|S=1$ of \Cref{thm:adjustment}'s selective p-value $p/\alpha$ given selection. Thanks to \Cref{thm:adjustment}'s correction, this distribution again is stochastically dominates  the uniform distribution.}
    \label{fig:distributions}
\end{figure}

In classical statistics, a p-value is a random variable $p \in [0, 1]$ that stochastically dominates the uniform distribution $U \sim \text{Unif([0, 1])}$ under the null $H_0$:

\begin{equation*}
    P_{H_0}(p \leq t) \leq P(U \leq t)= t \text{ for all } t \in [0, 1].
\end{equation*}
The left-most panel of \Cref{fig:distributions} illustrates what the distribution of a null p-value may look like. When working with p-values, we maintain Type I error control if we reject $H_0$ when $p$ is small:
\begin{equation*}
    P_{H_0}(\text{reject null}) = P_{H_0}(p \leq \alpha) \leq P(U \leq \alpha)  = \alpha.
\end{equation*}
In essence, stochastic dominance allows us to use the uniform distribution as a reference distribution. We control the probability of $p$ being small under the null by comparing it to the probability of a uniform being small. 

In the problems we consider, we use $p$ to test the null $H_0$ only after it has been ``selected''. In full generality, we consider a p-value $p$ for testing the null $H_0$ that is conditionally valid given some random vector $Z$:
\begin{equation}
    \label{eq:valid_given_z}
    P_{H_0}(p \leq t | Z) \leq P(U \leq t)  = t \text{ for all } t \in [0, 1],  \text{ a.e. under } P_{H_0}.
\end{equation}
The random vector $Z$ plays an integral role in determining the relationship between $p$ and a binary selection random variable $S \in \{0, 1\}$, which takes value one when $p$ is selected and zero otherwise. This relationship is governed by a \textbf{selection function} \footnote{Formally, the selection function $s(\cdot, \cdot)$ such is an instantiation of the conditional expectation $s(p, Z) = E[S|p , Z ]$. },
\begin{equation*}
    s(x, z) = P(S = 1 \mid p = x, Z = z).
\end{equation*}
Intuitively, we imagine observing $p=x$ and $Z=z$ and then flipping a biased coin that comes up heads with probability $s(x, z)$. We only use $p$ to test the null $H_0$ when the coin comes up heads, and otherwise do not perform inference. This process turns out to capture what happens in a wide range of selective inference problems. In the problems we consider, we decide the selection process, so $s(x, z)$ is known. In cases where no $Z$ is present, we can imagine $Z=0$ everywhere and write the selection function $s(x)$ purely in terms of $x$. 

Because we only perform inference when the coin comes up heads, our goal should be to design a procedure that controls Type I error conditional on this selection event:
\begin{equation}
    \label{eq:selective_error_control}
    P_{H_0}(\text{reject } H_0 \mid S = 1) \leq \alpha 
\end{equation}
As illustrated by our next example, the classical approach of rejecting when $p \leq \alpha$ does not maintain selective Type I error control as in \eqref{eq:selective_error_control}. 

\begin{example}[Publication bias and the failure of classical inference]
\label{exm:publication_bias}
Consider a p-value $p$ that is uniform on $[0, 1]$ under the null $H_0$. If we use the selection function $s(x) = I(p \leq \alpha)$, i.e., we select $p$ when it is at most $\alpha$, then $p \mid S=1 \sim \text{Unif}([0, \alpha])$. Our classical procedure will clearly fail to control Type I error conditional on selection:
\begin{equation*}
    P_{H_0}(\text{reject } H_0 \mid S=1) = P_{H_0}(p \leq \alpha \mid S=1) = 1 > \alpha. 
\end{equation*}
\end{example}

\Cref{exm:publication_bias} is a standard example in the literature on publication bias. If researchers publish studies only when they observe a p-value below $\alpha$, readers always see significant findings regardless of the true data generating distribution. As a consequence, the reader's observed Type I error rate can be as high as one. \Cref{fig:distributions}'s middle panel displays the distribution of the left panel's null p-value, but after it has been selected for inference via \Cref{exm:publication_bias}'s selection process. It is clear from the picture that the null p-value distribution no longer stochastically dominates the uniform distribution after selection. 

Essentially, after selection, the uniform distribution no longer suffices as a reference distribution. Naturally, we may instead try and use the distribution of a uniform after it has been selected by the same selection process. Formally, let $U \sim \text{Unif}([0, 1])$ be a uniform random variable that exists on the same probability space as $p$ and is independent of $Z$, and let $S' \in \{0, 1\}$ be a different binary selection random variable whose joint distribution with $U$ and $Z$ is governed by the same selection function \footnote{Again, formally $E[S'| U, Z]= s(U, Z)$. }
\begin{equation*}
     P(S' = 1 \mid U = x, Z=z ) = s(x, z).
\end{equation*}
Instead of the uniform distribution, we can use the the conditional distribution $U \mid Z,  S' = 1$ of $U$ given selection as our reference distribution. This approach is valid exactly when our p-value is selectively dominant, as we define below. 

\begin{definition}[Selective dominance]
    \label{def:selective_dominance}
    Consider a p-value $p$ for the null $H_0$ that is valid given $Z$ as in \eqref{eq:valid_given_z}. We say that $p$ is \textbf{selectively dominant given $Z$} if, under the null $H_0$, it has a conditional probability density function (PDF) given $Z$, and it satisfies 
    \begin{equation}
    \label{eq:selective_dominance}
    P_{H_0}(p \leq t | Z, S=1) \leq P(U \leq t | Z, S'=1) \text{ for all } t \in [0, 1], \text{ a.e. under } P_{H_0}(\cdot |S=1)
    \end{equation}
    for every selection function $s(x, z)$.  
\end{definition}

Essentially selective dominance requires that, under the null, the conditional distribution of $p$ given selection should stochastically dominate the conditional distribution of $U$ the same selection. Moreover, this should be true for \underline{any} selection process. Again, we mention that if $P(S=1) = 0$ (resp. $P(S'=1) = 0$) for some selection function, then we define the left-hand side (resp. right-hand side) of \eqref{eq:selective_dominance} to be zero. We show in \Cref{sec:adjustment_proof} that $P(S' = 1) = 0 \implies P(S=1) = 0$, so this convention does not cause any issues.  

As we will soon see, all the p-values that practitioners commonly use are selectively dominant as described in \eqref{eq:selective_dominance}. In \Cref{def:selective_dominance}, we restrict to p-values with conditional PDFs under the null because it makes our theory and methods simpler to state. Because we can always make a p-value both have a conditional PDF and be more powerful via randomization, this restriction is never a practical issue. Also, after applying our machinery with randomized p-values, the user can always de-randomize the resulting method if they would like.

To perform valid post-selection inference using a selectively dominant p-value, we can transform it so that it remains a p-value after selection. As \Cref{thm:adjustment} displays, we can ``undo'' the effects of selection by applying the conditional cumulative distribution function (CDF) $F_{U \mid Z, S' = 1}(\cdot)$ of $U$ given $Z$ and selection to $p$. In line with prior literature, we refer to this transformed p-value as a \textbf{selective p-value}. For simple selection functions, this selective p-value is often computable in closed form. 

\begin{theorem}[Selective dominance and error control]
    \label{thm:adjustment}
    Consider a p-value $p$ for the null $H_0$ that is selectively dominant given $Z$ as in \Cref{def:selective_dominance}. Then, for any selection function $s(x, z)$, the selective p-value 
    \begin{equation}
    \label{eq:adjustment}
        p_{sel} = F_{U \mid Z, S' = 1}(p) = \frac{\int_0^p s(x, Z) dx}{\int_0^1 s(x, Z) dx}
    \end{equation}
    has a null distribution that stochastically dominates the uniform conditional on $Z$ and selection:
    \begin{align}
        &P_{H_0}(p_{sel} \leq t \mid Z, S= 1) \leq t  \text{ for all } t \in [0, 1], \text{ a.e. under } P_{H_0}(\cdot |S=1), \label{eq:adjusted_error_control_cond} \\
        &P_{H_0}(p_{sel} \leq t\mid S= 1) \leq t \text{ for all } t \in [0, 1]. \label{eq:adjusted_error_control_marg}
    \end{align}
    Further, if $P$ is a distribution in $H_0$ under which $P(S=1) > 0$ and $p$ has an exact uniform distribution given $Z$, i.e.,  
    \begin{equation}
        \label{eq:adjusted_error_control_equality}
        P(p \leq t | Z) = t \text{ for all } t \in [0, 1],  \text{ a.e. under } P(\cdot |S=1), 
    \end{equation}
    then \eqref{eq:adjusted_error_control_cond} and \eqref{eq:adjusted_error_control_marg} hold with equality. 
\end{theorem}

Essentially, \Cref{thm:adjustment} tells us that if we want selective Type I error control as in \eqref{eq:selective_error_control}, then we should reject $H_0$ when $p$ is less than the $\alpha$ quantile of $U \mid Z, S' = 1$ rather than the $\alpha$ quantile of $U$. We show in \Cref{sec:adjustment_proof} that the denominator of \eqref{eq:adjustment} is positive a.e. under $P_{H_0}(\cdot | S=1)$, so $p_{sel}$ is well defined whenever we actually use it to make an inferential statement. 

The right-most panel of \Cref{fig:distributions} depicts what happens when we apply \Cref{thm:adjustment}'s correction to the left panel's p-value (we will derive this correction later). Unlike the null distribution of $p$ given selection (middle panel), the null distribution of $p_{sel}$ given selection again stochastically dominates the uniform distribution. 

\subsection{Characterizing selectively dominant p-values and examples}

\Cref{thm:density} tells us that p-values are selectively dominant precisely when their conditional PDF is non-decreasing under the null. 

\begin{theorem}[Selective dominance and increasing density]
    \label{thm:density}
    If, under the null $H_0$, the conditional PDF of the p-value $p$ given $Z=z$ is non-decreasing on $[0, 1]$ for every $z$, then $p$ is selectively dominant given $Z$ as described in \Cref{def:selective_dominance}. Conversely, if the conditional PDF of $p$ given $Z=z$ is everywhere continuous and not non-decreasing on $[0, 1]$ for a set of $z$ that have positive probability under some distribution in $H_0$, then $p$ is not selectively dominant given $Z$.  
\end{theorem}

In what follows, we give a number of examples of selectively dominant p-values. Our examples include all the common p-values that practitioners use. We recommend that the unfamiliar reader review uniformly most powerful (UMP) and uniformly most powerful unbiased (UMPU) testing \cite[Chapter 3 and Chapter 4]{Lehmann} prior to proceeding.

\begin{example}[Two-sided testing in parametric families]
\label{exm:two-sided}
Consider observing data from a parametric family $P_{\theta}$ and testing the null $H_0 : \theta = \theta_0$. Because the null is a point null, most p-values we construct will have an exact $\text{Unif}([0, 1])$ distribution under the null and are therefore trivially selectively dominant. 
\end{example}

\begin{example}[One-sided testing in monotone likelihood ratio families]
\label{exm:mlr}
Consider observing one-dimensional data from a parametric family $X \sim P_{\theta}$ that admits density $p_{\theta}(x)$ with respect to some base measure $\mu$. We say that $P_{\theta}$ has a monotone likelihood ratio (MLR) in the real valued function $T(x)$ if, the densities $p_{\theta}(x)$ share a common support and, for any $\theta \leq \theta'$, the ratio $p_{\theta'}(x)/p_{\theta}(x)$ is a non-decreasing function of $T(x)$. In this case, the UMP test for the null $H_0: \theta \leq \theta_0$ rejects when $T(X)$ is large. The associated randomized p-value for this test (see \Cref{sec:one_sided_appdx}) is selectively dominant. 
\end{example}

\begin{example}[Testing in in exponential families]
\label{exm:exp_fam}
Suppose we observe data $X \in \R^m$ from an exponential family $P_{\theta}$ parameterized by $\theta \in \R^n$ i.e., under $P_{\theta}$ the data $X$ has density  
\begin{equation*}
    g_{\theta}(x) = \exp( \theta_1 T_1(X) + \dots + \theta_n T_n(X) - \psi(\theta) ) g(x) 
\end{equation*}
with respect to some base measure $\mu$. In both the case of testing the two-sided null $H_0: \theta_1 \neq \theta_{0, i}$ or one-sided null $H_0: \theta_i \leq \theta_{0, i}$, the UMPU test conditions on the nuisance statistics $T_{-i}(X)$. The p-value associated with the UMPU test for $H_0: \theta_1 \neq \theta_{0, i}$ has an exact $\text{Unif}([0, 1])$ distribution conditional on $T_{-i}(X)$, so it is trivially selectively dominant given $Z = T_{-i}(X)$. For testing $H_0: \theta_1 \leq \theta_{0, i}$, we are in the setting of an MLR family once we condition on $T_{-i}(X)$, so \Cref{exm:mlr} implies that the p-value associated with the UMPU test is also selectively dominant given $Z = T_{-i}(X)$.
\end{example}

\begin{example}[Permutation testing]
In a permutation test we observe data $X \in \mathcal{X}$ and compute a test statistic $T(X)$ that, under the null $H_0$, has a distribution that is invariant under a finite group of transformations $G : \mathcal{X} \rightarrow \mathcal{X}$. That is, $T(X) \overset{d}{=}_{H_0} T(g(X))$ for all $ g \in G$. To run the test, we consider a collection of group elements $g_1, g_2, \dots, g_w$ where $g_1 = id$ is fixed to be the identity transformation and $g_2, \dots, g_w$ are either a random sample from $G$ with replacement or a random sample from $G \setminus \{id \}$ without replacement. The test then rejects when $T(X)$ is large compared to the $T(g_j(X))$. In particular, the randomized permutation test from \cite{Hemerik} uses a p-value that always has an exact $\text{Unif}([0, 1])$ distribution under $H_0$ and is therefore is trivially selectively dominant. Details about this p-value can be found in \Cref{sec:perm_test_appdx}
\end{example}

\begin{example}[$F$-tests for regression]
Suppose we observe data $Y \in \R^n$ from a linear model $Y = \X \beta + \epsilon$ parameterized by $\beta \in \R^d$,  where $\X \in \R^{n \times d }$ has full column rank with probability one and is independent of $\epsilon \sim N(0, \sigma^2 I_n)$. We can run an $F$-test to try to reject the null $H_0 : \beta_{q+1} = \dots = \beta_d = 0$ that the last $d-q$ coefficients are zero. The p-value associated with the $F$-test has an exact $\text{Unif}([0, 1])$ distribution given $\X$, so it is trivially selectively dominant given $Z = \X$. We describe this test and its associated p-value more explicitly in \Cref{sec:f_test_appdx}. 
\end{example}

Establishing \Cref{exm:mlr} and \Cref{exm:exp_fam} is non-trivial, and the bulk of \Cref{sec:one_sided_appdx} is spent doing so. The majority of the article focuses on examples from these two non-trivial settings. 

\subsection{Example applications of selective dominance}

Having developed our machinery, we provide a few examples that illustrate how to use it.

As an introductory example, we show how to correct for \Cref{exm:publication_bias}'s publication bias. Using our selective dominance machinery, we can provide a one-line derivation of the p-value adjustment from \cite{Hung2020}. \cite{Hung2020} derive this correction specifically for p-values coming from z- and t-tests, but our machinery applies for all selectively dominant p-values. 

\begin{example}[Correcting for publication bias]
\label{exm:correction}
Suppose we have a selectively dominant p-value $p$ for the null hypothesis $H_0$, and we choose to test $H_0$ only after observing that $p \leq \alpha$. We can apply our framework with $s(x) = I(x \leq \alpha)$. The selective p-value from \eqref{eq:adjustment} is $p/\alpha$, so \Cref{thm:adjustment} tells us that rejecting when $p \leq \alpha^2$ controls selective Type I error:
\begin{equation*}
    P_{H_0}(p \leq \alpha^2 | S= 1) =  P_{H_0}(p/\alpha \leq \alpha | S= 1) \leq \alpha  
\end{equation*}
\end{example}

In \Cref{exm:publication_bias}, if one computes the selective p-value $p_{sel}$ from \eqref{eq:adjustment} carefully, they will find it equals $\min(p/\alpha, 1)$, rather than $p/\alpha$ as we have claimed. However, these two expressions are almost surely equal under $P(\cdot|S=1)$, the conditional probability measure given selection. \Cref{thm:adjustment}'s result therefore applies regardless of which one we choose to use. Throughout the remainder of our examples, we will continue to give expressions for $p_{sel}$ that, at the very least, are equal to \eqref{eq:adjustment} a.e. under $P(\cdot|S=1)$.

Since essentially all the p-values researchers use are selectively dominant, \Cref{exm:correction} gives a simple way for readers to make valid inferences in the presence of publication bias: declare a studies' result significant when the associated p-value is at most $\alpha^2$.

Our rule of thumb of rejecting when $p \leq \alpha^2$ should also deliver valid inferences in the presence of p-hacking. Rather than discarding an experiment after observing a p-value larger than $\alpha$, researches more often tweak their analysis until the p-value crosses the significance threshold. This process, known as p-hacking, is difficult to study theoretically (which is why \cite{Hung2020} do not study it). But it has been empirically well-established that, under the null, p-values resulting from p-hacking have left-skewed distributions, i.e., null p-hacked p-values can be reasonably modeled as having an increasing density on $[0,\alpha]$ \citep{Simonsohn}. The transformed quantity $p/\alpha$ then has a null density that is increasing on $[0, 1]$, so \Cref{thm:density} guarantees that it is a valid p-value. Thus, we again find it safe to reject when $p/\alpha \leq \alpha \iff p \leq \alpha^2$. 

Our second example shows how to use \Cref{thm:adjustment} to perform inference using the ``winning'' p-value. It illustrates how our selective dominance machinery enables us to easily test data-dependent hypotheses. \Cref{exm:winner} uses the fact that, if $p$ is selectively dominant and $Z$ is independent of $p$, then $p$ is selectively dominant given $Z$. Although this is intuitively obvious, we show it carefully in \Cref{sec:sel_dom_independence}. 

\begin{example}[Inference on the winning p-value]
    \label{exm:winner} Suppose we have $n$ independent and selectively dominant p-values $p_i$ for the null hypotheses $H_{0, i}$, and we choose to test only the $j$th null $H_{0, j}$ after observing that $p_j$ is the smallest of the bunch. Applying our framework with $p =p_j$, $Z = p_{-j}$, and the selection function $s(x, z) = I(x < \min_{k} z_k)$, it is straightforward to compute that the selective p-value $p_{sel}$ from \eqref{eq:adjustment} is $p_j/\min_{i \neq j} p_i$, so \Cref{thm:adjustment} tells us that rejecting when $p_j \leq \alpha \min_{i \neq j} p_i$ controls selective Type I error:
    \begin{equation}
        \label{eq:winner_error_control}
        P_{H_{0, j}}(p_j \leq \alpha \min_{i \neq j} p_i \mid S = 1) = P_{H_{0, j}}(p_j/\min_{i \neq j} p_i \leq \alpha  \mid S = 1) \leq \alpha.
    \end{equation} 

    If we let $W$ be the index of the smallest p-value (with ties broken randomly), it is now easy to see that rejecting the data-dependent ``winning'' null $H_{0, W}$ when $p_{(1)} \leq \alpha p_{(2)}$ controls Type I error both conditionally on $W$ and marginally. Conditional error control is immediate: If $H_{0, j}$ is not true, then trivially $P(\text{falsely reject } H_{0, W} \mid W = j) = 0 \leq \alpha$. For the case that $H_{0, j}$ is true, the event $W=j$ is the same event as selecting $p_j$ for inference in \eqref{eq:winner_error_control}, so 
    \begin{align*}
        P(\text{falsely reject } H_{0, W} \mid W = j) &= P(p_{(1)} \leq \alpha p_{(2)} \mid W = j)\\
        &= P(p_j \leq \alpha \min_{i \neq j} p_i \mid W = j)\\
        &\leq \alpha.
    \end{align*}
    Marginal error control then follows from the law of total probability. 
    \begin{align*}
        P(\text{falsely reject } H_{0, W}) &= \sum_{j=1}^n P(W=j)P(\text{falsely reject } H_{0, j} \mid W = j) \\
                                          &\leq \alpha \sum_{j=1}^n P(W=j)\\
                                          &\leq \alpha. 
    \end{align*}
    If the nulls are all true and the $p_i$ are exactly uniform, then the inequalities become equalities and our error control is tight. 
\end{example}

Rejecting the null $H_{0, W}$ when $p_{(1)} \leq \alpha p_{(2)}$ may seem like a strange procedure, but we will see in \Cref{sec:winner} that doing so is central to the conditional inference for winners method that arises from \cite{Fithian2017} and appears often in the contemporary inference on winners literature.

Lastly, we show how our framework also applies to data-carving. Specifically, we consider \cite{Fithian2017}'s variant of the file-drawer problem. \cite{Fithian2017} argue that data-splitting, which involves using one chunk of the data for selection and the remaining independent chunk for inference, is often an inadmissible approach in selective inference problems. In such settings, data-carving, as we describe below, results in strictly more powerful procedures. Although it initially appears that data-carving does not involve selecting a p-value via the mechanism we described earlier, we use a coupling argument to  show that it can be viewed in this way. This both serves to illustrate the breadth of our framework's applicability, as well as provide a new perspective on data carving. 

\begin{example}[Data carving and the file-drawer problem]
    \label{exm:carve}
    In the file-drawer problem, we observe two independent samples $X_1, X_2 \sim N(\mu, 2)$ (e.g., $X_1$ comes from the first half of the data and $X_2$ from the second). We test the null $H_0: \mu \leq 0$, but only when we observe that $X_1 > t$ for some $t \in \R$. \Cref{sec:carve_appdx} provides details for the computations done in this example. 
    
    Data splitting ignores the first observation, which was used for selection, and simply tests the $H_0$ using the p-value $p_{split} = 1 - \Phi(X_2/\sqrt{2}) $. Intuitively, because this p-value is independent of the selection process, we should maintain Type I error control without any correction. Applying our framework with $p = p_{split}$, $Z = X_1$, and the selection function $s(x, z) = I(z > t)$, we indeed find that the selective p-value is the same as $p_{split}$. 
    
   Data-carving attempts to use the more powerful p-value $p_{full} = 1 - \Phi( (X_1 + X_2)/2 )$. This p-value leverages information from both samples, despite one of them being used for selection. How can we apply our framework to this data-carving problem? In how we have stated the problem, it is \underline{not} the case that we observe $p_{full}$ and then decide whether or not to use it for inference. Instead, we decide based on $X_1$, and unlike in data-splitting, $p_{full}$ is not a valid p-value given $X_1$. We can compute, however, the probability that selection happened given that $p_{full}$ took a particular value:
    \begin{equation*}
        P(X_1 > t | p_{full} = x) = 1 - \Phi(t - \Phi^{-1}(1-x)). 
    \end{equation*}
    We may as well imagine that we observed that $p_{full} = x$, and then decided to use it to test the null $H_0$ with probability $1 - \Phi(t - \Phi^{-1}(1-x))$. Although this is not what happens in the original problem (in our new characterization, we may test $H_0$ even when $X_1 \leq 3$), the conditional distribution of $p_{full}$ given selection is the same in both cases. Hence, we can apply our framework with $p = p_{full}$ and $s(x) = 1 - \Phi(t - \Phi^{-1}(1-x))$. Theorem 9 of \cite{Fithian2017} tells us that \Cref{thm:adjustment}'s resulting selective p-value,
    \begin{equation*}
        p_{carve} = \frac{\int_0^{p_{full}}  1 - \Phi(t - \Phi^{-1}(1-x)) dx}{\int_0^1  1 - \Phi(t - \Phi^{-1}(1-x)) dx} = \frac{\int_0^{p_{full}}  1 - \Phi(t - \Phi^{-1}(1-x))dx }{1 - \Phi(t/\sqrt{2})}, 
    \end{equation*}
    will result in strictly more rejections than $p_{split}$, while still maintaining selective Type I error control. 
\end{example}

Crucially, in \Cref{exm:carve}, the conditional distribution of the random variable $X_1$ used for selection given the p-value $p_{full}$ did not depend on the unknown parameter $\mu$. Hence, the selection function $s(x)$ had no dependence on $\mu$, and we were able to construct our selective p-value without any issues. This did not happen by accident, and is actually a consequence of more general and interesting fact regarding the relationship between between data splitting, data carving, data fission \citep{Leiner}, and data thinning \citep{Dharamshi, Neufeld}.

In the most basic version of data fission, we add and subtract independent normal noise $Z \sim N(0, 1)$ to a normal sample $X \sim N(\mu, 1)$ to get two independent samples $X_1, X_2 \sim N(\mu, 2)$ centered at the same mean. This is meant to mimic data splitting: the first sample can be used for selection and the second for inference. Data thinning generalizes this idea by considering a random vector $X$ from a parametric family and adding noise to make $k$ new random vectors $X_1, \dots, X_k$ that (1) are independent and (2) can be used to recover $X$ via a deterministic function $X = T(X_1, \dots, X_k)$. Vanilla data thinning involves using some of the $X_i$ to perform selection and then the rest to do inference. Data carving, however, suggests using a p-value $p = p(T(X_1, \dots, X_n))$ that leverages all of the data, despite some of the $X_i$ having been used for selection. Because the noise we add to $X$ to get the $X_i$ has no dependence on the unknown parameter, the joint distribution of the $X_i$ given $p$ also has no dependence on the unknown parameter. Therefore, contrary to the what \cite{Leiner} suggest, the selection function $s(x)$ is always \underline{known}, and we can always apply our framework to data carve and get more power. If the selection process is highly complicated, it is true that $s(x)$ may be very difficult to compute, but in theory it is always accessible to us via extensive simulations. 

Our framework also applies to regression problems, including \cite{Lee2016}'s foundational problem of doing inference on LASSO selected regression coefficients. This example is a bit long, so we have deferred it to \Cref{sec:lasso_appdx}.  

\Cref{exm:correction,exm:winner,exm:carve} all share a common theme. In all three examples, the practitioner cheats. They peek at the p-value and, to varying degrees, they only test the null when the p-value looks promising. The purpose of selective procedures is to adjust the p-value in a way that accounts for this cheating. The harsher the cheating is, the more this adjustment inflates the original p-value.

\section{Inference on winners}
\label{sec:winner}

In this section we use our framework to study the inference on winners problem. Along with discussing the implications of \Cref{exm:winner}, we also show how hybrid inference \citep{Andrews2023}, which offers a solution to the exploding interval problem, arises naturally in our framework. For both the conditional and hybrid approaches, our discussion results in novel interpretations, generalizations, and methods. To be concrete, we will imagine performing inference at level $\alpha = 0.1$ throughout the section.

For now, we focus on the independent data setting. One can apply the tools we develop in the next section, however, to perform inference on winners when data is generated from a multi-parameter exponential family, which encompasses many correlated data settings.

\subsection{Conditional inference}

\begin{figure}[]
    \centering
    \scalebox{1}{
    \hspace{-0.03\textwidth}
    \begin{minipage}{0.33\textwidth}
        \centering
        \includegraphics[width=\textwidth]{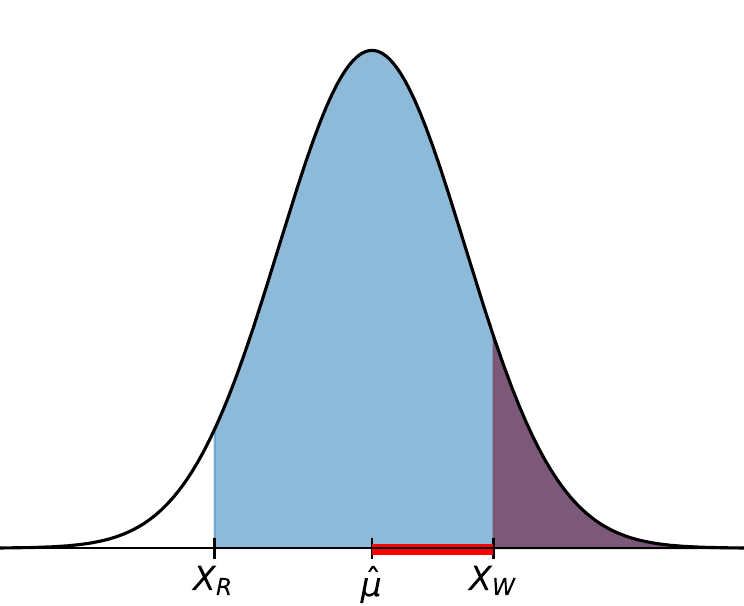}
    \end{minipage}\hfill
    \hspace{0.02\textwidth}
    \begin{minipage}{0.33\textwidth}
        \centering
        \includegraphics[width=\textwidth]{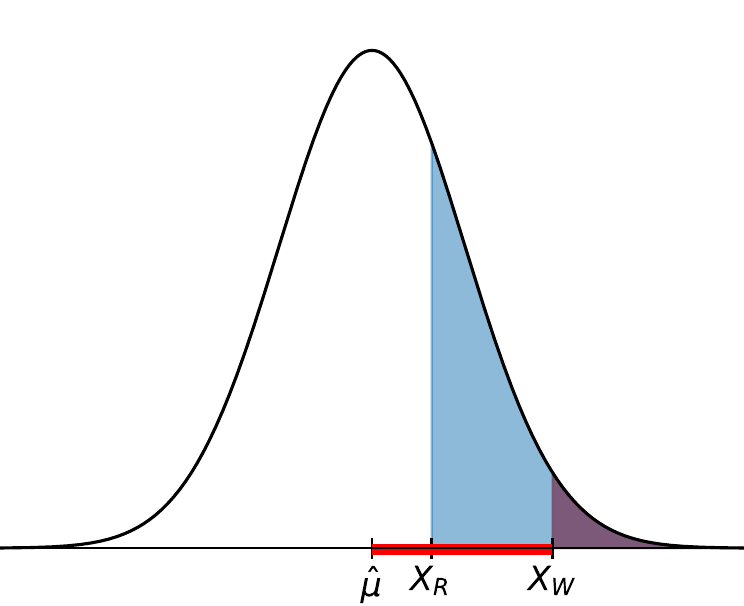}
    \end{minipage}\hfill
    \hspace{0.02\textwidth}
    \begin{minipage}{0.33\textwidth}
        \centering
        \includegraphics[width=\textwidth]{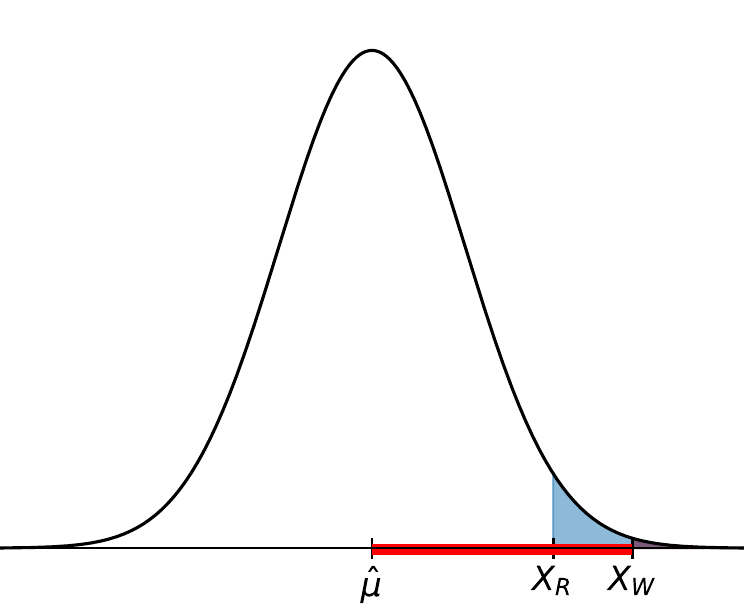}
    \end{minipage}
    }
    \caption{We plot the level $\alpha=0.1$ conditional LCB $\hat{\mu}$ for different gaps between the winning value $X_W$ and the runner up value $X_R$ and highlight the distance between $\hat{\mu}$ and $X_W$ in red. The LCB $\hat{\mu}$ is chosen exactly so that the tail probability $P(N(\hat{\mu}, 1) > X_R)$, shaded in blue, is $1/\alpha=10$ times the tail probability $P(N(\hat{\mu}, 1) > X_W)$, shaded in red (the overlap appears purple). As $X_W$ and $X_R$ get closer, we need to take $\hat{\mu}$ further back for this condition to be satisfied.}
    \label{fig:tail_prob}
\end{figure}

In this sub-section we discuss \cite{Fithian2017}'s conditional approach for performing inference on winners. This conditional approach turns out to be highly related to the testing procedure that we derived in \Cref{exm:winner}. 

\begin{corollary}[Testing the winner]
    \label{cor:cond}
    Suppose that $p_i$ are $n$ independent and selectively dominant p-values for the nulls $H_{0, i}$, and let $W$ be the index of the smallest p-value (with ties broken randomly). Rejecting $H_{0, W}$ when $p_{(1)} \leq \alpha p_{(2)}$ controls Type I error at level $\alpha$ conditionally on $W$, and therefore also marginally. 
\end{corollary}

Unlike Sidak's simultaneous approach, which rejects the winning null when the smallest p-value is small in absolute terms, the conditional approach rejects the winning null when the smallest p-value is small relative to the second smallest p-value. This procedure is strange, but fairly easy to interpret: we reject the winning null when the most extreme observation is $1/\alpha = 10$ times more extreme under its null than the second most extreme observation is under its null. This can be quite a stringent requirement! 

Once written in terms of p-values, its easy to mathematically see the merits and pitfalls of the conditional approach. If all the p-values except the smallest provide essentially no evidence against the null,  then $p_{(2)} \approx 1$ and we reject when $p_{(1)} < \alpha$, the same p-value cutoff as a one-dimensional problem. On the other hand, if even just one other p-value provides a similar amount of evidence against the null as the smallest p-value, then then $p_{(2)} \approx p_{(1)}$ and we will never reject because $p_{(1)} \not < \alpha p_{(1)} \approx \alpha p_{(2)}$.

The p-value viewpoint also makes it clear that the conditional approach can only outperform the simultaneous approach when some null p-values are conservative (i.e., they have super-uniform distributions). We give a heuristic argument here, and a more formal statement in \Cref{sec:beta_dist_appdx}. Suppose one of our p-values $p_1$ is a very strong signal (so it is very small with high probability) but the remaining p-values $p_2, \dots, p_n$ are null p-values that are uniform (i.e., they are not conservative). Our conditional procedure will reject when our smallest p-value, likely $p_1$, is less than $\alpha$ times the smallest of $p_2, \dots, p_n$. The minimum of these $n-1$ uniform p-values is $1/n$ on average. Hence, roughly speaking, the conditional approach also rejects when $p_{(1)}$ is less than $\alpha/n$, which is essentially the same as Sidak's simultaneous approach for moderately sized $n$ or small $\alpha$.   

\subsubsection{Confidence regions for the winner}

In parametric settings, we can invert \Cref{cor:cond}'s test to get selective confidence regions for the winning parameter. Consider observing independent data $X_i \sim P_{\theta_i}$ from an MLR family $P_{\theta}$ parametrized by $\theta \in \R$. Let $p_i^{\theta_0}$ (which is a function of solely $X_i$) be the UMP p-value for testing the null $H^{\theta_0}_{0, i} : \theta_i \leq \theta_0$. Details regarding these p-values can be found in \Cref{sec:one_sided_mlr_appdx}. We can define the winner $W = \argmin_{j \in [n]} p_j^{\theta_0}$ to be the index of the smallest and therefore most promising p-value. This winning index will be the same irrespective of $\theta_0$  \footnote{If we use the same auxiliary randomness to compute $p_i^{\theta_0}$ for every $\theta_0$, then one index will result in the smallest p-value for every $\theta_0$ and $W$. The way we construct our UMP p-values in \Cref{sec:one_sided_mlr_appdx} ensures that this is the case. }. By inverting \Cref{cor:cond}'s test, we get an LCB

\begin{equation}
    \label{eq:winning_lcb}
    \{\theta_0 : p^{\theta_0}_{(1)}/p^{\theta_0}_{(2)} > \alpha  \}
\end{equation}
for the winning parameter $\theta_W$ that holds conditionally on $W$ with probability exactly $1-\alpha$:
\begin{equation*}
    P( \theta_W \in\{\theta_0 : p^{\theta_0}_{(1)}/p^{\theta_0}_{(2)} > \alpha  \} |W ) = 1-\alpha.
\end{equation*}
The fact that the confidence region \eqref{eq:winning_lcb} is actually an LCB is a consequence of the selective p-value $p_{(1)}^{\theta_0}/p^{\theta^0}_{(2)}$ being monotone non-decreasing in null parameter $\theta_0$. \Cref{sec:one_sided_monotone_appdx} provides general conditions (which will hold for essentially every selective problem) under which selective p-values like $p_{(1)}^{\theta_0}/p^{\theta^0}_{(2)}$ are monotone in the null parameter. We show these conditions apply here, and also argue that \eqref{eq:winning_lcb} has exact $1-\alpha$ coverage in \Cref{sec:cond_appdx}. We also show in \Cref{sec:cond_appdx} how to invert \Cref{cor:cond}'s test to get a confidence interval (CI) instead of LCB, and that both our CI and LCB match \cite{Fithian2017}'s approach in the Gaussian case. 

Writing the conditional LCB terms of p-values helps us better understand its behavior. In particular, recalling our motivating Gaussian example \eqref{eq:motivating_lcb}, we learn that \Cref{fig:winner}'s LCB stretches back exactly to the $\hat{\mu}$ under which it is $1/\alpha = 10$ times less likely to see something as extreme as the winner than something as extreme as the runner-up. Noting this fact, \Cref{fig:tail_prob} illustrates why the LCB \eqref{eq:motivating_lcb} diverges to $-\infty$ as the winner and runner-up get closer. If the winner and runner-up are very close, we need the LCB $\hat{\mu}$ to be very far back for the winner to be ten times more extreme than the runner-up. Thanks to the decay of the Gaussian tail, however, we can always find such a mean if we go far back enough. 

For non-Gaussian data, the amount the conditional LCB \eqref{eq:winning_lcb} stretches back depends on the right tail decay of $P_{\theta}$. The faster the tail decays as $\theta$ decreases, the larger the first $\hat{\theta}$ for which the winner is $1/\alpha = 10$ times as extreme as the runner-up. Seeing as the Gaussian distribution, which has a rapidly decaying tail, still results in exploding lower bounds, we should expect that the distance from the winning observation to the lower bound can often be quite large in many settings.  

As an example, consider observing independent exponential random variables $X_i \sim \text{Exp}(\lambda_i)$. Crucially, the parameter space $\lambda \in (0, \infty)$ is bounded below. The exponential distribution has an MLR in $T(x) = 1/x$, so the UMP test for $H^{\lambda_0}_{0, i}: \lambda_i \leq \lambda_0$ uses a p-value $p^{\lambda_0}_i$ that rejects when $X_i$ is small. It turns out that 
\begin{equation*}
    \lim_{\lambda^0 \downarrow 0} p^{\lambda_0}_{(1)}/p^{\lambda_0}_{(2)} = X_{(1)}/X_{(2)},
\end{equation*}   
so the conditional LCB for the winning parameter $\lambda_W$,
\begin{equation}
\label{eq:exp_winning_lcb}
 \{\lambda_0 : p^{\lambda_0}_{(1)}/p^{\lambda_0}_{(2)}  > \alpha\},
\end{equation}
is vacuous whenever $X_{(1)}/X_{(2)} > \alpha$, i.e., with positive probability the confidence region \eqref{eq:winning_lcb} spans the whole parameter space $(0, \infty)$. A careful derivation of this test and result can be found in \Cref{sec:exponential_winner_appdx}. 

The failure of the exponential conditional LCB \eqref{eq:exp_winning_lcb} manifests in real data examples. On a dataset of car engine failure times \citep{Molotaliev}, we find that the conditional LCB \eqref{eq:exp_winning_lcb} is always vacuous. The dataset has the failure times of one-hundred car engines, which we model as independent exponential random variables. Over many sub-samples of just $n=2$ failure times, the LCB \eqref{eq:winning_lcb}, which does inference on the worse of the two engines, always gives a vacuous lower bound of zero. In contrast, the simultaneous approach always gives a non-vacuous lower bound. \Cref{fig:car_engine} depicts the results. The result is concerning. Despite conditional LCB having exact $1-\alpha=0.9$ coverage, our empirical coverage is one (the vacuous LCB must always cover the parameter). This suggests that the exponential distribution is likely not an appropriate model for this dataset (the failure times are not spread out enough), even though it is often a natural choice for modeling failure times. 
\begin{figure}
    \centering
    \includegraphics[width=0.5\textwidth]{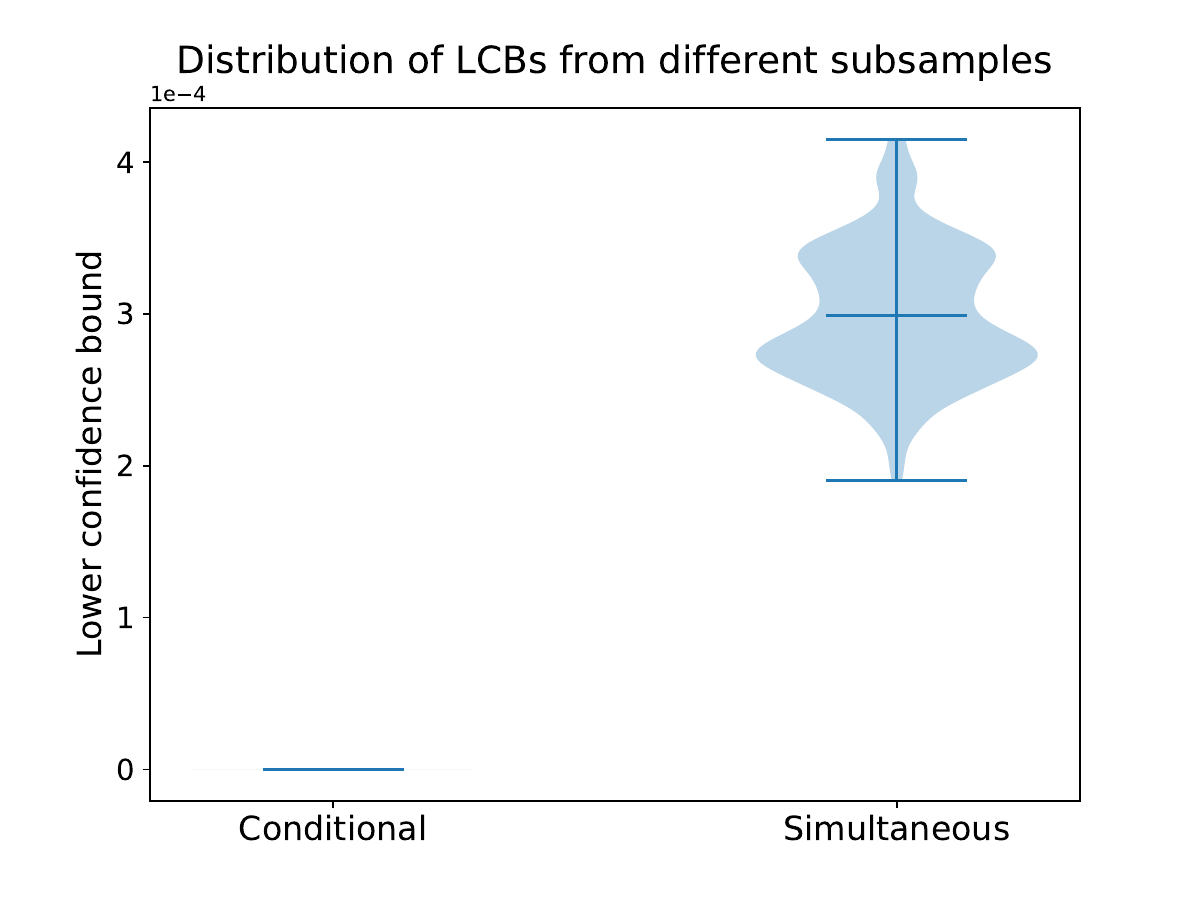} 
    \caption{Over $B=1000$ different sub-samples of $n=2$ failure times from the dataset \cite{Molotaliev}, the distribution of the conditional and simultaneous LCB for the ``winning'' parameter. The conditional LCB is always vacuous.}
    \label{fig:car_engine}
\end{figure}

\subsubsection{More discoveries via the closure principle}

Once written in terms of p-values, it is natural to treat \Cref{cor:cond}'s test as a test of the global null and try to close it (as in \cite{Marcus}). Closing a global null test precludes us from making confidence regions, but it allows us to make more individual discoveries. Closed global null testing procedures are often computationally intractable to implement, so it is interesting that \Cref{cor:cond}'s global null test admits a tractable closure.  

\begin{corollary}[Closed testing for winners]
    \label{cor:cond_closed}
    Suppose that $p_i$ are $n$ independent and selectively dominant p-values for the nulls $H_{0, i}$. Let $H_{0, (j)}$ denote the null corresponding to the $j$th smallest p-value (with ties broken randomly) and define $p_{(n+1)} = 1$.  Rejecting the null hypotheses $H_{0, (k)}$ for which $p_{(j)} \leq \alpha p_{(j+1)} $ for all $j \leq k$ controls the family-wise error rate at level $\alpha$. 
\end{corollary}

As is often the case for closed procedures, \Cref{cor:cond_closed} procedure is best understood sequentially. We reject $H_{0, (1)}$ when $p_{(1)} \leq \alpha p_{(2)}$. Then, if we rejected $H_{0, (1)}$, we reject $H_{0, (2)}$ when $p_{(2)} \leq \alpha p_{(3)}$, so on and so forth until we fail to reject. 

\subsection{Hybrid inference}

Hybrid inference, originally proposed by \cite{Andrews2023}, is an inference on winners procedure that attempts to balance the benefits of the simultaneous and conditional approaches. It is a very elegant idea, but it currently only applies to Gaussian data and can be difficult to parse and implement. Using our selective dominance framework, we give a simpler exposition of hybrid inference that enables its application in more general settings, provided that the data is independent. As a bonus, our new procedure is very easy to understand and implement.

\Cref{cor:hyb} presents our hybrid testing procedure. We give the sketch of a proof and defer a detailed proof to \Cref{sec:hyb_proof_appdx}. 

\begin{corollary}[Hybrid test for the winner]
    \label{cor:hyb}
    Suppose that $p_i$ are $n$ independent and selectively dominant p-values for the nulls $H_{0, i}$, and let $W$ be the index of the smallest p-value (with ties broken randomly). Fix some $\beta \leq \alpha$ and define $\beta_n  = 1 - (1-\beta)^{1/n}$. Rejecting $H_{0, W}$ when 
    \begin{equation}
        \label{eq:hybrid_cutoff_thm}
        p_{(1)} \leq \frac{\alpha-\beta}{1-\beta}p_{(2)}  + \left(1 - \frac{\alpha - \beta}{1 - \beta}\right) \beta_n
    \end{equation} 
    controls Type I error at level $\alpha$. 
\end{corollary}

\begin{proof}[Proof sketch]    
    Let $B$ be the event that the smallest p-value comes from a null and is at most $\beta_n$. We know from Sidak's procedure that $P(B) \leq \beta$. Hence, on the complementary event $B^c$, which has probability $\geq 1-\beta$, it suffices to ensure that we fail to falsely reject $H_{0, W}$ with probability at least $(1-\alpha)/(1-\beta)$. Supposing $H_{0, j}$ is true, imagine testing $H_{0, j}$ using $p_j$ only when $B^c$ happens and $W=j$. This is exactly like selecting $p_j$ to use for inference when it is between $\beta_n$ and $\max_{i \neq j} p_i $. For this selection, \Cref{thm:adjustment}'s selective p-value is given by $(p_j - \beta_n)/(\max_{i \neq j} p_i  - \beta_n)$. Thus, we can ensure that we fail to reject $H_{0, j}$ when $B^c$ and $W=j$ happen with probability at least $(1-\alpha)/(1-\beta)$ if we fail to reject whenever
    \begin{equation*}
        \frac{p_j - \beta_n}{\max_{i \neq j} p_i  - \beta_n} > 1 - \frac{1-\alpha}{1-\beta} \iff p_j > \frac{\alpha-\beta}{1-\beta} \max_{i \neq j} p_i + \left(1 - \frac{\alpha-\beta}{1-\beta}\right)\beta_n.
    \end{equation*}
    Our hybrid inference procedure fails to reject in this case.  
\end{proof}

Written in terms of p-values, it is easy to see how the hybrid approach balances the benefits of the simultaneous and conditional approaches. It will reject both when the smallest p-value is small in absolute terms or when it small relative to the second smallest p-value. When the other p-values provide essentially no evidence against the null (i.e., $p_{(2)} \approx 1$), hybrid rejects when $p_{(1)}$ is at most $(\alpha-\beta)/(1-\beta)$, a cutoff that has no dependence on the problem dimension $n$. Thus, hybrid performs at least on par with the conditional procedure run at level $(1-\alpha)/(1-\beta)$. On the other hand, even if some other p-value provides as much evidence against the null as the smallest, hybrid still rejects whenever the level $\beta$ simultaneous approach does. This is because when $p_{(1)} \leq \beta_n$, the hybrid cutoff is a mixture of two things that are at least $p_{(1)}$, so we reject. 

The parameter $\beta$ allows hybrid inference to interpolate between the simultaneous and conditional approaches. When we set $\beta = 0$ then the hybrid cutoff \eqref{eq:hybrid_cutoff_thm} becomes $\alpha p_{(2)}$ and we recover the conditional method, and if we set $\beta=\alpha$ it becomes $\alpha_n = 1 - (1-\alpha)^{1/n}$ and we recover the simultaneous method. 

\subsubsection{Confidence regions}

In parametric settings, we can get hybrid confidence regions for the winning parameter by inverting \Cref{cor:hyb}'s test. Again suppose we have independent data $X_i \sim P_{\theta_i}$ from an MLR family $P_{\theta}$ parametrized by $\theta \in \R$, and let $p_i^{\theta_0}$ be the UMP p-value for testing the null $H^{\theta_0}_{0, i} : \theta_i \leq \theta_0$. By inverting \Cref{cor:hyb}'s test we get a hybrid LCB for the winning parameter $\theta_W$:
\begin{equation}
    \label{eq:hyb_lcb}
    \left\{ \theta_0 \in \R:  p^{\theta_0}_{(1)} > \frac{\alpha-\beta}{1-\beta}p^{\theta_0}_{(2)}  + \left(1 - \frac{\alpha - \beta}{1 - \beta}\right) \beta_n \right\}.
\end{equation}
We argue in \Cref{sec:hybrid_appdx} that the confidence region \eqref{eq:hyb_lcb} indeed gives a LCB. We also show in \Cref{sec:hybrid_appdx} how to invert \Cref{cor:hyb}'s test to get a CI (rather than an LCB), and that both our CI and LCB match the original construction from \cite{Andrews2023} in the Gaussian case. 

\subsubsection{Comparison to the union bound}

\begin{figure}[]
    \centering
    \includegraphics[width=0.45\textwidth]{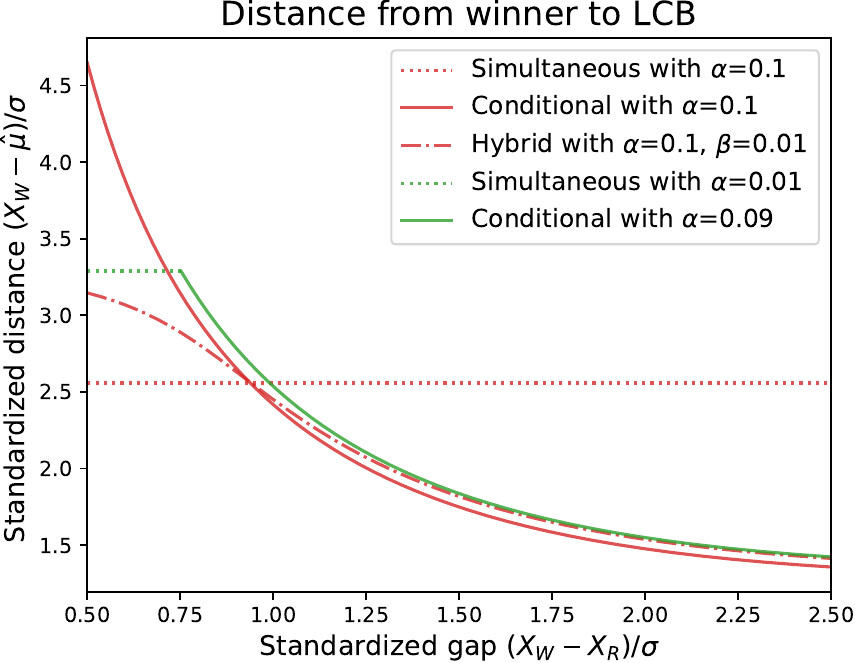}
    \caption{For the $n = 20$ dimensional Gaussian problem $X_i \sim N(\mu_i, \sigma^2)$ with largest observation $X_W$ and second largest observation $X_R$, the standardized distance $(X_W - \hat{\mu})/\sigma$ from $X_W$ to the level $\alpha = 0.1$ hybrid, conditional, and simultaneous LCB $\hat{\mu}$ as a function of the standardized gap $(X_W - X_R)/\sigma$ between the winner and runner-up. For hybrid we take $\beta = 0.01$. The larger of the level $\alpha = 0.09$ conditional LCB and level $\alpha = 0.01$ simultaneous LCB are shown in green (i.e., the union
    bound LCB).}
    \label{fig:hybrid}
\end{figure}

Another way to balance the benefits of the conditional and simultaneous approaches is to apply a union bound. Naively, we can reject the winning null whenever the level $\beta$ simultaneous approach rejects or the level $\alpha -\beta$ conditional approach rejects, i.e., whenever
\begin{equation}
    \label{eq:union_bound_cutoff}
    p_{(1)} \leq \max\{(\alpha - \beta) p_{(2)}, \beta_n\}.
\end{equation}
The union bound harshly switches between the simultaneous and conditional approaches, whereas the hybrid approach smoothly interpolates between them. This is illustrated in \Cref{fig:hybrid}, which compares the LCBs resulting from the hybrid versus union bound approaches in the $n=20$ dimensional Gaussian problem. 

Written in terms of p-values, we easily see that the hybrid approach dominates the union bound approach, which is affirmed by \Cref{fig:hybrid}. Both methods reject when $p_{(1)} \leq \beta_n$. When $p_{(1)} > \beta_n$, it is quick to verify that the hybrid cutoff \eqref{eq:hybrid_cutoff_thm} will be strictly larger than the union bound cutoff \eqref{eq:union_bound_cutoff}, meaning hybrid will reject whenever the union bound does and more \footnote{the authors \cite{Andrews2023} only point out that hybrid dominates the level $\beta$ classical approach, which is weaker than our statement}. 

Practically speaking, however, hybrid inference does not result in much improvement over the union bound, especially as it pertains to making discoveries. This is already somewhat evident in \Cref{fig:hybrid}, where we see that the hybrid LCB, although always larger than the union bound LCB, is still always very close to it. As the variance $\sigma^2$ gets large, the absolute difference $\hat{\mu}_{hyb} - \hat{\mu}_{union} $ between the hybrid and union bound LCBs grows with $\sigma$, but the relative difference $({\mu}_{hyb} - \hat{\mu}_{union})/\sigma$   (in units of standard deviation) remains the same (see \Cref{sec:hybrid_gap_appdx}). Accordingly, even when the hybrid cutoff \eqref{eq:hybrid_cutoff_thm} is larger than that of the union bound \eqref{eq:union_bound_cutoff}, it is provably not much larger. We detail why in \Cref{sec:hybrid_sim_appdx}, where we also run a number of simulations comparing the power of the hybrid and union bound approaches. In our simulations, we are unable to find a setting where the hybrid approach results in a appreciable power gain.  

Overall, we suggest viewing hybrid inference as a procedure that squeezes the remaining power out of the union bound approach. As it is not computationally more expensive and our p-value viewpoint makes it equally easy to implement, it is always worth using in place of the union bound. 

\subsubsection{Applying the closure principal}

As was true in the conditional case, treating \Cref{cor:hyb}'s test as a global null test and closing it allows us to make more discoveries. As we allow $\beta$ to range from $0$ to $\alpha$, this closed procedure interpolates between \Cref{cor:cond_closed}'s closed procedure and the Holm-Sidak procedure, which is the closure of Sidak's global null test. 

\begin{corollary}[Closed hybrid testing for winners]
    \label{cor:hyb_closed}
    Suppose that $p_i$ are $n$ independent and selectively dominant p-values for the nulls $H_{0, i}$. Let $H_{0, (j)}$ denote the null corresponding to the $j$th smallest p-value (with ties broken randomly) and define $p_{(n+1)} = 1$. Fixing some $\beta \leq \alpha$, rejecting the null hypotheses $H_{0, (k)}$ for which
    \begin{equation*}
        p_{(j)} \leq \frac{\alpha - \beta}{1-\beta} p_{(j+1)} + \left(1 - \frac{\alpha - \beta}{1-\beta} \right) \beta_{n - j + 1}   
    \end{equation*}
    for every $j \leq k$ controls FWER error at level $\alpha$. 
\end{corollary}

This closed procedure is also best understood sequentially. We reject $H_{0, (1)}$ when $p_{(1)} \leq \frac{\alpha - \beta}{1-\beta} p_{(2)} + (1 - \frac{\alpha - \beta}{1-\beta}) \beta_n$. Then, if we rejected $H_{0, (1)}$, we reject $H_{0, (2)}$ when $p_{(2)} \leq  \frac{\alpha - \beta}{1-\beta} p_{(3)} + (1 - \frac{\alpha -\beta }{1-\beta}) \beta_{n-1}$, so on and so forth until we fail to reject.

\section{Rank verification in exponential families}
\label{sec:rank_verification}

In this section we consider the problem of verifying that that the winning parameter is actually larger than the other parameters, i.e., rather than doing inference on the winning parameter, we do inference on the gap between the winning and remaining parameters. Mainly, we illustrate how \cite{Hung2019}'s rank verification procedure fits nicely in our framework. We show, however, that \cite{Hung2019} do not appropriately handle cases where there can be ties for the winner. In contrast, our selective dominance framework naturally handles these cases properly.

Overall, the section serves to illustrate how our selective dominance machinery provides a straightforward way to correctly design intricate and counter-intuitive selective procedures. For examples of how one may apply these methods, we refer the reader to the original article \cite{Hung2019}, which contains many compelling and important examples. 

\subsection{Warm-up: rank verification and Type III error}

To motivate the rank verification problem and shed some light on its relationship with selective dominance, we first consider a seemingly unrelated classical statistical question about Type III errors. 

A researcher wants to test if the unknown means of two univariate Gaussian samples, $X_1 \sim N(\mu_1, 1/\sqrt{2})$ and $X_2 \sim N(\mu_2, 1/\sqrt{2})$, are different. They end up rejecting the null hypothesis $H_0: \mu_1 = \mu_2$ because they observe that the two-sided p-value $2(1 - \Phi(|X_1 - X_2|))$ is below $\alpha$. After rejecting, they note $X_1 > X_2$, and claim ``not only are the two means are different, but they must be different because $\mu_1$ is bigger than $\mu_2$''. The researcher, however, only rejected the null that the means are equal. Can they make a claim about the direction of inequality? This is a question of Type III error, and we can use our selective dominance framework to show that the researcher's claim is actually statistically valid. 

Based on the claim, it seems that what the researcher really wants to do is test the one-sided null $H_{0, 12} : \mu_1 \leq \mu_2$ whenever they observe that $X_1 > X_2$, and test the complementary one-sided null $H_{0, 21} : \mu_2 \leq \mu_1$ whenever they observe that $X_2 < X_1$. To test the null $H_{0, ij} : \mu_i \leq \mu_j$ we normally use the UMP p-value $p_{ij} = 1 - \Phi(X_i - X_j)$. In the researcher's case, however, they only select this p-value to use for inference when they observe that $X_i > X_j$, or equivalently when $p_{ij} < 1/2$. Since the $p_{ij}$ are selectively dominant (by \Cref{exm:mlr}), \Cref{thm:adjustment} tells us that the researcher should correct for this selection and reject $H_{0, ij}$ only when $2p_{ij} \leq \alpha \iff p_{ij} \leq \alpha/2$. Letting $W$ be the index of the winner and $R$ of the runner-up, the final procedure is to reject $H_{0, WR}$ when $p_{WR} \leq \alpha/2$. 

The procedure described above is a rank verification procedure. It affirms not just that the means are different, but that the mean of the winning observation is the strictly larger of the two (i.e., of higher rank). This procedure, which rejects when the smaller of the two one-sided p-values is at most $\alpha/2$, is identical to the procedure that rejects when our earlier two-sided p-value is at most $\alpha$. Hence, for reasons likely unbeknownst to them, the researcher's original claim is indeed statistically valid. We walk through deriving the above procedure more carefully (with as much detail as we did in \Cref{exm:winner}) in \Cref{sec:rank_verification_warm_up_appdx}.

It turns out that, for the $n$-dimensional Gaussian problem (we have considered the 2-dimensional problem up to this point), it has long been known that we can confirm the winning mean as the strict largest via a level $\alpha/2$ one-sided test comparing the winner and runner-up \citep{Gutmann}. \cite{Hung2019} claim further that running the test at level $n/(n-1) \cdot \alpha/2$ maintains marginal Type I error control. This claim, however, is not true. By taking $\mu_1 = \mu_2$ and $\mu_3 = \dots = \mu_n = -\infty$, one can verify that the marginal Type I error of the level $\alpha/2$ one-sided test is exactly $\alpha$. Hence, the error cannot be inflated any further.  

Interestingly, \cite{Hung2019}'s claim becomes true if, rather than wanting to verify that the winning mean is strictly bigger than the other means, we want to verify that it is just at \underline{least as} big. Again let us focus on the 2-dimensional case. If, instead of testing the null $H_{0, WR}: \mu_W \leq \mu_R$, we test the null $H_{0, WR} :\mu_W < \mu_R$, then we can indeed run the one-sided test comparing the winner and runner-up at level $\alpha$ instead of level $\alpha/2$. Supposing (without loss of generality) that $\mu_1 \geq \mu_2$, the proof is straightforward:
\begin{align*}
    &P_{\mu_1, \mu_2}(\text{falsely reject } H_{0, WR}) &\\
    & =P_{\mu_1, \mu_2}(p_{21} < \alpha | W = 2)P(W=2)I(\mu_2 < \mu_1) & \text{(no false rejection when $W=1$ or $\mu_2 =\mu_1$)}\\
    &\leq P_{\mu_1, \mu_2}(2p_{21} < 2\alpha | p_{21} \leq 1/2)\cdot \frac{1}{2} \cdot I(\mu_2 < \mu_1) & \text{($X_1$ wins w.p. at least $1/2$ )}\\
    &\leq 2 \alpha \cdot \frac{1}{2} \cdot I(\mu_2 < \mu_1) \leq \alpha. & \text{(selective dominance)}
\end{align*}

As rejection probabilities are typically continuous functions of the parameters, it is counter-intuitive that excluding the boundary of the null makes a tangible difference. For data-dependent hypotheses, however, the false rejection region can be a highly discontinuous function of the parameters, which elicits this counter-intuitive behavior. For example, under the data-dependent null $H_{0, WR} :\mu_W < \mu_R$, the false rejection region $\emptyset$ is empty when $\mu_1 = \mu_2$, but equal to $\{X_2 - X_1 > z_{1-\alpha}\}$ whenever $\mu_1$ is even just $\epsilon > 0$ larger than $\mu_2$. Under the null, $H_{0, WR} :\mu_W \leq \mu_R$ the false rejection region $\{|X_1 - X_2| > z_{1-\alpha}\}$ is as large as possible when $\mu_1 = \mu_2$, which is exactly what prevents us from inflating the level of the test. For both the nulls $H_{0, WR} :\mu_W < \mu_R$ and $H_{0, WR} :\mu_W \leq \mu_R$, \Cref{fig:warm_up_error_control} fixes $\mu_2=0$ and plots the Type I error of the level $\alpha$ one-sided test comparing the winner and runner up for different values of $\mu_1$. The discontinuity at $\mu_1=0$ illustrates why we have Type I error control for $H_{0, WR} :\mu_W < \mu_R$ but not for $H_{0, WR} :\mu_W \leq \mu_R$. 

A nice implication of our discussion is the following surprising fact: if we want to to verify that the winning mean amongst $X_1 \sim N(\mu_1, 1/\sqrt{2})$ and $X_2 \sim N(\mu_2, 1/\sqrt{2})$ is \underline{at least} as large as the other mean, we can run a one-sided test comparing the winner to the runner up at level $\alpha$, i.e., with \underline{no correction}, and still maintain Type I error control. 

\begin{figure}[]
    \centering
    \scalebox{1}{
    \begin{minipage}{0.4\textwidth}
        \centering
        \includegraphics[width=\textwidth]{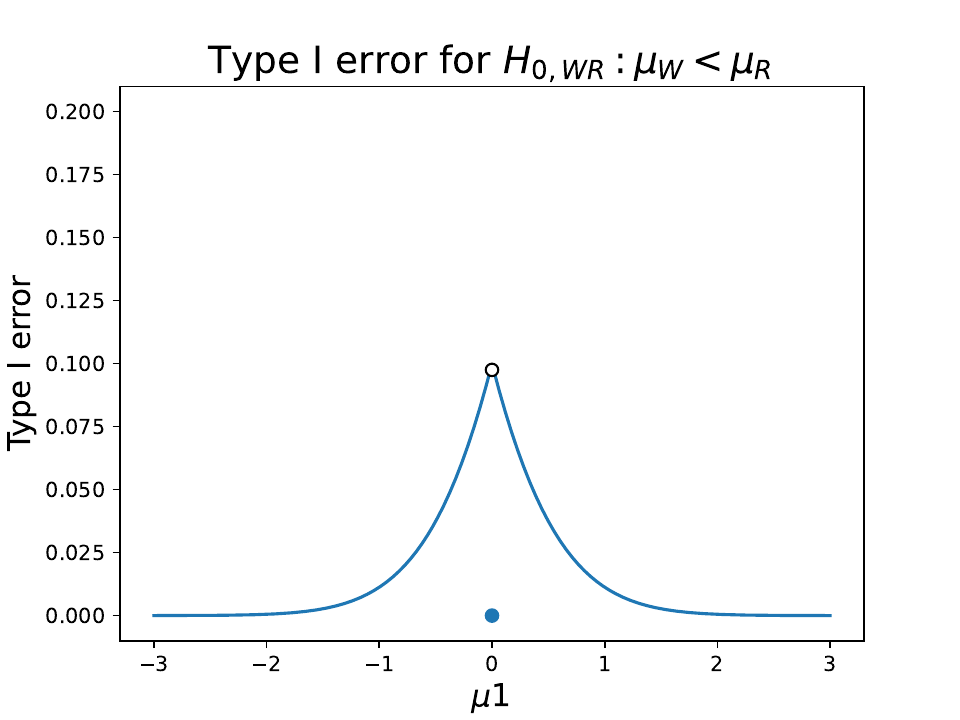}
    \end{minipage}\hfill
    \begin{minipage}{0.4\textwidth}
        \centering
        \includegraphics[width=\textwidth]{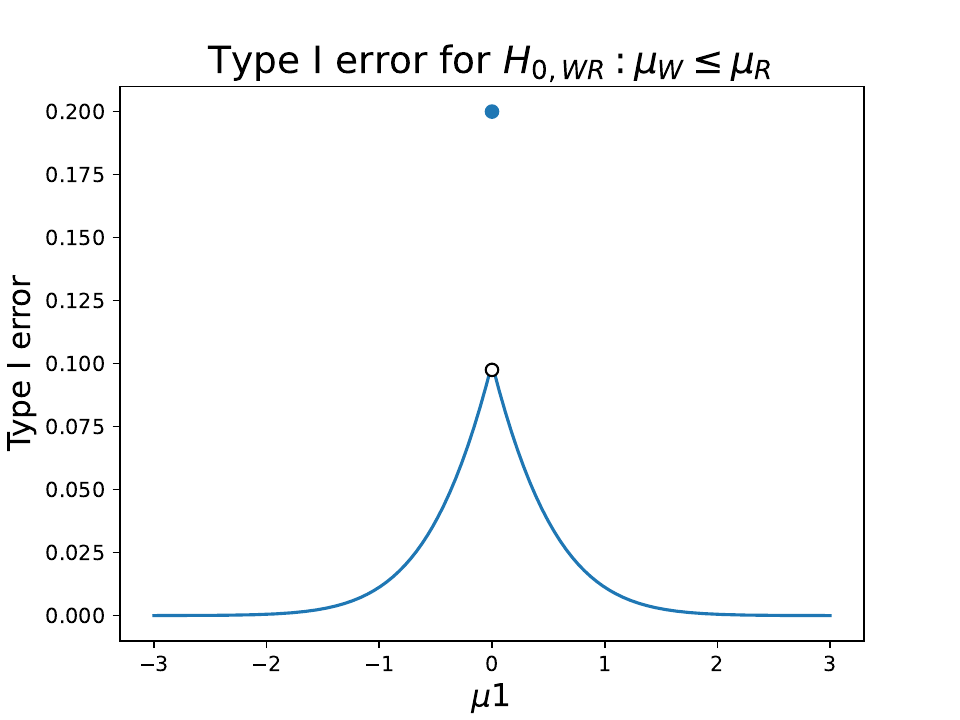}
    \end{minipage}\hfill
    }
    \caption{For $\mu_2=0$ and different $\mu_1$, the Type I error of rejecting $H_{0, WR}: \mu_W - \mu_R < 0$ and  $H_{0, WR}: \mu_W - \mu_R \leq 0$ when the level $\alpha$ one-sided test comparing the winner of $X_1 \sim N(\mu_1, 1/\sqrt{2})$ and $X_2 \sim N(\mu_2, 1/\sqrt{2})$  to the runner-up rejects.}
    \label{fig:warm_up_error_control}
\end{figure}

\subsection{Rank verification in exponential families}

\begin{figure}[]
    \centering
    \scalebox{1}{
    \begin{minipage}{0.4\textwidth}
        \centering
        \includegraphics[width=\textwidth]{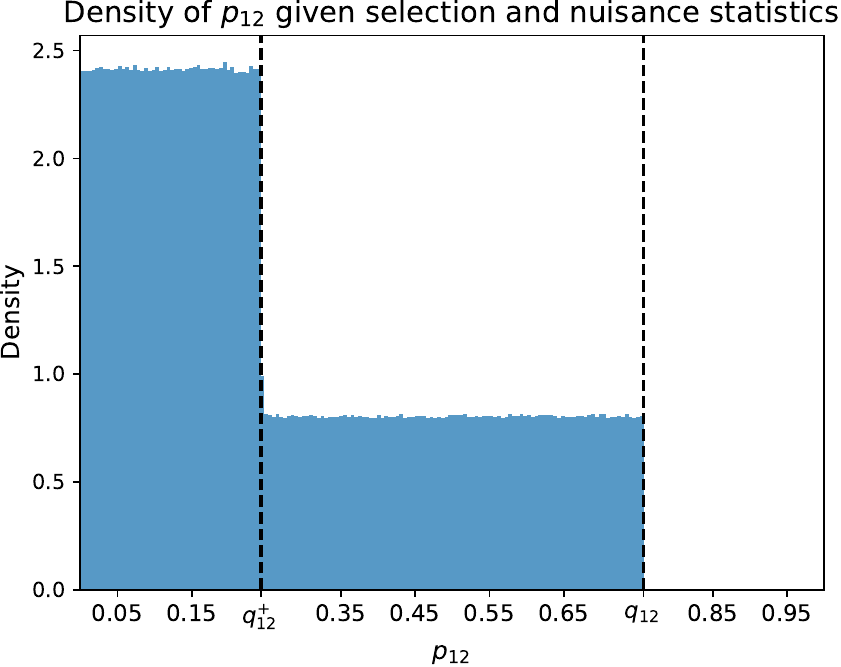}
    \end{minipage}\hfill
    \raisebox{0.18cm}{
    \begin{minipage}{0.4\textwidth}
        \centering
        \includegraphics[width=\textwidth]{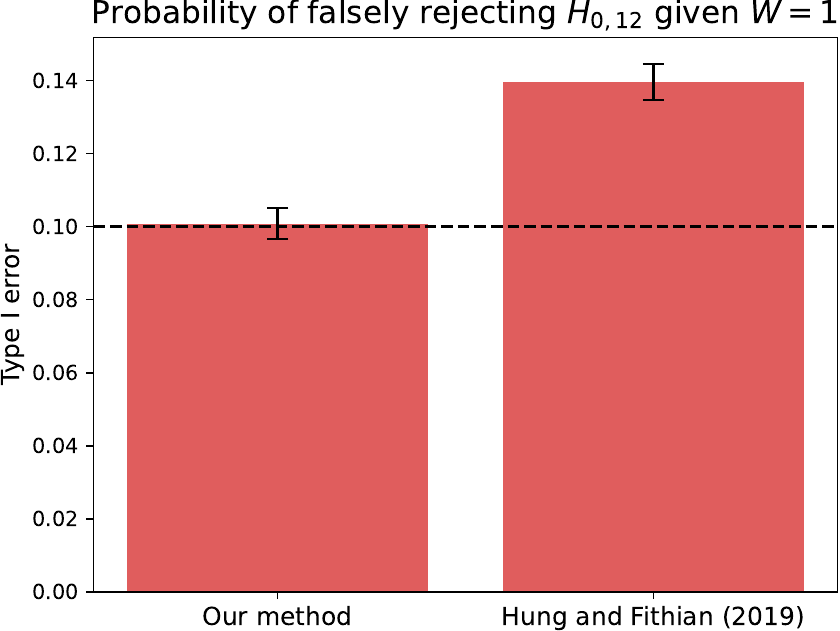}
    \end{minipage}\hfill
    }}
    \caption{Considering three independent binomials $X_i \sim \text{Bin}(b, s_i)$ with $b=4$ and $s_i = 1/2$, the first panel (left) depicts $N=10^6$ draws from the conditional distribution of the p-value $p_{12}$ (used for testing $H_{0, 12}: s_1 \leq s_2$) given $W = 1$ and the nuisance statistics $(X_1 + X_2)/2 = 2$, $X_3=2$. When $p_{12}  < q^+_{12}$, then $X_1$ is the sole winner and the p-value is selected for inference with probability one, but when $p_{12}^0 \in [q^+_{12}, q_{12}]$ there is a three-way tie and it is selected for inference with probability $1/3$. Hence, the p-value's conditional distribution is not uniform on $[0, q_{12}]$ as \cite{Hung2019} implicitly assume. The next panel (right) displays the consequence. Conditional on $W=1$, \cite{Hung2019} do not maintain Type I error control when testing $H^{0}_{W2}$ at level $\alpha=0.1$ (denoted by horizontal dashed line), whereas our method does. Error bars denote a 99\% confidence interval.}
    \label{fig:error_control}
\end{figure}

In this sub-section we illustrate how to do rank verification when we observe data $X \in \R^n$ from a natural multi-parameter exponential family $P_{\theta}$ with density
\begin{equation}
    \label{eq:exp_fam}
    g_{\theta}(x) = \exp(\theta_1 x_1 + \dots + \theta_n x_n - \psi(\theta))g(x),
\end{equation}
with respect to some base measure \footnote{We assume that $T_i(x) = x_i$ so that our discussion more closely mirrors \cite{Hung2019}, although we do not need to. \cite{Hung2019} also assume that $g(x)$ is Schur concave for other purposes, but we leave this assumption out.}. Like \cite{Hung2019}, we want to verify that $\theta_W$ is strictly larger than the remaining $\theta_j$, where $W$ is the index of the largest $X_i$. In case of ties, we follow \cite{Hung2019} and set $W$ to randomly be one of the winning indices. Since $E_{\theta}[X_i] = \theta_i$, the winning parameter $\theta_W$ may reasonably be the largest of the $\theta_i$. 

Considering the nulls $H_{0, ij} : \theta_i \leq \theta_j$, we want to reject the data-dependent null $\cup_{j \neq W} H_{0, Wj} $ and affirm that $\theta_W$ is strictly larger than every other parameter \footnote{If we performed the same analysis for the nulls $H^{\delta}_{0, ij} : \theta_i  - \theta_j \leq \delta$ then we could verify that $\theta_W$ is more than $\delta$ larger than any other $\theta_j$. Inverting these tests would result in a LCB for the difference $\theta_W - \max_{j \neq W } \theta_j$ between the winning and next largest parameter.}. Considering $j \neq W$, our strategy will be to come up with a valid test for $H_{0, Wj}$, and then reject $\cup_{j \neq W} H_{0, Wj}$ whenever we reject $H_{0, Wj}$ for every $j \neq W$. Fixing $i \neq j$, we start by constructing the UMPU p-value $p_{ij}$ for testing $H_{0, ij}: \theta_i \leq \theta_j$. Ultimately, we will only use this p-value to test $H_{0, ij}$ when $i$ is selected as our winning index, and we will correspondingly adjust the p-value to account for this selection. 

Defining the transformed sufficient statistics $Y \in \R^n$ by
\begin{equation}
    \label{eq:reparam}
    Y_i = \frac{X_i - X_j}{2}, \qquad  Y_j = \frac{X_i + X_j}{2}, \qquad  Y_{\ell} = X_{\ell} \text{ for } \ell \neq i, j,
\end{equation}
the random vector $Y$ has an exponential family density given by
\begin{equation}
    \tilde{g}_{\theta}(y) = \exp\left( (\theta_i - \theta_j) y_i + (\theta_i + \theta_j) y_j + \sum_{\ell \neq i, j} \theta_{\ell} y_{\ell} - \psi(\theta)  \right)\tilde{g}(y)
\end{equation}
with respect to some other base measure. It is then well-established (see \Cref{sec:one_sided_mlr_appdx}) that if we denote the conditional left-continuous survival function of $Y_i$ and its right-hand limit as
\begin{equation}
    G_{ij}(y_i | y_{-i}) = P_{\theta_i = \theta_j}(Y_i \geq y_i | Y_{-i} = y_{-i}) \qquad G_{ij}^+(y_i |y_{-i}) = \lim_{u \downarrow y_i } G_{ij}(u | y_{-i}),
\end{equation}
then the UMPU p-value $p_{ij}$ for testing $H_{0, ij}: \theta_i \leq \theta_j$ is given by 
\begin{equation}
    \label{eq:umpu_rank_verification}
    p_{ij} = G^+_{ij}(Y_i | Y_{-i}) + U_{ij, aux}(G_{ij}(Y_{i}|Y_{-i}) - G^+_{ij}(Y_i|Y_{-i})),
\end{equation}
where $U_{ij, aux}$ are $\text{Unif}([0, 1])$ random variables that are independent from each other and the data. By \Cref{exm:exp_fam}, the p-value $p_{ij}$ is selectively dominant given $Y_{-i}$. 

Crucially, we can tell if $X_i$ is a winner by examining the p-value $p_{ij}$. It is straightforward to confirm that $X_i$ is the sole winner exactly when $Y_i > \max_{k \neq i } Y_k - Y_j $. Equivalently, this happens when $p_{ij}$ is strictly smaller than 
\begin{equation}
    \label{eq:rank_verification_lower}
    q^+_{ij}(Y_{-i}) = G^+_{ij}(\max_{k \neq i} Y_k - Y_j | Y_{-i}).
\end{equation}
Likewise, one can confirm that $X_i$ is one of multiple winners exactly when $Y_i = \max_{k \neq i } Y_k - Y_j$, or equivalently when $p_{ij}$ is at least $q^+_{ij}$ but at most 
\begin{equation}
    \label{eq:rank_verification_upper}
    q_{ij}(Y_{-i}) = G_{ij}(\max_{k \neq i} Y_k - Y_j | Y_{-i}).
\end{equation}
Moreover, in the case that there are multiple winners, the number of winners is also a deterministic function of $Y_{-i}$:
\begin{equation}
    \label{eq:rank_verification_num_ties}
    N_{i}(Y_{-i}) = 1 + | \{ \ell \neq i : Y_{\ell} = \max_{k \neq i} Y_k  \} |.
\end{equation}
Note that $N_{i}(Y_{-i})$, which is always at least two, is \underline{not} the same as the number of winners, which can be one. Rather, it is the number of winners there will be if $X_i$ is a winner and at least one other $X_k$ is as well (see \Cref{sec:ties_appdx} for details).  

Leveraging these facts, we can apply our framework to come up with a valid test for $H_{0, Wj}$. Essentially, we use $p_{ij}$ to test $H_{0, ij}$ with probability one when it is less than $q^+_{ij}$, and with probability $1/N_i$ (we randomly select one of the $N_i$ winners) when it is between $q^+_{ij}$ and $q_{ij}$. Explicitly, letting $p = p_{ij}$ and $Z = Y_{-j}$, we can apply our framework with the selection function
\begin{equation*}
    s(x, z) = 
    \begin{cases} 
    1 & \text{if } x < q_{ij}^+(z), \\
    \frac{1}{N_i(z)} & \text{if } x \in [q_{ij}^+(z), q_{ij}(z)] \\
    0 & \text{otherwise}
    \end{cases}.
\end{equation*}
This is a piece-wise linear function that is easy to integrate, and, after some computations detailed in \Cref{sec:rank_verficiation_adj_appdx}, \Cref{thm:adjustment} tells us to reject when the selective p-value
\begin{equation}
    \label{eq:rank_verification_selective_p_val}
     \frac{p_{ij} - \left(1 - \frac{1}{N_i} \right)(p_{ij} - q^+_{ij})_+ }{q^+_{ij} + \frac{1}{N_i}(q_{ij} - q^+_{ij}) }
\end{equation}
is at most $\alpha$. 

The crucial difference between our derivation and \cite{Hung2019}'s is that \cite{Hung2019} use the same selective p-value when there can and cannot be ties amongst the $X_k$. If there cannot be ties amongst the $X_k$, then $q_{ij} = q^{+}_{ij}$ always, and the selective p-value \eqref{eq:rank_verification_selective_p_val} simplifies to $p_{ij}/q_{ij}$. If we use $p_{ij}/q_{ij}$ when ties are possible, however, we will not achieve conditional error control (as is claimed in \cite{Hung2019}). We give an example in \Cref{fig:error_control}, where $X \in \R^3$ is composed of three independent binomials. The left panel of \Cref{fig:error_control} depicts the conditional distribution of $p_{12}$ given $W=1$ for a specific setting of the nuisances statistics $Y_{-j}$. It makes it clear that rejecting when $p_{ij}/q_{ij} \leq \alpha $ does not maintain conditional error control, as is affirmed in \Cref{fig:error_control}'s right panel.

\section{Combining selective p-values}
\label{sec:multiple}

In this section we illustrate how our selective dominance viewpoint allows us to combine inferences across many p-values, even post-selection. Rather than selecting just one p-value to use for inference as in \Cref{sec:dominance}, some of this section's methods select many. To accommodate this, we generalize \Cref{sec:dominance}'s framework to allow us to select and perform post-selection inference with many p-values, not just one. This generalization is intuitive, and for sake of brevity we have deferred a formal account of it to \Cref{sec:multiple_p_vals_appdx}.

\subsection{Publication bias aware meta-analysis}

\begin{figure}
    \centering
    \includegraphics[width=0.55\textwidth]{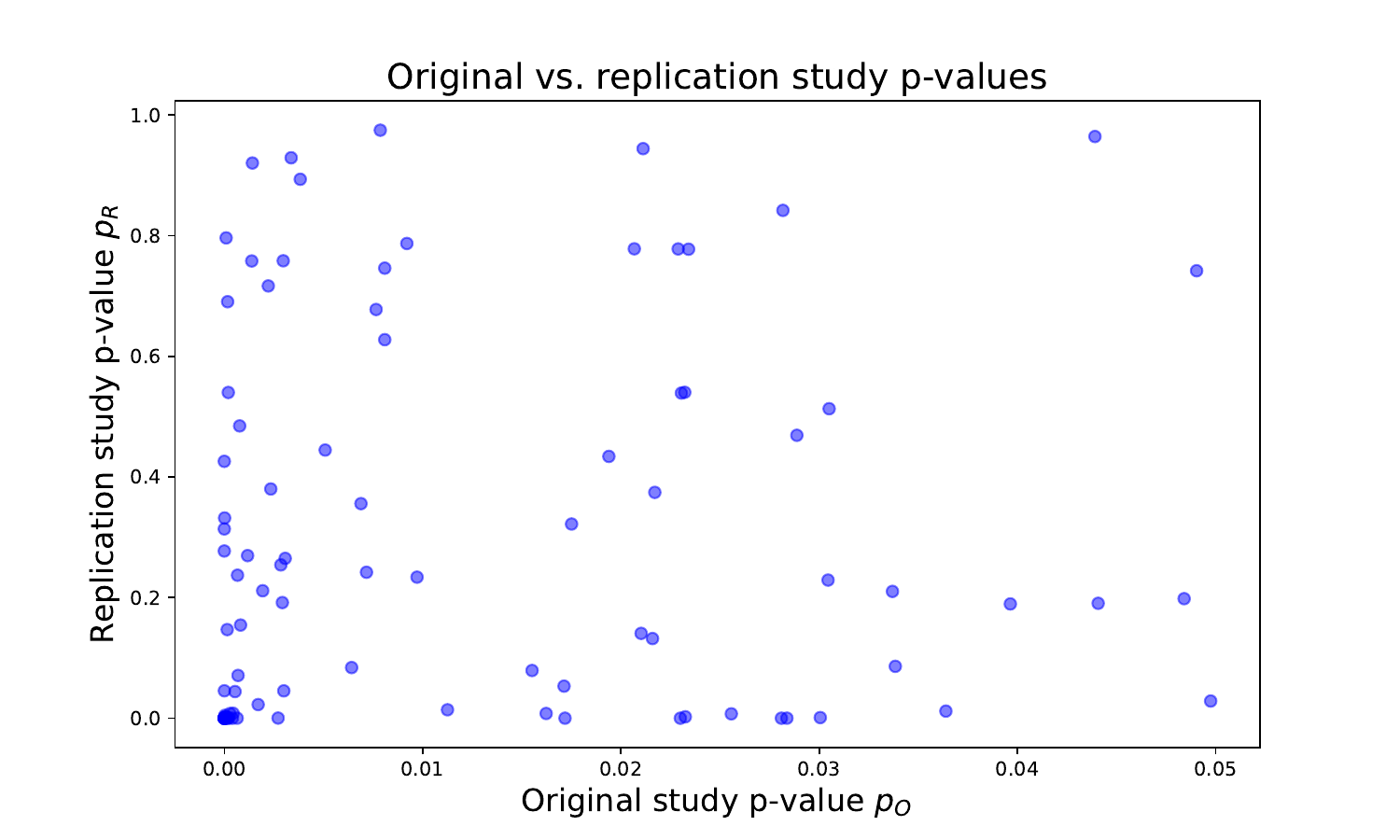}  
    \caption{Scatter plot of the original $p_O$ and replication $p_R$ p-values for 92 psychology studies from the open science collaboration's replication analysis \cite{OSF}. Note that the x-axis, which ranges from $[0, 0.05]$ is on a different scale than the y-axis, which ranges from $[0, 1]$. }
    \label{fig:replication}
\end{figure}

Performing replication studies is a crucial part of the scientific process, especially when one is wary that the original study may suffer from publication bias. To judge the prevalence of publication bias in psychology, the open science collaboration conducted a mass replication analysis of psychology studies \citep{OSF}. Via their efforts, we have access to p-values from 92 pairs of original and replication psychology studies \footnote{We exclude seven studies whose original p-value $p_O$ is larger than $\alpha=0.05$.}, depicted in \Cref{fig:replication}. We refer to p-values from the original study as $p_O$ and p-values from the replication study as $p_R$. The p-values $p_O$ from the original studies are significant at the $\alpha=0.05$ level, while only 34 of the replication p-values $p_R$ are significant. 

Although the original study p-values suffer from publication bias, they still contain valuable and usable information. By \Cref{exm:correction} and its subsequent discussion, $p_O/\alpha$ should be a valid p-value even in the presence of publication bias or p-hacking. Then, via Fisher's combination test, we can use both the corrected original p-value $p_O/\alpha$ and uncorrected replication p-value $p_R$ for inference. As a refresher, Fisher's test considers $n$ independent p-values $p_i$ for the nulls $H_{0, i}$ and rejects the global null $\cap_{i=1}^n H_{0, i}$ when the test statistic $-2 \sum_{i=1}^n \log(p_i) $ is at least as large as the $1-\alpha$ quantile of the $\chi^2_{2n}$ distribution. In our case, we have two independent p-values $p_O/\alpha$ and $p_R$ that test the same null hypothesis, and we can reject this null hypothesis when 
\begin{equation*}
    -2 (\log(p_O/\alpha) + \log(p_R)) \geq \text{Quantile}(1-\alpha, \chi^2_4). 
\end{equation*}
This approach is analogous to data-carving. After using part of our data for selection (the original study) and part for inference (the replication study), we still make use of the information remaining in the first part after selection for inference as well. Unlike existing approaches to data-carving, which are often complex and problem specific, using Fisher's combination test along with our selective dominance framework provides a general and simple way to data-carve. 

Our combination approach allows us to make more powerful inferences on the open science dataset. Our approach finds that 47 study pairs have significant findings, whereas using solely the original study p-value $p_O/\alpha$ or the replication study p-value $p_R$ results in only 39 or 34 significant findings respectively. It is surprising that, even after a harsh adjustment for publication bias, the corrected original study p-values result in more discoveries than replication study p-values. It is hard to gauge if this is due to chance, differences between the original and replication studies (e.g., minor differences in population demographics, devices used for measurement, sample size), or because there is somehow even harsher selection bias in the original studies than what we have accounted for. It may even be the case that some replicators felt incentivized to induce bias in the opposite direction, and tried to ensure that the replication studies were \underline{not} significant. 

\subsection{Adaptive versions of Fisher's combination test}

\begin{figure}
    \centering
    \hspace{-0.035\textwidth}
    \begin{minipage}{0.32\textwidth}
        \centering
        \includegraphics[width=\textwidth]{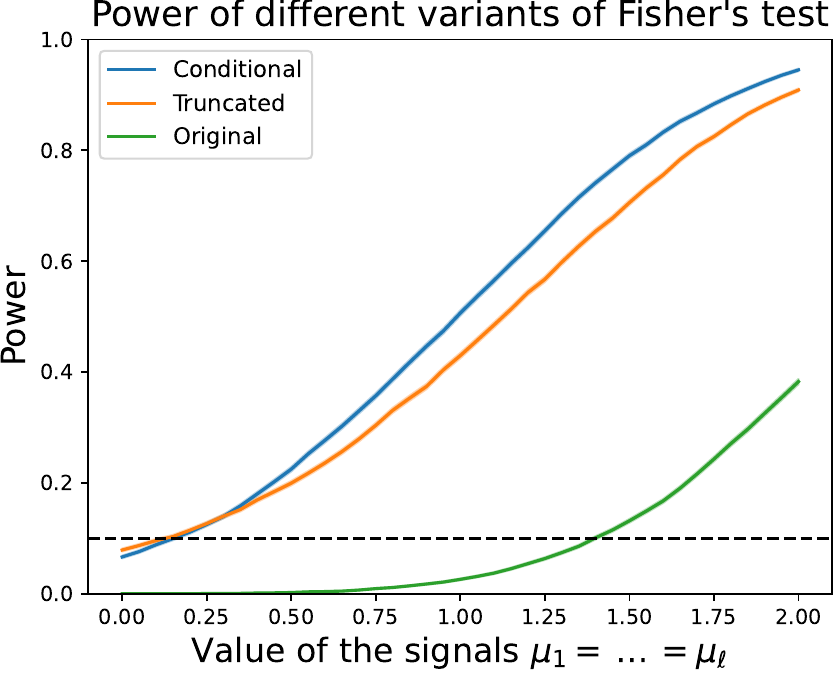}
        \caption*{(a) $\ell=3$}
    \end{minipage}
    \hfill
    \hspace{0.01\textwidth}
    \begin{minipage}{0.32\textwidth}
        \centering
        \includegraphics[width=\textwidth]{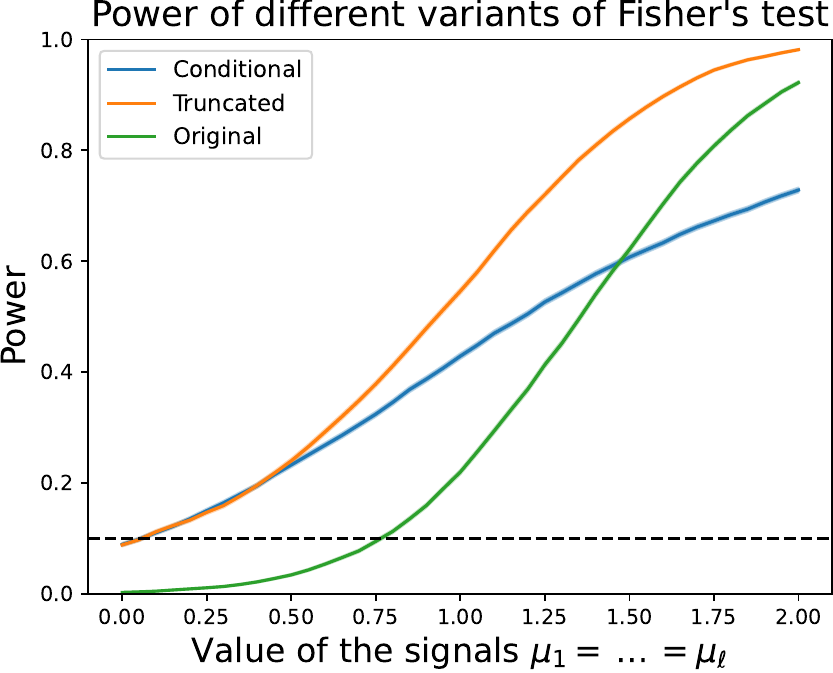}
        \caption*{(b) $\ell=5$}
    \end{minipage}
    \hfill
    \hspace{0.01\textwidth}
    \begin{minipage}{0.32\textwidth}
        \centering
        \includegraphics[width=\textwidth]{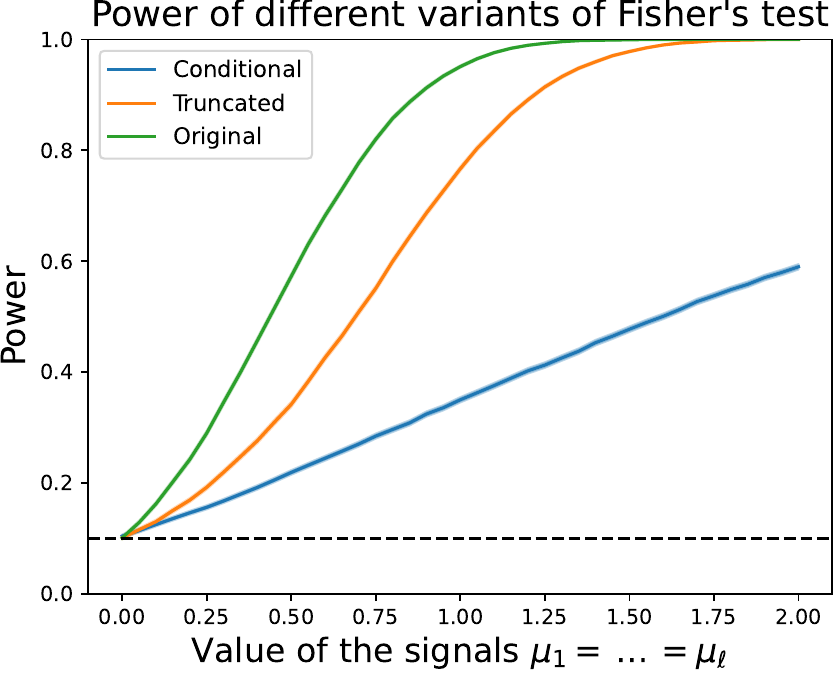}
        \caption*{(c) $\ell=10$}
    \end{minipage}
    \caption{ For $\ell=3$ (left), $\ell=5$ (middle), and $\ell=10$ (right), power of the $k=3$ conditional, $\tau=0.5$ truncated, and original Fisher's combination test for data drawn from $N(\mu, I_{10})$ with $\mu_{1}= \dots = \mu_{\ell}$ varying according to the x-axis and $\mu_{\ell + 1} = \dots = \mu_n = -2$. Power results from an average over $N=10^4$ trials, error bands (barely visible) denote one standard error, and the level $\alpha=0.1$ is denoted by the dashed line.}
    \label{fig:fisher}
\end{figure}

By employing similar ideas to the previous sub-section, we can come up with variants of Fisher's combination test that are more powerful when some null p-values are conservative (i.e., they have super-uniform distributions). \Cref{cor:cfisher}, which is of similar flavor to the conditional inference on winners procedure, gives a conditional version of Fisher's combination test that only uses the bottom $k$ p-values for inference. In \Cref{cor:cfisher}
s procedure, $k$ must be pre-specified. 

\begin{corollary}[Fisher's top-$k$ combination test]
    \label{cor:cfisher}
    Suppose that $p_i$ are $n$ independent and selectively dominant p-values for the nulls $H_{0, i}$, and let $H_{0, (j)}$ denote the null corresponding to the $j$th smallest p-value (with ties broken randomly). For some fixed $k \leq n$,  rejecting the data-dependent global null $\cap_{j=1}^k H_{0, (j)}$ (and therefore also the global null $\cap_{i=1}^n H_{0, i}$) when 
    \begin{equation}
        \label{eq:fisher_conditional}
        -2 \sum_{j=1}^k \log(p_{(j)}/p_{(k+1)}) \geq \text{Quantile}(1-\alpha, \chi^2_{2k})
    \end{equation}
    controls Type I error at level $\alpha$ conditional on the indices of the smallest $k$ p-values (and therefore also marginally). 
\end{corollary}

Examining \eqref{eq:fisher_conditional}, we see that when the $(k+1)$st p-value provides essentially no evidence against the null (so $p_{(k+1)} \approx 1$), running \Cref{cor:cfisher}'s test is like running Fisher's test using just the bottom $k$ p-values and ignoring that any selection took place. In this case, our test statistic will be essentially identical to Fisher's original test statistic, but the critical value required for rejection will be much smaller. On the flip-side, if many of the $p_{(j)}$ for $j \leq k$ are not sufficiently smaller than $p_{(k+1)}$, then \Cref{cor:cfisher}'s test statistic will be small and the test will fail to reject.

Another approach to improving Fisher's combination test is truncation: only use a p-value for inference if it is below some fixed threshold $\tau \in \R$ \citep{Zaykin}. Previous work, however, only establishes the validity of this test when the null p-values have exact $\text{Unif[0, 1]}$ distributions \citep{Zaykin, Zhang}. \Cref{cor:tfisher}, however, gives a version of Fisher's truncated combination test that is still valid whenever the p-values are independent and selectively dominant \footnote{We give a variant that is valid conditional on which p-values are selected. If we just care about rejecting the global null $\cap_{i=1}^n H_{0, i}$, we can give a more powerful test via marginalization, as \cite{Zaykin} do.}.

\begin{corollary}
    \label{cor:tfisher}
    Suppose that $p_i$ are $n$ independent and selectively dominant p-values for the nulls $H_{0, i}$, and fix $n$ thresholds $\tau_i \in [0, 1]$. Letting $j \in J$ denote the random set of indices for which $p_j \leq \tau_j$, rejecting the data-dependent global null $\cap_{j \in J} H_{0, j}$ (and therefore also the global null $\cap_{i=1}^n H_{0, i}$) when 
    \begin{equation*}
        -2 \sum_{j=1}^k \log(p_j/\tau_j) \geq \text{Quantile}(1-\alpha, \chi^2_{2|J|})
    \end{equation*} 
    controls Type I error at level $\alpha$ conditional on $J$ (and therefore also marginally). 
\end{corollary}

If some $p_j$ are substantially lower than their truncation point $\tau_j$ but most are above it, then \Cref{cor:tfisher} will be powerful. In this case, \Cref{cor:tfisher}'s test will have a slightly smaller statistic compared to Fisher's original combination test, but a much smaller critical value. Hence, the truncated Fisher test is most powerful when some p-values come from strong alternatives but many come from conservative nulls. As such, \Cref{cor:tfisher} generalizes the truncated Fisher test to the settings where it is most applicable. On the flip-side, if most of the $p_j$ are below $\tau_j$, the truncated test statistic will pay a penalty due to selection, while the critical value required for rejection will remain essentially unchanged compared to Fisher's original test. 

To illustrate the benefits and drawbacks of these methods, we display their power alongside that of Fisher's original test for a simple $n=10$ dimensional Gaussian problem. We sample $X \sim N(\mu, I_n)$ and use the p-values $p_i = 1 - \Phi(X_i)$ try and detect the existence of a positive mean. For $\ell \in \{3, 5, 10\}$, we vary the strength $\mu_1 = \dots = \mu_{\ell} > 0$ of our signals and set $\mu_{\ell + 1} = \dots = \mu_n = -2$ to be conservative nulls. We do inference using the bottom $k=3$ p-values for the conditional version of Fisher's method and set the truncation $\tau = 0.5$ for the truncated version (i.e., we include $p_i$ for which $X_i > 0$). The results are displayed in \Cref{fig:fisher}.

As expected, the new methods outperform Fisher's original method when conservative nulls are present. When $\ell=3$ and the bottom three p-values are much smaller than the rest, the conditional method does incredibly well. As expected, its performance quickly degrades when $\ell = 5, 10$ and the fourth smallest p-value becomes close to the bottom three. The truncated method is more robust, and still considerably improves power when $\ell=5$. Unsurprisingly, both selective methods perform worse than Fisher's original method when $\ell=10$ and every $\mu_i$ is a signal (i.e., all p-values are sub-uniform).

\section*{Acknowledgements}
I would like to thank John Cherian, Kevin Guo, Yash Nair, Will Hartog, Trevor Hastie, Jonathan Taylor, and James Yang for helpful discussions. I would like to especially thank James Yang for his thoughts regarding some measure theoretic aspects of the framework presented in this paper.

\bibliographystyle{plainnat}
\bibliography{bibliography.bib}

\begin{appendix}

\section{Additional derivations, details, and comments}

\subsection{Permutation test p-value}
\label{sec:perm_test_appdx}

The randomized permutation test from Proposition 3 of \cite{Hemerik} uses the p-value
\begin{equation*}
    p = \frac{\#\{1 \leq j \leq w : T(g_j(X)) > T(X) \}}{w} + U_{aux} \frac{\#\{1 \leq j \leq w : T(g_j(X)) = T(X) \}}{w},
\end{equation*}
where $U_{aux} \sim \text{Unif}([0, 1])$ adds auxiliary randomness that is independent of $X$.

\subsection{$F$-test p-value}
\label{sec:f_test_appdx}

For the matrix $\X \in \R^{n \times d}$, let $\X_q \in \R^{n \times q}$ denote the matrix that just contains the first $q$ columns of $\X$.
Then define 
\begin{equation*}
    \text{RSS}_f = \|Y -  \X (\X^\top \X)^{-1} \X^\top Y\|^2_2
\end{equation*}
\begin{equation*}
    \text{RSS}_r = \|Y -  \X_q (\X_q^\top \X_q)^{-1} \X_q^\top Y\|^2_2
\end{equation*}
Under the null $H_0 : \beta_{q+1} = \dots = \beta_d = 0$, the $F$ statistic 
\begin{equation*}
    F = \frac{(\text{RSS}_r - \text{RSS}_f)/(p-q)}{\text{RSS}_f/(n-p-1)}
\end{equation*}
has an exact $F_{d-q, n-q-1}$ distribution conditional on $\X$. Let $G$ denote the CDF of the $F_{d-q, n-q-1}$ distribution. Under $H_0$, the p-value $p = 1 - G(p)$ for this test has an exact $\text{Unif}([0, 1])$ distribution conditional on $\X$.

\subsection{Selective dominance and independence}
\label{sec:sel_dom_independence}

Consider a p-value $p$ for the null $H_0$ which is selectively dominant. Suppose that under $P_{H_0}$ the p-value $p$ is always independent of some random variable $Z$. We show that $p$ is then selectively dominant given $Z$. 

Fix a distribution $P$ in the null $H_0$ and let $f(x)$ be PDF of $p$. Considering a selection function $s(x)$ such that $\int_0^1 s(x) f(x) > 0$, we have that 
\begin{equation*}
    \frac{\int_0^t s(x) f(x) dx }{\int_0^1 s(x) f(x) dx } \leq  \frac{\int_0^t s(x) dx }{\int_0^1 s(x)dx } \text{ for all } t \in [0, 1]
\end{equation*}
from the selective dominance \eqref{eq:selective_dominance} of $p$. 

Now, consider a new selection function $s(x, z)$. Note that due to independence, $f(x)$ is still the conditional PDF of $p$ given $Z=z$.  For any $z$ such that $\int_0^1 s(x, z) f(x) > 0$, it follows from above that 
\begin{equation*}
    \frac{\int_0^t s(x, z) f(x) dx }{\int_0^1 s(x, z) f(x) dx } \leq  \frac{\int_0^t s(x, z) dx }{\int_0^1 s(x, z)dx } \text{ for all } t \in [0, 1]
\end{equation*}
If under $s(x, z)$, we have $P(S=1) = 0$, then the selective dominance condition \eqref{eq:selective_dominance} trivially holds. On the other hand, if $P(S=1) > 0$, then we have shown in \Cref{sec:adjustment_proof} that $\int_0^1 s(x, Z) f(x) > 0$ a.e. under $P(\cdot | S=1)$. Thus the above guarantees that a.e. under $P(\cdot | S=1)$,
\begin{equation*}
    \frac{\int_0^t s(x, Z) f(x) dx }{\int_0^1 s(x, Z) f(x) dx } \leq  \frac{\int_0^t s(x, Z) dx }{\int_0^1 s(x, Z)dx } \text{ for all } t \in [0, 1].
\end{equation*}
By arguments from \Cref{sec:adjustment_proof}, the left-hand side is equal to the left-hand side of \eqref{eq:selective_dominance} a.e. under $P(\cdot | S=1)$, and the right-hand side equals the right-hand side of \eqref{eq:selective_dominance} a.e. under $P(\cdot | S=1)$, so we have established the claim.

\subsection{Data carving for the file-drawer problem}
\label{sec:carve_appdx}

We have two data samples $X_1 \sim N(\mu, 2)$ and $X_2 \sim N(\mu, 2)$ that are independent and want to test $H_0 : \mu \leq 0$. Suppose we only do inference because we observed that $X_1  > t$ for some threshold $t$. If we consider the p-value $p_{full} = 1 - \Phi((X_1 + X_2)/2  )$, then our selection function is given by 

\begin{align*}
    s(x) &= P( X_1 > t | p_{full} = x ) \\
         &=  P( X_1 > t | \frac{X_1 + X_2}{2} =  \Phi^{-1}(1 -x) )\\
         &= 1 - \Phi(t - \Phi^{-1}(1 - x))
\end{align*}
where we have used that 
\begin{equation*}
    \begin{bmatrix}
    X_1 \\ \frac{X_1 + X_2}{2}
    \end{bmatrix} \sim N \left(\begin{bmatrix}
        \mu \\ \mu
        \end{bmatrix}, \begin{bmatrix}
            2  & 1 \\ 1 & 1
            \end{bmatrix} \right)
\end{equation*}
so 
\begin{equation*}
    X_1 | \frac{X_1 + X_2}{2} = y \sim N(y, 1)
\end{equation*}
Thus our corrected p-value is given by 
\begin{equation*}
    p_{carve} = \frac{\int_0^{p_{full}}  1 - \Phi(t - \Phi^{-1}(1-x) )   dx }{\int_0^1 1 - \Phi(t - \Phi^{-1}(1-x) ) dx}
\end{equation*}
\begin{equation*}
    p_{carve} = \frac{\int_{\bar{X}}^{\infty} \phi(z) (1 - \Phi(t - z))   dz }{\int_{-\infty}^{\infty} \phi(z) (1 - \Phi(t - z))  dz } = \frac{\int_{\bar{X}}^{\infty} \phi(z) (1 - \Phi(t - z))   dz }{1 - \Phi(t/\sqrt{2})   }
\end{equation*}

We now show that $p_{carve}$ is monotone non-decreasing in $t$. Letting $Z$ and $Y$ be independent standard normal random variables and fixing some constant $a$, the selective p-value is given by 
\begin{equation*}
   p_{carve} = \frac{P(Z + Y > t, Z > a)}{P(Z + Y > t)} = P(Z > a | Z + Y > t)
\end{equation*}
for $a = \bar{X}$. Letting $W = Z + Y$ we can write $Z = \frac{1}{2}W +\epsilon$ where $\epsilon$ is independent of $W$. This gives us 
\begin{equation*}
    p_{carve} = P(\frac{1}{2} W + \epsilon > a | W > t) = E[P(W > 2(a - \epsilon)| W > t, \epsilon)| W > t ] = E[P(W > 2(a - \epsilon)| W > t, \epsilon)]
\end{equation*}
Then the fact that $p_{carve}$ is monotone non-decreasing in $t$ follows from the fact that $P(W > c | W > t)$ is monotone non-decreasing in $t$ for every constant $c$:
\begin{equation*}
{P(W > c |W >t)} = \begin{cases} 
\frac{P(W > c)}{P(W > t)} & \text{if } t \leq c, \\
1 & \text{if } t > c.
\end{cases}
\end{equation*}

\subsection{Post selection inference for the LASSO}
\label{sec:lasso_appdx}

In the below example, we freely refer to results from \cite{Lee2016}. 

\begin{example}[Post selection inference for the LASSO]

    Suppose that $p$ is a valid p-value for testing the null $H_0$ and it is selectively dominant given $Z$. Imagine there are known functions $q^+(z) \leq q^-(z)$ and $v(z)$, and we only choose to test $H_0$ with $p$ when $x(Z) > 0$ and $p \in [q^+(Z), q^{-}(Z)]$.  Applying our framework with the selection function $s(x) = I(p \in [q^+(Z), q^{-}(Z)])I(v(Z) > 0)$, \Cref{thm:adjustment} tells us that we control selective Type I error if we reject according to the selective p-value $p_{sel} = \frac{p - q^+(Z)}{q^-(Z) - q^+(Z)}$:
    \begin{equation}
        \label{eq:lasso_error_control}
        P_{H_{0}}\left(\frac{p - q^+(Z)}{q^-(Z) - q^+(Z)} \leq \alpha | S = 1\right)  \leq \alpha.
    \end{equation} 
    
    Now, suppose we observe n-dimensional data $Y \in N(\mu, \sigma^2)$ and a fixed matrix $\X \in \R^{n \times p}$. For a fixed regularization strength $\lambda$, the LASSO solution is composed of a fitted coefficient vector $\hat{\beta} \in \R^p$ and a sign vector $\hat{s} \in \R^p$ that satisfy the following Karush-Kuhn-Tucker conditions:
    \begin{align*}
        &X^{\top}(X \hat{\beta} - Y) + \lambda \hat{s} = 0,\\
        &\hat{s}_i = \text{sign}(\hat{\beta}_j)  &\text{ if } \hat{\beta}_j \neq 0,\\
        &\hat{s}_i \in [-1, 1] &\text{ if } \hat{\beta}_j = 0.
    \end{align*}
    We define the fitted model from the LASSO to be the set 
    \begin{equation*}
        \hat{M} = \{i \in \{1, \dots, p\} : |\hat{s}_i| = 1\}
    \end{equation*}
    For almost every $\lambda$, this is equal to the set of $i \in [p]$ for which $\hat{\beta}_i$ is non-zero.

    Now we establish our inferential target. For a subset $M \subseteq [p]$, define $\X_M \in \R^{n \times |M|}$ to be the matrix that has the columns of $\X$ indexed by $M$. Our goal is to do inference on the population parameters for the model $M$,
    \begin{equation*}
        \beta^M = \argmin_{b \in R^{|M|}} E[ \|Y - \X_M b \|^2_2] = \X^{\top} \mu.
    \end{equation*}
    Particularly, we focus on doing inference on the $j$th index of this parameter vector $\beta^M_j = e_j^{\top} \X_M^{\dagger} \mu$ by testing the null $H_{0, j}^M: \beta^M_j \leq 0$. Letting $\Omega^M = X^{\dagger} (X^{\dagger})^{\top}$, we see that $e_j^{\top}X_M^{\dagger} y \sim N(\beta_j^M, \sigma^2 \Omega^M_{jj} )$. Thus we can test this null using the p-value 
    \begin{equation*}
        p^M_j = 1 - \Phi\left( \frac{e_j^{\top}X_M^{\dagger} y}{\sigma \sqrt{\Omega^M_{jj}}}\right),
    \end{equation*}
    which is selectively dominant by \Cref{exm:mlr}. However, we only use this p-value to test the null when we observe that $\hat{M} = M$. 
    
    Define $Z = (I_n - (X_M^{\dagger})^{\top}e_je_j^{\top} X_M^{\dagger}/\Omega^M_{jj}) Y  $ so that $Y$ and $Z$ are independent. Lemma 5.1 of \cite{Lee2016}, often dubbed the polyhedral lemma, precisely tells us that there are functions 
    \begin{align*}
        &q_j^{M, s, +}(z) = 1 - \Phi\left( \frac{\mathcal{V}^{M, s, +}(z)}{\sigma \sqrt{\Omega^M_{jj}}}\right) \\
        &q_j^{M, s, -}(z) = 1 - \Phi\left( \frac{\mathcal{V}^{M, s, -}(z)}{\sigma \sqrt{\Omega^M_{jj}}}\right) \\
        &\mathcal{V}^{M, s, 0}(z)
    \end{align*}
    of $z$, such that 
    \begin{equation*}
        \{\hat{M} = M, \hat{s}_M = s \} = \{ p^{M}_j \in [q_j^{M, s, +}(Z), q_j^{M, s, -}(Z)],   \mathcal{V}_{M, s}^0(Z)\geq 0\} 
    \end{equation*}
    The functions $\mathcal{V}^{M, s, -}(z)$, $\mathcal{V}^{M, s, +}(z)$, $\mathcal{V}^{M, s, 0}(z)$ match those given in \cite{Lee2016}. 

    Now it is easy to see that rejecting the data-dependent null  $H_{0, j}^{\hat{M}}$ when $ (p^{\hat{M}}_j - q_j^{\hat{M}, \hat{s}, +}(Z))/(q_j^{\hat{M}, \hat{s}, -}(Z) - q_j^{\hat{M}, \hat{s}, +}(Z))$ is at most $\alpha$ controls Type I error conditionally on $\hat{M}$, and therefore also marginally. If $H_{0, j}^{M}$ is false, then trivially $P(\text{falsely reject } H_{0, j}^{\hat{M}} | \hat{M} = M, \hat{s}_M = s) = 0 \leq \alpha $. For the case that $H_{0, j}^M$ is true, the selection event $ \hat{M} = M$ and $\hat{s}_M = s$ is the same as selecting $p^M_j$ for inference in \eqref{eq:lasso_error_control}, so 
    \begin{align*}
        P(\text{falsely reject } H_{0, j}^{\hat{M}}| \hat{M} = M, \hat{s}_M = s) &= P \left( \frac{p^{\hat{M}}_j - q_j^{\hat{M}, \hat{s}, +}(Z)}{q_j^{\hat{M}, \hat{s}, -}(Z) - q_j^{\hat{M}, \hat{s}, +}(Z)} \leq \alpha \;\middle|\; \hat{M} = M, \hat{s}_M = s \right)\\
        &=P \left( \frac{p^{M}_j - q_j^{M, s, +}(Z)}{q_j^{M, s, -}(Z) - q_j^{M, s, +}(Z)} \leq \alpha \;\middle|\; \hat{M} = M, \hat{s}_M = s \right)\\
        &\leq \alpha 
    \end{align*}
    Conditional error control on $\hat{M}$ then follows from the law of total probability,
    \begin{align*}
        P(\text{falsely reject } H_{0, j}^{\hat{M}} | \hat{M} = M) &= \sum_{s \in \{-1, 1 \}^{|M|}} P(\hat{s}_M = s | \hat{M} = M) P(\text{falsely reject } H_{0, j}^{\hat{M}} | \hat{M} = M, \hat{s}_M = s)\\
                                                                   &\leq \alpha \sum_{s \in \{-1, 1 \}^{|M|}} P(\hat{s}_M = s | \hat{M} = M)\\
                                                                   &=\alpha,
    \end{align*}
    as does marginal error control,
    \begin{align*}
        P(\text{falsely reject } H_{0, j}^{\hat{M}}) &= \sum_{M \subseteq [p]} P(\hat{M} = M) P(\text{falsely reject } H_{0, j}^{\hat{M}} | \hat{M} = M)  \\
                                                     &\leq \alpha \sum_{M \subseteq [p]}P(\hat{M} = M)\\
                                                     &= \alpha. 
    \end{align*}
    \end{example}

\subsection{Comparing conditional and simultaneous inference}
\label{sec:beta_dist_appdx}

We consider a setting where we have $n$ independent and selectively dominant p-values $p_1, \dots, p_n$ that are all anti-conservative, i.e.,  $p_j \preceq \text{Unif}([0, 1])$. At worst, these p-values are exact uniforms (e.g., they come from the boundary of the null). 

We will show that, on an event with probability at least $1-\epsilon$, the conditional procedure, which rejects when $p_{(1)} \leq \alpha p_{(2)}$, can only reject if $p_{(1)} \leq C_{\epsilon}/n$ for some constant $C_{\epsilon} > 0$. Hence, without conservative nulls, the conditional approach behaves roughly on the same order as the classical approach (Sidak).   

Letting $U_1, \dots, U_n$ be independent $\text{Unif}([0, 1])$ random variables, two facts are clear. First that $U_{(2)} \sim \text{Beta}(2, n-1)$ has mean $\frac{2}{n + 1} < \frac{2}{n}$ and standard deviation $\sqrt{\frac{2(n-1)}{(n+1)^2(n+2)}} \leq \frac{2}{n}$, and second that $p_{(2)} \preceq U_{(2)}$. 

Fix any $\epsilon > 0$. We have by Chebyshev's inequality that  
\begin{align*}
    P\left(p_{(2)} \leq \frac{2}{n}(1 + \epsilon^{-\frac{1}{2}})\right) &\geq P(p_{(2)} \leq E[U_{(2)}] + \sqrt{\Var(U_{(2)})}/\sqrt{\epsilon})\\
    &\geq P(U_{(2)} \leq E[U_{(2)}] + \sqrt{\Var(U_{(2)})}/\sqrt{\epsilon})\\
    &\geq P(|U_{(2)} - E[U_{(2)}]| < \sqrt{\Var(U_{(2)})}/\sqrt{\epsilon} ) \\
    &\geq 1 - \epsilon
\end{align*}
Define $C_{\epsilon} = 2(1 + \epsilon^{-\frac{1}{2}})\alpha$ and the event $A_{\epsilon} = \{ p_{(2)} \leq \frac{2}{n}(1 + \epsilon^{-\frac{1}{2}})\}$. The event $A_{\epsilon}$ has probability at least $1-\epsilon$. Furthermore, on this event, the conditional procedure never rejects when $p_{(1)} > C_{\epsilon}/n$, completing our claim. 

In other words, the conditional procedure can only reject when $p_{(1)} \leq C_{\epsilon}/n$ (aside for on a small probability event), and hence suffers the same curse of dimensionality as the classical method:
\begin{equation*}
    P(p_{(1)} \leq \alpha p_{(2)} \text{ and } p_{(1)} > C_{\epsilon}/n) \leq P(A_{\epsilon}^c) \leq \epsilon \text{ for some constant } C_{\epsilon} \text{ independent of } n.  
\end{equation*}

\subsection{Conditional inference on winners}
\label{sec:cond_appdx}

We first provide the standard derivation for the conditional LCB and CI for the winning mean in the setting of independent Gaussian data. Then we show that it matches our p-value viewpoint. We then argue that the selective p-value we use for this winner problem is monotone in the null parameter, which allows us to prove the some facts about the more general conditional LCB and CI that applies for all MLR families. Finally, we use our p-value viewpoint to argue that, in the independent Gaussian case, the standardized distance between the winner $X_W$ and conditional LCB is purely a function of the standardized distance between the winner $X_W$ and runner-up $X_R$. 

\subsubsection{Standard derivation}

Suppose we observe independent Gaussian data $X \sim N(\mu, I_n)$ and want to make a LCB for the mean $\mu_W$ of the winner $W = \argmax_{i \in [n]} X_i$. 

To construct the conditional LCB, we follow the framework of \cite{Fithian2017}. Letting $R$ be the index of the runner-up (second largest observation), we note that the deviation of $X_{W}$ from $\mu_{W}$ has a truncated normal distribution once we condition on $W$ and the nuisance statistics $X_{-W}$:
\begin{equation}
    \label{eq:cond_dist}
    X_W - \mu_W \mid W, X_{-W} \sim TN(0, 1, X_{R} - \mu_{W}, \infty).
\end{equation}
Let
\begin{equation}
    \label{eq:cond_quantile}
    q_{1-\alpha}(x_r, \mu_w) = \text{Quantile}_{\mu}(1-\alpha, X_W -  \mu_{W} \mid W=w, X_{-w} = x_{-w})
\end{equation}
denote the $1-\alpha$ quantile of this conditional distribution \eqref{eq:cond_dist}, which is a function of largest entry $x_r$ of $x_{-w}$ and the mean $\mu_w$ at the winning index. It is straightforward to show that
\begin{equation}
\label{eq:cond_lcb}
     LCB_{cond}(X) = \{\mu_0 : \mu_0 > X_{W} - q_{1-\alpha}(X_R, \mu_0)  \} 
\end{equation}
is a $1-\alpha$ confidence region for $\mu_{W}$ that has exact coverage conditional on $W$ and  $X_{-W}$:
\begin{equation*}
    P_{\mu}( \mu_{W} \in LCB_{cond}(X) | W, X_{-W})  = P_{\mu}( X_W - \mu_W < q_{1-\alpha}(X_R,\mu_{W}) |W, X_{-W}) = 1-\alpha.
\end{equation*}
Later, via our p-value viewpoint, we will argue that this confidence region is indeed an LCB. 

Fixing some $0 < \alpha_1, \alpha_2 < \alpha$ such that $\alpha_1 + \alpha_2 = \alpha$, we also see that 
\begin{equation}
    \label{eq:cond_ci}
        CI_{cond}(X) = \{\mu_0 : X_W - q_{\alpha_2}(X_R, \mu_0) > \mu_0 > X_W - q_{1-\alpha_1}(X_R, \mu_0)  \} 
\end{equation}
is a $1-\alpha$ confidence region that has exact coverage conditional on $W$ and $X_{-W}$. Again, later via our p-value viewpoint we will argue that this region is a CI. 

\subsubsection{The p-value viewpoint }

For the same setting as above, we want to use the p-values $p^{\mu_0}_i = 1 - \Phi(X_i - \mu_0)$ to characterize when the conditional LCB does not include $\mu_0 \in \R$. This happens exactly when $\mu_0$ is not included in the set \eqref{eq:cond_lcb}. Examining \eqref{eq:cond_dist}, \eqref{eq:cond_quantile}, and \eqref{eq:cond_lcb}, this happens when $X_W - \mu_0$ is at least as large as the $1-\alpha$ quantile $Q$ of a standard normal truncated to be larger than $X_R - \mu_0$. This quantile satisfies 
    \begin{equation*}
        \alpha = \frac{1 - \Phi(Q) }{1 - \Phi(X_R - \mu_0) }.
    \end{equation*}
Solving for $Q$ gives $Q = \Phi^{-1}(1 - \alpha(1 - \Phi(X_R - \mu_0)))$, meaning we reject exactly when 
\begin{align*}
    X_{W} - \mu_0 \geq \Phi^{-1}(1 - \alpha(1 - \Phi(X_{R} - \mu_0))) &\iff 1 - \Phi(X_{W} - \mu_0) \leq \alpha(1 - \Phi(X_{R} - \mu_0))\\
    &\iff p^{\mu_0}_{(1)} \leq \alpha p^{\mu_0}_{(2)}.
\end{align*}

Now we do the same for the conditional CI \eqref{eq:cond_ci}. We want to characterize when $\mu_0$ is not included in the set \eqref{eq:cond_ci}. Examining \eqref{eq:cond_dist}, \eqref{eq:cond_quantile}, and \eqref{eq:cond_ci}, this happens either when $X_W - \mu_0$ is at least as large as the $1-\alpha_1$ quantile or at most as small as the $\alpha_2$ quantile of the same truncated normal distribution. That is, either 
\begin{align*}
    X_{W} - \mu_0 \geq \Phi^{-1}(1 - \alpha_1(1 - \Phi(X_{R} - \mu_0))) & \iff p^{\mu_0}_{(1)} \leq \alpha_1 p^{\mu_0}_{(2)},
\end{align*}
or 
\begin{align*}
    X_{W} - \mu_0 \leq \Phi^{-1}(1 - (1-\alpha_1)(1 - \Phi(X_{R} - \mu_0))) & \iff p^{\mu_0}_{(1)} \geq (1-\alpha_2) p^{\mu_0}_{(2)},
\end{align*}
This will match the conditional CI we give later, which we will write in terms of p-values.

\subsubsection{Monotonicity of the selective p-value}

Recall the setting where $X_i \sim P_{\theta_i}$ are independent samples from some parametric family $P_{\theta}$ with MLR in $T(x)$, and let $p^{\theta_0}_i$ be the UMP p-values for testing $H_0 : \theta \leq \theta_0$. 

We want to show that the selective p-value $p^{\theta_0}_j/\min_{i \neq j} p^{\theta_0}_j$ is monotone non-decreasing in $\theta_0$. This will imply that $p^{\theta_0}_{(1)}/p^{\theta_0}_{(2)}$ is monotone non-decreasing in $\theta_0$. So long as we can write our selection function in terms of the data with no dependence on $\theta_0$, \Cref{sec:one_sided_monotone_appdx} guarantees that this will be the case. Recall that each p-value $p^{\theta_0}_i$ is a function of $T(X_i)$ and auxiliary uniform random variables $U_{i, aux}$ that are independent of the data and each other. Imagining using our framework with $p=p^{\theta_0}_j$ and $Z=(T(X_{-j}), U_{-j, aux})$, we can write our selection function in terms of the data $T(X_i)$ and $U_{i, aux}$ as 

\begin{equation*}
    \tilde{s}(t_j, u_j, t_{-j}, u_{-j}) =
    \begin{cases} 
    1, & \text{if } t_j > \max_{i \neq j } t_i, \\
    1, & \text{if } t_j = \max_{i \neq j } t_i \text{ and } u_j \leq \max_{i \neq j : t_i = t_j} u_i, \\
    0, & \text{otherwise}.
\end{cases}
\end{equation*}
where $t_{-j}$ and $u_{-j}$ jointly represent $z$. 

Because this selection function does not depend on $\theta_0$, \Cref{sec:one_sided_monotone_appdx} guarantees that the selective p-value resulting from it will be monotone. Note that, to apply \Cref{sec:one_sided_monotone_appdx}, we also needed to ensure that $Z$ did not depend on $\theta_0$, which is was not the case in \Cref{exm:winner}'s original treatment, but is the case in our current treatment. 

\subsubsection{The conditional LCB and CI for MLR families}
Again suppose that $X_i \sim P_{\theta_i}$ are independent samples from some MLR family $P_{\theta}$, and let $p^{\theta_0}_i$ be the UMP p-values for testing $H_0 : \theta \leq \theta_0$. Let $W = \argmin_{i \in [n]} p^{\theta_0}_i$ be the index of the smallest p-value, and define the parameter vector $\Theta = (\theta_1, \dots, \theta_n)$. 

First we claim that 
\begin{equation*}
    \{\theta_0 : p^{\theta_0}_{(1)}/p^{\theta_0}_{(2)} > \alpha\}
\end{equation*}
is a LCB for the winning parameter $\theta_W$ that has exact $1-\alpha$ coverage conditional on $W$. The region being an LCB follows from what we showed earlier: $p^{\theta_0}_{(1)}/p^{\theta_0}_{(2)}$ is monotone non-decreasing in $\theta_0$. We get exact $1-\alpha$ coverage because $p^{\theta_j}_j$ has an exact uniform distribution given $p^{\theta_j}_{-j}$ under $P_{\Theta}$ (see \Cref{lem:uniform}). Thus conditional on $W = j$ and $p^{\theta_j}_{-j}$ (i.e., conditional on selection and $Z$), \Cref{thm:adjustment} tells us that the selective p-value $p^{\theta_j}_{j}/ \min_{i \neq j} p^{\theta_j}_i$ has an exact uniform distribution conditional on selection:
    
\begin{align*}
        P_{\Theta}\left( \theta_W \in \left\{ \theta_0 : \frac{p^{\theta_0}_{(1)}}{p^{\theta_0}_{(2)}} > \alpha \right\} | W=j \right) &= P_{\Theta}\left( \theta_j \in \left\{ \theta_0 : \frac{p^{\theta_0}_{j}}{ \min_{i \neq j} p^{\theta_0}_i} > \alpha \right\} | W=j \right)\\
                                     &= P_{\Theta}\left(  \frac{p^{\theta_j}_{j}}{ \min_{i \neq j} p^{\theta_j}_i} > \alpha  | W=j \right)\\
                                    &= 1-\alpha,
\end{align*} 

Likewise, we see that 
\begin{equation*}
    \{\theta_0 : \alpha_1 < p^{\theta_0}_{(1)}/p^{\theta_0}_{(2)} < 1-\alpha_2 \}
\end{equation*}
is a CI for $\theta_W$ that has exact $1-\alpha$ coverage conditional on $W$. The fact that it is a CI again follows from the monotonicity of $p^{\theta_0}_{(1)}/p^{\theta_0}_{(2)}$, and exact coverage follows from an identical argument to the one above.  

\subsubsection{Distance between winner and conditional LCB }
\label{sec:gap_appdx}

Consider observing Gaussian data $X \sim N(\mu, \sigma^2 I_n)$ and let $W$ and $R$ be the indices of the winner and runner up respectively. We will use our p-value viewpoint to show that that the standardized distance $D = (X_W - \hat{\mu})/\sigma$ between the winner $X_W$ and the conditional LCB $\hat{\mu}$ for $\mu_W$ depends only on the standardized gap $(X_W - X_R)/\sigma$ between the winner and runner-up. Per our earlier discussions, the conditional LCB $\hat{\mu}$ satisfies 
\begin{equation*}
    \frac{p^{\hat{\mu}}(X_W)}{p^{\hat{\mu}}(X_R)} = \alpha \iff \frac{1 - \Phi((X_W - \hat{\mu})/\sigma )}{1 - \Phi((X_R - \hat{\mu})/\sigma)} = \alpha \iff \frac{1 - \Phi(D)}{1 - \Phi(D - (X_W - X_R)/\sigma)} =\alpha.
\end{equation*}
Clearly then, $D$ is a function of $(X_W - X_R)/\sigma$ and does not at all depend on the problem dimension $n$. 

\subsection{Conditional inference on winners for exponentials}
\label{sec:exponential_winner_appdx}

Recall the exponential distribution $X \sim \text{Exp}(\lambda_i)$ which has PDF
\begin{equation*}
    f_{\lambda}(x) = 
\begin{cases} 
\lambda e^{-\lambda x} & x > 0, \\ 
0 & x \leq 0 
\end{cases}.
\end{equation*}
Defining $T(x) = 1/x$, we see for $x > 0$ and $\lambda_2 \geq \lambda_1$, the ratio  
\begin{equation*}
   f_{\lambda_2}(x)/f_{\lambda_1}(x) = \frac{\lambda_2}{\lambda_1} \exp\left( - \frac{\lambda_2 - \lambda_1}{T(x)}\right)
\end{equation*}
is monotone non-decreasing in $T(x)$. Thus this family has an MLR in $T(x)$, and the UMP test for $H_0 : \lambda \leq \lambda_0$ thus rejects when $T(X)$ is large, or correspondingly when $X$ is small. In particular, noting that the CDF of $X$ is given by 
\begin{equation*}
F_X(x) = 
\begin{cases} 
1 - e^{-\lambda x} & x > 0, \\ 
0 & x \leq 0 
\end{cases}
\end{equation*}
it rejects according to the p-value $p^{\lambda^0} = 1 - e^{-\lambda_0 X}$ (see \Cref{sec:one_sided_mlr_appdx} for details about UMP testing in MLR families).

Now consider observing some independent data $X_i \sim \text{Exp}(\lambda_i)$. We know from \Cref{sec:cond_appdx} that the selective p-value $p^{\lambda_0}_{(1)}/p^{\lambda_0}_{(2)}$ is monotone non-decreasing in $\lambda_0$. Using L'Hopital's rule, we can compute that 
\begin{align*}
    \lim_{\lambda_0 \downarrow 0} \frac{p^{\lambda_0}_{(1)}}{p^{\lambda_0}_{(2)}} &=  \lim_{\lambda_0 \downarrow 0}  \frac{1 - e^{-\lambda_0 X_{(1)}} }{1 - e^{-\lambda_0 X_{(2)}}}\\
    &= \lim_{\lambda_0 \downarrow 0} \frac{X_{(1)}e^{-\lambda_0 X_{(1)}}}{X_{(2)}e^{-\lambda_0 X_{(2)}}}\\
    &= \frac{X_{(1)}}{X_{(2)}},
\end{align*}
which is sufficient to imply our claim in the main text.  

\subsection{Hybrid inference on winners}
\label{sec:hybrid_appdx}

We first provide the standard derivation for the hybrid LCB and CI for the winning mean in the setting of independent Gaussian data. Then we show that it matches our p-value viewpoint. We then provide more general hybrid confidence regions that apply for all MLR families, and show that they are indeed an LCB and CI. Finally, we use our p-value viewpoint to argue that, in the independent Gaussian case, the standardized distance between the winner $X_W$ and conditional LCB is purely a function of the standardized distance between the winner $X_W$ and runner-up $X_R$ as well as the problem dimension $n$. 

\subsubsection{Standard derivation}

We want to make a hybrid LCB for the mean $\mu_{W}$ of the winner $W = \argmax_{i \in [n]} X_i$ in the case of independent Gaussian data $X \sim N(\mu, I_n)$ with unknown mean $\mu \in \R^n$. 

The core idea behind hybrid inference is giving a confidence region $C_{hyb}(X)$ that has a very high probability of containing $\mu_W$ on a ``good'' event $G_{\mu}$. Oddly, this good event depends on the unknown parameter. For some $\beta < \alpha$, we need $G_{\mu}$ to happen with probability at least $1-\beta$. Then, if we ensure that $C_{hyb}(X)$ has at least $(1-\alpha)/(1-\beta)$ coverage on the $G_{\mu}$, it will achieve $1-\alpha$ coverage overall:
\begin{align*}
       P_{\mu}(\mu_{W} \in C_{hyb}(X)) &= P_{\mu}(G_{\mu})P_{\mu}(\mu_{W} \in C_{hyb}(X) | G_{\mu}) + P_{\mu}(G_{\mu}^c)P_{\mu}(\mu_{W} \in C_{hyb}(X) | G_{\mu}^c)\\
       &\geq (1-\beta)\left(\frac{1-\alpha}{1-\beta} \right)\\
       &=1-\alpha.
\end{align*}

Considering some $\beta < \alpha$ and defining $\beta_n = 1 - (1-\beta)^{1/n}$, our good event is that the confidence lower bounds $X_i -  z_{1-\beta_n}$ for the means $\mu_i$ all simultaneously hold:
\begin{equation*}
    G_{\mu} = \{X_i < \mu_i + z_{1-\beta_n} \text{ for all } i \in [n] \}.
\end{equation*}
It is not hard to show that this good event happens with probability exactly $1-\beta$.

We note that our treatment differs slightly from \cite{Andrews2023}. \cite{Andrews2023}, who focus on making confidence intervals. Their good event is that there are CIs (not LCBs) that simultaneously cover all the $\mu_i$. The purpose of restricting to this good event is that on this good event, we know all the $X_i$ are not too far from their $\mu_i$, and therefore our inferences on this event should not result in exploding intervals (even if they are of conditional flavor). Because the explosion only happens for the lower endpoint, however, it is most important to ensure that the $X_i$ are not too far above their means. Hence we put all our $\beta$ error budget on the good event into making simultaneous LCBs, rather than two-sided CIs like \cite{Andrews2023}.  

Now, we can make a confidence region that contains the mean with probability at least $(1-\alpha)/(1-\beta)$ on this good event. If we condition on $G_{\mu}$ along with $W$ and $X_{-W}$, the deviation of $X_W$ from $\mu_W$ has a truncated normal distribution like \eqref{eq:cond_dist} that is further truncated from above:
\begin{equation}
    \label{eq:hyb_dist}
    X_W - \mu_W \mid W, X_{-W}, G_{\mu} \sim TN(0, 1, X_{R} - \mu_{W}, z_{1-\beta_n} ).
\end{equation}
On the good event $G_{\mu}$, we always have that $X_R - \mu_{W} < X_W - \mu_{W} < z_{1-\beta_n}$, so the lower truncation is indeed below the upper one. Let
\begin{equation}
\label{eq:hyb_quantile}
    q^{h}_{\frac{1-\alpha}{1-\beta}}(x_r, \mu_w) = \text{Quantile}_{\mu}\left(\frac{1-\alpha}{1-\beta}, X_W -  \mu_W \mid W=w, X_{-w} = x_{-w}, G_{\mu}\right)
\end{equation}
denote the $(1-\alpha)/(1-\beta)$ quantile of the conditional distribution \eqref{eq:hyb_dist}, which is a function of the largest entry $x_r$ of $x_{-w}$ and the winning mean $\mu_{w}$. Per the prior discussion, the function \eqref{eq:hyb_quantile} only makes sense if $x_r - \mu_0$ at most $z_{1-\beta_n}$, and we will take the quantile \eqref{eq:hyb_quantile} to be $-\infty$ if it is not. It is then straightforward to show that
\begin{equation}
\label{eq:hyb_lcb_}
     LCB_{hyb}(X) = \{\mu_0 : \mu_0 > X_{W} - q^{h}_{\frac{1-\alpha}{1-\beta}}(X_R, \mu_0)  \} 
\end{equation}
contains $\mu_{W}$ with high probability conditional on $W$, $X_{-W}$, and the event $G_{\mu}$:
\begin{equation*}
    P_{\mu}( \mu_{W} \in LCB_{hyb}(X) | W, X_{-W}, G_{\mu})  = P_{\mu}( X_W - \mu_W < q^{h}_{\frac{1-\alpha}{1-\beta}}(X_R, \mu_W) |W, X_{-W}, G_{\mu}) = \frac{1-\alpha}{1-\beta}.
\end{equation*}
Based on our earlier discussions, this is sufficient to imply that $LCB_{hyb}(X)$ from \eqref{eq:hyb_lcb_} will contain $\mu_{W}$ with probability at least $1-\alpha$. We argue later (via our p-value viewpoint) that \eqref{eq:hyb_lcb_} is an LCB.

Similarly, fixing some $\gamma_1, \gamma_2 > 0$ such that $\gamma_1 + \gamma_2 = 1 = (1-\alpha)/(1-\beta) = \frac{\alpha -\beta}{1-\beta}$, we can argue that 
\begin{equation}
    \label{eq:hyb_ci_}
         CI_{hyb}(X) = \{\mu_0 : X_W - q^{h}_{\gamma_2}(X_R, \mu_0)  > \mu_0 > X_{W} - q^{h}_{1 - \gamma_1}(X_R, \mu_0)  \} 
\end{equation}
will contain $\mu_W$ with probability at least $1-\alpha$. Again, we will show via our p-value viewpoint that \eqref{eq:hyb_ci_} is a CI. 

\subsubsection{The p-value viewpoint} 

For the same setting as above, we want to use the  p-values $p^{\mu_0}_i = 1 - \Phi(X_i - \mu_0)$ to characterize when the hybrid LCB \eqref{eq:hyb_lcb_} does not contain $\mu_0 \in \R$. Examining \eqref{eq:hyb_dist}, \eqref{eq:hyb_quantile}, and \eqref{eq:hyb_lcb_}, we can consider two cases to figure out when this happens. \newline 

\noindent \textbf{Case One - $X_R - \mu_0 \geq z_{1 - \beta_n}$:} If $X_R - \mu_0 \geq z_{1 - \beta_n}$, then $q^h_{\frac{1-\alpha}{1-\beta}}(W, X_{-W}, \mu_0) = -\infty$, so $\mu_0$ cannot be in \eqref{eq:hyb_lcb_}. This case happens precisely when 
\begin{equation*}
    X_R-\mu_0 \geq z_{1 - \beta_n} \iff 1 - \Phi(X_R - \mu_0) \leq 1 - \Phi(z_{1- \beta_n}) \iff p^{\mu_0}_{(2)} \leq \beta_n.
\end{equation*}

\noindent \textbf{Case Two - $X_R - \mu_0 < z_{1 - \beta_n}$:} If $X_R - \mu_0 < z_{1 - \beta_n}$, then $\mu_0$ is not in \eqref{eq:hyb_lcb_} exactly when $X_W - \mu_0$ is at least as large as the $\frac{1-\alpha}{1-\beta}$ quantile $Q$ of a standard normal truncated to be larger than $X_{R} - \mu_0$ but smaller than $z_{1- \beta_n}$. This quantile satisfies 
\begin{equation*}
    \frac{\alpha - \beta}{1- \beta} = \frac{\Phi(z_{1 - \beta_n})  - \Phi(Q) }{\Phi(z_{1-\beta_n})  - \Phi(X_R - \mu_0) } = \frac{1 - \beta_n  - \Phi(Q) }{1 - \beta_n - \Phi(X_R - \mu_0) } 
\end{equation*}
Solving for $Q$ gives 
\begin{equation*}
    Q = \Phi^{-1} \left(1 - \beta_n - \frac{\alpha - \beta}{1 - \beta}\left(1 - \beta_n - \Phi(X_{R} - \mu_0) \right) \right) = \Phi^{-1} \left( \left(1 - \frac{\alpha - \beta}{1 - \beta}\right) (1 - \beta_n) + \frac{\alpha-\beta}{1-\beta}\Phi(X_{R} - \mu_0) \right),
\end{equation*}
so $\mu_0$ is not in \eqref{eq:hyb_lcb_} exactly when
\begin{align*}
        X_{W} - \mu_0 \geq \Phi^{-1} \left( \left(1 - \frac{\alpha - \beta}{1 - \beta}\right) (1- \beta_n) + \frac{\alpha-\beta}{1-\beta}\Phi(X_{R} - \mu_0) \right) \\
        \iff 1 -\Phi(X_{W} - \mu_0) \leq \left(1 - \frac{\alpha - \beta}{1 - \beta}\right)\beta_n + \frac{\alpha-\beta}{1-\beta}(1-\Phi(X_{R} - \mu_0))  \\
        \iff p^{\mu_0}_{(1)} \leq \frac{\alpha-\beta}{1-\beta}p^{\mu_0}_{(2)}  + \left(1 - \frac{\alpha - \beta}{1 - \beta}\right)\beta_n \\
\end{align*}

It turns out we can combine these two cases. Because $p^{\mu_0}_{(1)} \leq p^{\mu_0}_{(2)}$, the fact that $p^{\mu_0}_{(2)} \leq  \beta_n$ in Case One implies that $p^{\mu_0}_{(1)} \leq \beta_n$ also. Therefore in Case One, $p^{\mu_0}_{(1)}$ must be strictly smaller than a mixture of $p^{\mu_0}_{(2)}$ and $\beta_n$:
\begin{equation}
    \label{eq:hybrid_cuttoff_appdx}
    p^{\mu_0}_{(1)} \leq  \frac{\alpha-\beta}{1-\beta}p^{\mu_0}_{(2)}  + \left(1 - \frac{\alpha - \beta}{1 - \beta}\right)\beta_n. 
\end{equation}
Therefore, if say $\mu_0$ is not included whenever \eqref{eq:hybrid_cuttoff_appdx}, we will always say it is not included in Case One, and we will say it is not included at the appropriate times in Case Two. 

Now, we do the identical analysis for the hybrid CI \eqref{eq:hyb_ci_}. Case One is identical. We redo Case Two below. \newline

\noindent \textbf{Case Two - $X_R - \mu_0 < z_{1 - \beta_n}$:} If $X_R - \mu_0 < z_{1 - \beta_n}$, then $\mu_0$ is not in \eqref{eq:hyb_ci_} exactly when $X_W - \mu_0$ is at least as large as the $1 - \gamma_1$ quantile of the same truncated normal, and at most as small as that truncated normal's $\gamma_2$ quantile. That is, when either 
\begin{align*}
    X_{W} - \mu_0 \geq \Phi^{-1} \left( \gamma_1 (1- \beta_n) + (1-\gamma_1)\Phi(X_{R} - \mu_0) \right) \\
        \iff p^{\mu_0}_{(1)} \leq \gamma_1 p^{\mu_0}_{(2)}  + (1 -\gamma_1)\beta_n \\
\end{align*}
or 
\begin{align*}
    X_{W} - \mu_0 \leq \Phi^{-1} \left( (1 - \gamma_2) (1- \beta_n) + \gamma_2\Phi(X_{R} - \mu_0) \right) \\ 
    \iff p_{(1)}^{\mu_0} \geq (1- \gamma_2)p_{(2)}^{\mu_0} + \gamma_2 \beta_n
\end{align*}

By the same reasoning as above, we can combine the two cases, and say that we do not include $\mu_0$ whenever 

\begin{equation*}
    p^{\mu_0}_{(1)} \leq \gamma_1 p^{\mu_0}_{(2)}  + (1 -\gamma_1)\beta_n  \text{ or } p_{(1)}^{\mu_0} \geq (1- \gamma_2)p_{(2)}^{\mu_0} + \gamma_2 \beta_n.
\end{equation*}

\subsubsection{The hybrid LCB and CI for MLR families}

Recall the setting where $X_i \sim P_{\theta_i}$ are independent samples from some parametric family $P_{\theta}$ with MLR in $T(x)$, and let $p^{\theta_0}_i$ be the UMP p-values for testing $H_0 : \theta \leq \theta_0$. Let $\Theta = (\theta_1, \dots, \theta_n)$ be the parameter vector and let $W \argmin_{i \in [n]} p^{\theta_0}_i$ be the index of the smallest p-value. 

We will present a hybrid LCB and CI for $\theta_W$ and directly prove their validity. Importantly, note that the event $G_{\Theta} = \{p^{\theta_i}_i \geq \beta_n \text{ for all } i \in [n]\}$ happens with probability $1-\beta$ under $P_{\Theta}$. Thus it suffices to show that our confidence regions have $(1-\alpha)/(1-\beta)$ coverage on this event, per our earlier discussion. Also importantly, thanks to $p^{\theta_j}_j$ having a uniform distribution under $P_{\theta_j}$ (see \Cref{lem:uniform}), the conditional distribution of $p^{\theta_j}_j$ given $W=j$, $p^{\theta_j}_{-j}$, and $G_{\Theta}$ is a uniform distribution on $[\beta_n, \min_{i \neq j} p^{\theta_j}_i]$. Using this, we can see that 
\begin{equation*}
    \left\{\theta_0 : p^{\theta_0}_{(1)} > \frac{\alpha-\beta}{1-\beta}p^{\theta_0}_{(2)}  + \left(1 - \frac{\alpha - \beta}{1 - \beta}\right)\beta_n  \right\}
\end{equation*}
is a $1-\alpha$ confidence region for $\theta_W$:
\begin{align*}
    &P_{\Theta}\left( \theta_W \in \left\{\theta_0 : p^{\theta_0}_{(1)} > \frac{\alpha-\beta}{1-\beta}p^{\theta_0}_{(2)}  + \left(1 - \frac{\alpha - \beta}{1 - \beta}\right)\beta_n  \right\} | W= j, p^{\theta_j}_{-j}, B_{\Theta} \right) \\
    &= P_{\Theta}\left( \theta_j \in \left\{\theta_0 : \frac{p^{\theta_0}_j - \beta_n}{\min_{i \neq j}p^{\theta_0}_i  - \beta_n} > \frac{\alpha-\beta}{1-\beta}  \right\} | W= j, p^{\theta_j}_{-j}, B_{\Theta} \right) \\
    &= P_{\Theta}\left( \frac{p^{\theta_0}_j - \beta_n}{\min_{i \neq j}p^{\theta_0}_i  - \beta_n} > \frac{\alpha - \beta}{1-\beta} | W= j, p^{\theta_j}_{-j}, B_{\Theta} \right) \\
    &= \frac{1-\alpha}{1-\beta},
\end{align*} 
To see that the region is actually an LCB, we rewrite it as 
\begin{equation}
    \label{eq:hyb_lcb_monotone}
    \left\{\theta_0 : \frac{p^{\theta_0}_{(1)}}{p^{\theta_0}_{(2)}} > \frac{\alpha-\beta}{1-\beta}  + \left(1 - \frac{\alpha - \beta}{1 - \beta}\right)\beta_n \bigg/p^{\theta_0}_{(2)}  \right\}
\end{equation}
because we know that $p^{\theta_0}_{(2)}$ is monotone non-decreasing in $\theta_0$ and also from \Cref{sec:cond_appdx} that $p^{\theta_0}_{(1)}/p^{\theta_0}_{(2)}$ is monotone non-decreasing in $\theta_0$, once some $\theta_0$ belongs to \eqref{eq:hyb_lcb_monotone} all larger $\theta_0$ will as well. Hence, it is an LCB.  

Considering $\gamma_1, \gamma_2 > 0$ such that $\gamma_1 + \gamma_2 = (\alpha-\beta)/(1-\beta)$, we can do the exact same analysis for the hybrid CI
\begin{equation*}
    \left\{\theta_0 : \gamma_1 p^{\theta_0}_{(2)} + (1-\gamma_1)\beta_n < p^{\theta_0}_{(1)} < (1- \gamma_2)p_{(2)}^{\theta_0} + \gamma_2 \beta_n\right\}
\end{equation*}
to show that it both covers $\theta_W$ with probability at least $1-\alpha$ and that it is indeed a CI. 

Note that, when put in our selective dominance framework, the hybrid selection event actually depends on the null parameter $\theta_0$. Hence we cannot argue monotonicity of the selective p-value via \Cref{sec:one_sided_monotone_appdx} as we did in the conditional case, and instead had to use different arguments to justify that our confidence regions were actually an LCB and CI. 

\subsubsection{Distance between winner and hybrid LCB}
\label{sec:hybrid_gap_appdx}

Consider observing Gaussian data $X \sim N(\mu, \sigma^2 I_n)$ and let $W$ and $R$ be the indices of the winner and runner up respectively. We will use our p-value viewpoint to show that the standardized distance $D = (X_W - X_R)/\sigma$ between the winner $X_W$ and the hybrid LCB $\hat{\mu}$ for $\mu_W$ depends only on the standardized gap $(X_W - X_R)/\sigma$ between the winner and runner-up. Per our earlier discussion the hybrid LCB $\hat{\mu}$ satisfies 
\begin{align*}
    & p^{\hat{\mu}}(X_W) = \frac{\alpha-\beta}{1-\beta} p^{\hat{\mu}}(X_R) + \left(1 - \frac{\alpha-\beta}{1-\beta} \right)\beta_n\\
    & 1 - \Phi\left(\frac{X_W - \hat{\mu}}{\sigma}\right) = \frac{\alpha-\beta}{1-\beta} \left(1 - \Phi\left(\frac{X_R - \hat{\mu}}{\sigma}\right)\right) + \left(1 - \frac{\alpha-\beta}{1-\beta} \right)\beta_n\\ 
    &\iff 1 - \Phi(D) = \frac{\alpha-\beta}{1-\beta} \left(1 - \Phi\left(D - \frac{X_W - X_R}{\sigma}\right) \right) + \left(1 - \frac{\alpha-\beta}{1-\beta} \right)\beta_n.
\end{align*}
Clearly then, $D$ is a function of $(X_W - X_R)/\sigma$ and $n$.

\subsection{Comparing hybrid inference to the union bound}
\label{sec:hybrid_sim_appdx}

\begin{figure}
    \centering
    \hspace{-0.035\textwidth}
    \begin{minipage}{0.32\textwidth}
        \centering
        \includegraphics[width=\textwidth]{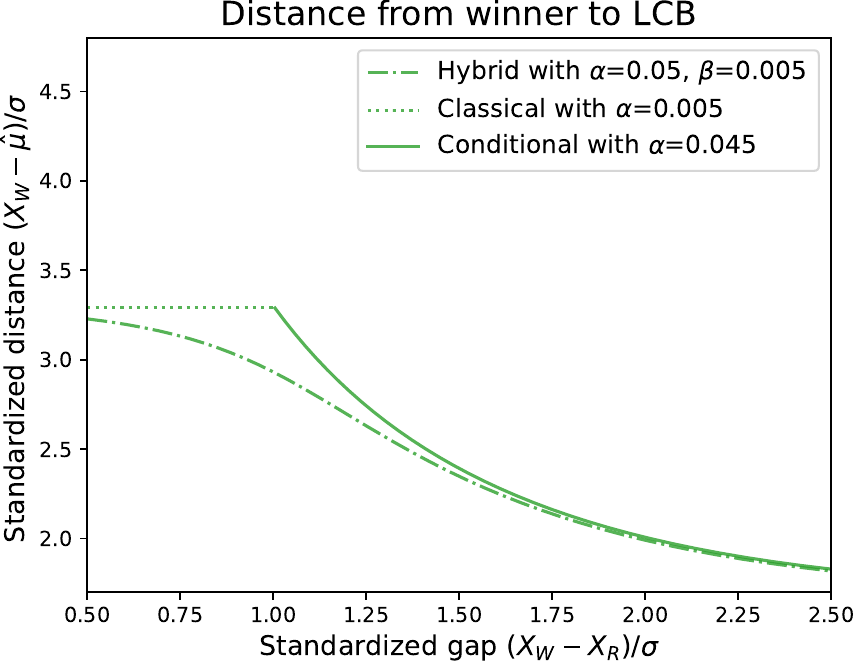}
        \caption*{(a) $n=10$}
    \end{minipage}
    \hfill
    \hspace{0.01\textwidth}
    \begin{minipage}{0.32\textwidth}
        \centering
        \includegraphics[width=\textwidth]{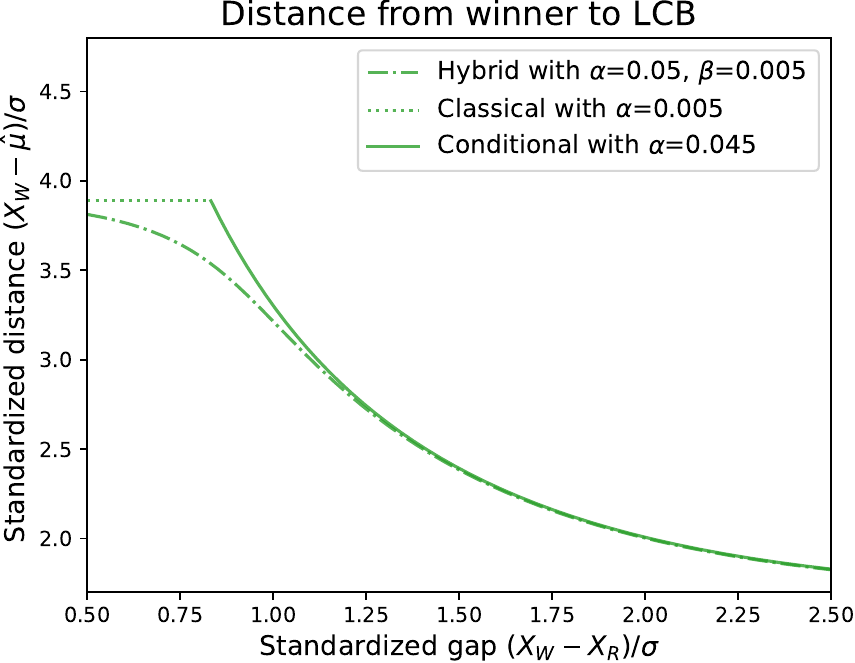}
        \caption*{(b) $n=100$}
    \end{minipage}
    \hfill
    \hspace{0.01\textwidth}
    \begin{minipage}{0.32\textwidth}
        \centering
        \includegraphics[width=\textwidth]{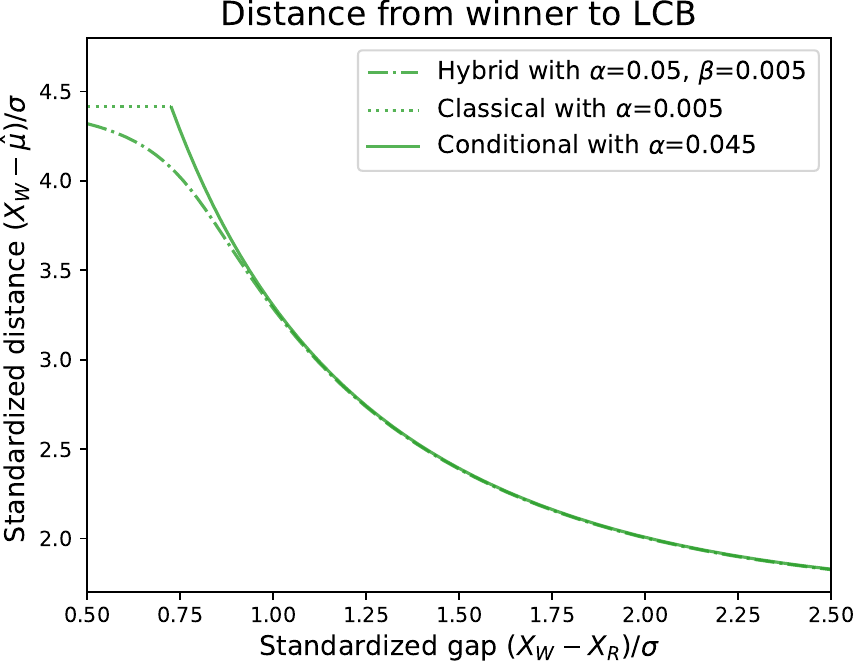}
        \caption*{(c) $n=1000$}
    \end{minipage}
    \caption{ For $n=10$ (left), $n=100$ (middle), and $n=1000$ (right) the distance between the hybrid LCB to the winner (dash-dot line) and union bound LCB to the winner (dotted and solid line) with $\alpha=0.05$ and $\beta=0.005$ plotted as a function of the gap between the winning and runner-up observation.}
    \label{fig:hybrid_union}
\end{figure}

As discussed earlier the hybrid cutoff \eqref{eq:hybrid_cutoff_thm} is strictly larger than the union bound cutoff \eqref{eq:union_bound_cutoff} when $p_{(1)} > \beta_n$. Thus, both procedures reject when $p_{(1)} \leq \beta_n$. When $p_{(1)} > \beta_n$, the hybrid procedure rejects but the union bound does not whenever  
\begin{equation*}
    p_{(1)} \in \bigg((\alpha - \beta)p_{(2)},  \frac{\alpha - \beta}{1-\beta}p_{(2)} + \left(1 -  \frac{\alpha - \beta}{1-\beta}\right)\beta_n \bigg]
\end{equation*}
When $p_{(2)} = 1$ and $n=1$, the length of the interval is $\beta$, which is the largest it can possible be. Thus, in the case where we can have additional rejections, the hybrid cutoff is never more than $\beta$ plus the union bound cutoff. \cite{Andrews2023} suggests taking $\beta=\alpha/10$, so when $\alpha = 0.05$, for example, $\beta=0.005$ is quite small. 

Still, this is not a precise statement about power gain. The computations required to compute the power gain analytically are messy, so instead we gauge the power gain via simulation. We sample data $X \sim N(\mu, I_n)$ for $n=10$ and attempt to reject the winning null $H_W: \mu_W \leq 0$ where $W = \argmin_{i \in [n]} 1 - \Phi(X_i)$ is the index of the smallest p-value $p_i = 1-\Phi(X_i)$.  

We choose $n=10$ because it is a reasonably small dimension size where one may apply hybrid inference (e.g., the main example from \cite{Andrews2023} has $n=13$). Let $R$ denote the index of the runner-up (second smallest p-value). For the dimensions  $n=10, 100, 1000$, \Cref{fig:hybrid_union} 
compares the distance from the winner $X_W$ the hybrid LCB and the union bound LCB. As illustrated in the plot, the benefit of hybrid inference dissipates as the dimension of the problem increases. This is because conditioning on the ``good event'' has less and less of an effect as $n$ grows. 

Our simulation results indicate that hybrid inference typically results in a fairly small power gain. We consider two simulated settings: \newline 

\noindent \textbf{Needle in a haystack: } First, we consider a needle in the haystack setting where $\mu_1 > 0$ and all the other $\mu_i$ for $ i\neq 1$ are set to $\mu_2$. We try $\mu_2 = -2, 0, 2$. The power comparison is ploted in \Cref{fig:hybrid_union_power}. Whenever we truly have a needle in the haystack problem, i.e., $\mu_1 > \mu_2$, hybrid inference results in essentially no power gain. The only setting where we see some gain (up to around $0.05$) is when $\mu_2 > \mu_1$. In this setting  we actually have a dense alternative (many small signals). We expect the top two p-values to be close to each other, so conditional methods should perform poorly. The union bound approach indeed performs essentially identically to the level $\beta$ classical test (not pictured). Hybrid, however, manages to eke out some additional power. Both methods pale in comparison to the level $\alpha$ classical test however, which would achieve power $>0.95$ throughout the whole plot (not pictured). For various values of $\sigma_1$ and $\sigma_2$, which we assume are known, we also tried re-running the experiments with  $X_1 \sim N(\mu_1, \sigma_1^2)$ and $X_i \sim N(\mu_2, \sigma_2^2)$ when $i > 1$. The results were not appreciably different.  \newline 

\noindent \textbf{Two possible signals: } Seeing as the hybrid and union bound approaches both reject based on the winning and running-up p-value, we ran a simulation for all pairs $\mu_1, \mu_2 \in \{-3, -2.9, \dots, 2.9, 3 \}$ with $\mu_1 > \mu_2$ and $\mu_1 > 0$. We forcibly set $\mu_i = -\infty$ for $i > 2$. For each setting we ran $N=10000$ to get an empirical estimate of power for each method. Across the $1492$ simulated settings, the median empirical power increase from hybrid was $\approx 0.003$, the $90$th percentile empirical power increase was $\approx 0.023$ and the maximum empirical power increase was $\approx 0.042$. As the results indicate, the power increases from hybrid were negligible for most settings. We also re-ran the same simulations but sampled $X_1 \sim N(\mu, \sigma_1^2)$ and $X_2 \sim N(\mu, \sigma_2^2)$ for various values of $\sigma_1$ and $\sigma_2$, which we assume are known. The results were not appreciably different, and, if anything, the difference in power was notably smaller for some values of $\sigma_1$ and $\sigma_2$.  

\begin{figure}
    \centering
    \hspace{-0.035\textwidth}
    \begin{minipage}{0.32\textwidth}
        \centering
        \includegraphics[width=\textwidth]{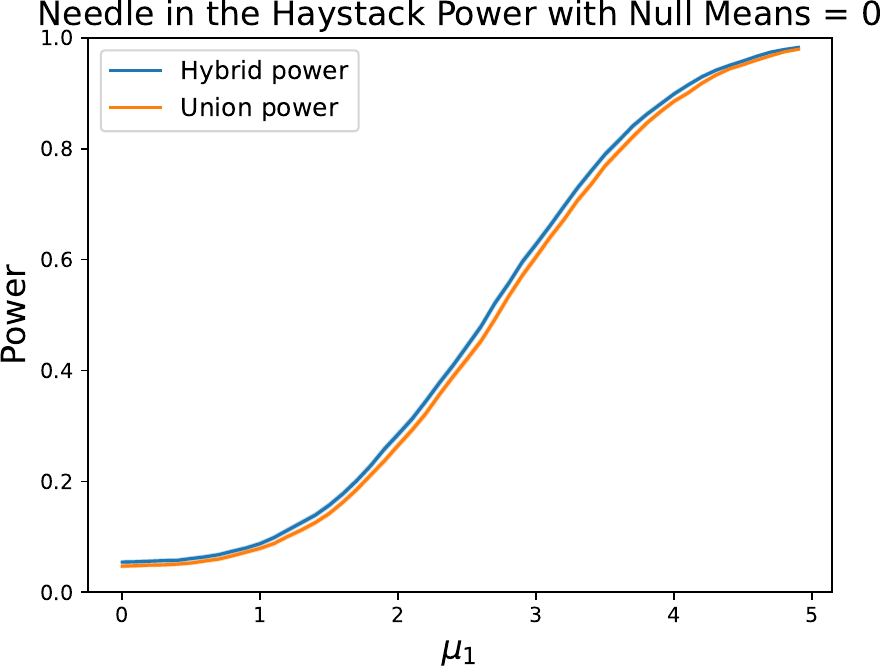}
        \caption*{(a) $\mu_2 = -2$}
    \end{minipage}
    \hfill
    \hspace{0.01\textwidth}
    \begin{minipage}{0.32\textwidth}
        \centering
        \includegraphics[width=\textwidth]{fig/hybrid_vs_union_null=0.pdf}
        \caption*{(b) $\mu_2 = 0$}
    \end{minipage}
    \hfill
    \hspace{0.01\textwidth}
    \begin{minipage}{0.32\textwidth}
        \centering
        \includegraphics[width=\textwidth]{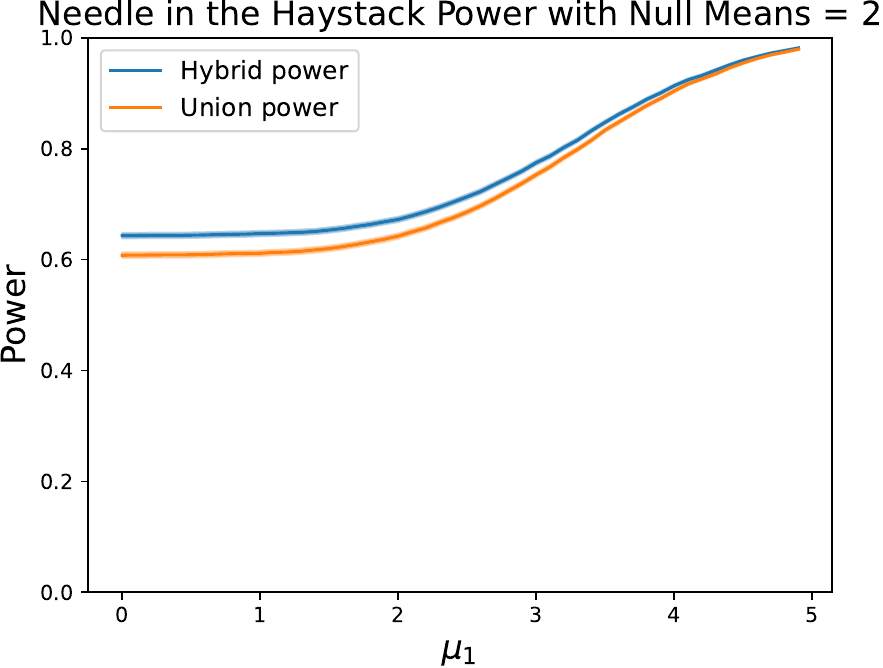}
        \caption*{(c) $\mu_2=2$}
    \end{minipage}
    \caption{ For $\mu_2=0$ (left), $\mu_2=-0.5$ (middle), and $\mu_2=-1$ (right) the empirical power over $N=10^4$ trials of the hybrid inference approach versus the union bound approach for the needle in the haystack alternative. One standard error bands are also plotted. For the most part, they are so small that they are hardly visible.  }
    \label{fig:hybrid_union_power}
\end{figure}

\subsection{Number of ties in rank verification}
\label{sec:ties_appdx}

Considering a random vector $X \in \R^n$ let 
\begin{equation}
    Y_i = \frac{X_i - X_j}{2}, \qquad  Y_j = \frac{X_i + X_j}{2}, \qquad  Y_{\ell} = X_{\ell} \text{ for } \ell \neq i, j,
\end{equation}
Suppose that $X_i \geq X_k$ for all $k \neq i$, and $X_i = X_{\ell}$ for some $\ell \neq i$ (i.e., there is at least one tie). We count the number of ties in for the winner in three different cases.

\begin{itemize}
    \item Suppose $Y_j > \max_{k \neq i, j} Y_k$. If any of the $X_k$ for $k \neq i, j$ were equal to $X_i$, then we would have 
    \begin{equation*}
       Y_k = X_k =  \frac{X_i + X_k }{2} \geq \frac{X_i + X_j}{2} = Y_j
    \end{equation*}
    which would be a contradiction. Thus, if there is a tie, the only possible tie is $X_j$. Since, 
    \begin{equation*}
        1 + |\{ \ell \neq i : Y_{\ell} = \max_{k \neq i} Y_k  \}| = 2
    \end{equation*}
    in this case, we are done. 
    \item Suppose that $Y_j < \max_{k \neq i, j} Y_k$. Then some $X_k$ for $k \neq i, j$ must be strictly than larger $X_j$. Otherwise we would have for every $k \neq i, j$ that  
    \begin{equation*}
        Y_j = \frac{X_i + X_j}{2} \geq  \frac{X_i + X_k}{2} \geq X_k = Y_k,
    \end{equation*}
     which is a contradiction. Thus, if there is a tie, the number of ties is one plus the number of $X_k$ for $k \neq i, j$ that are equal to each other. In this case, this matches
    \begin{equation*}
        1 + |\{ \ell \neq i : Y_{\ell} = \max_{k \neq i} Y_k  \}|,
    \end{equation*}
    so we are done.
    \item Now suppose that $Y_j = \max_{k \neq i, j} Y_k$. In this case we must have $X_j = X_i$. If not, then $X_j < X_i$ and  there must be some $\ell \neq i, j$ such that $X_i = X_{\ell}$, so 
    \begin{equation*}
        Y_j =  \frac{X_i + X_j}{2} < \frac{X_i + X_{\ell}}{2} = X_{\ell} = Y_{\ell},
    \end{equation*}
    which is a contradiction. Thus $Y_j = X_i$ in this case, and the number of ties is therefore clearly 
    \begin{equation*}
        1 + |\{ \ell \neq i : Y_{\ell} = \max_{k \neq i} Y_k  \}|.
    \end{equation*}
\end{itemize}

\subsection{Selective p-value for rank verification in exponential families}
\label{sec:rank_verficiation_adj_appdx}

Suppose $p$ is a p-value for the null $H_0$ that is selectively dominant given $Z$ and we select $p$ to use for inference according to the selection function 
    \begin{equation*}
        s(x, z) = 
        \begin{cases} 
        1 & \text{if } x < q^+(z), \\
        \frac{1}{N(z)} & \text{if } x \in [q^+(z), q(z)], \\
        0 & \text{otherwise},
        \end{cases},
    \end{equation*}
Then the adjusted p-value \eqref{eq:adjustment} is given by 
\begin{equation*}
    p_{sel} = \frac{\int_0^p s(x, Z) dx }{\int_0^1 s(x, Z)dx} = 
        \begin{cases} 
        \frac{p}{q^{+}(Z) + \frac{1}{N(Z)}(q(Z) - q^{+}(Z)) }  & \text{if } p  < q^+(Z), \\
        \frac{q^{+}(Z)+ \frac{1}{N(Z)}(p - q^{+}(Z))}{q^{+}(Z) + \frac{1}{N(Z)}(q(Z) - q^{+}(Z)) } & \text{if } p \in [q^+(Z), q(Z)], \\
        \end{cases},
\end{equation*}
which can be re-written as 
\begin{equation*}
   p_{sel} = \frac{p - (1-\frac{1}{N(Z)})(p - q^+(Z))_+ }{q^+(Z) + \frac{1}{N(Z)}(q^+(Z) - q(Z))}.
\end{equation*}
This is sufficient to imply the claim from the main text.

\subsection{Rank verification warm-up additional details}
\label{sec:rank_verification_warm_up_appdx}

\begin{example}[Rank verification in a simple case]
    \label{exm:rank_verification}
    Suppose that $p$ is a selectively dominant p-value for testing the null $H_0$, but we only choose to test $H_0$ when $p < 1/2$. Applying our framework with the p-value $p$ and selection function $s(x) = I(x < 1/2)$, \Cref{thm:adjustment} tells us that we control selective Type I error if we reject according to the selective p-value $p_{sel} = 2p$:
    \begin{equation}
        \label{eq:rank_verification_error_control}
        P_{H_{0}}\left(p_ \leq \frac{\alpha}{2} | S = 1\right) = P_{H_{0}}\left(2p \leq \alpha | S = 1\right) \leq \alpha.
    \end{equation} 
    
    Now, consider data $X_1 \sim N(\mu_1, 1/\sqrt{2})$ and $X_2 \sim N(\mu_2, 1/\sqrt{2})$. For the one-sided nulls $H_{0, ij}: \mu_i \leq \mu_j$ the UMP p-values $p_{ij} = 1 - \Phi(X_i - X_j)$ are selectively dominant by \Cref{exm:mlr}. Denoting the winner $W = \argmax_{j = 1, 2} X_j$ and runner-up $R = \argmin_{j = 1, 2} X_j$, it is now clear rejecting the data-dependent null $H_{0, WR}: \mu_W \leq \mu_R$ when $p_{WR} \leq \alpha/2$ maintains Type I error control both conditionally on $W$ and marginally. If $H_{0, ij}$ is not true, then trivially $P(\text{falsely reject } H_{0, WR} \mid W = i) = 0 \leq \alpha$. For the case that $H_{0, ij}$ is true, the event $W=i$ is identical to the event $p_{ij} < 1/2$, and hence is the same event as selecting $p_{ij}$ for inference in \eqref{eq:rank_verification_error_control}. Therefore, in this case,
    \begin{align*}
        P(\text{falsely reject } H_{0, WR} \mid W = i) &= P\left(p_{WR} \leq \frac{\alpha}{2}  \mid W = i\right)\\
        &= P\left(p_{ij} \leq \frac{\alpha}{2}  \mid W = i \right) \\
        &\leq \alpha.
    \end{align*}
    This establishes error control conditional on $W$. Marginal error control then follows from the law of total probability:
    \begin{align*}
        P(\text{falsely reject } H_{0, WR}) &= \sum_{i=1, 2} P(W=i) P(\text{falsely reject } H_{0,ij} \mid W = i) \\
                                          &\leq \alpha \sum_{i = 1, 2} P(W=i)\\
                                          &=\alpha. 
    \end{align*}
    If $\mu_1 = \mu_2$ then the inequalities become equalities and our error control is tight. 
    \end{example}

\subsection{Rank verification additional details}
\label{sec:rank_verification_appdx}

\begin{example}[Rank verification in exponential families]
    \label{exm:rank_verification_exp_fam}
    Suppose $p$ is a p-value for the null $H_0$ that is selectively dominant given $Z$. If we select $p$ to use for inference according to the selection function 
    \begin{equation*}
        s(x, z) = 
        \begin{cases} 
        1 & \text{if } x < q^+(z), \\
        \frac{1}{N(z)} & \text{if } x \in [q^+(z), q(z)], \\
        0 & \text{otherwise},
        \end{cases},
    \end{equation*}
    where $N(z)> 1$ and $0 \leq q^{+}(z) \leq q(z) \leq 1$ are known functions of $z$, then the selective p-value from \eqref{eq:adjustment} turns out to be (see \Cref{sec:rank_verficiation_adj_appdx} for computations)
    \begin{equation}
            \label{eq:tie_adj_p_val}
            p_{sel} = \frac{p - (1-\frac{1}{N(Z)})(p - q^+(Z))_+ }{q^+(Z) + \frac{1}{N(Z)}(q(Z) - q^+(Z))}.
     \end{equation}
    Therefore, \Cref{thm:adjustment} tells us that rejecting when \eqref{eq:tie_adj_p_val} is at most $\alpha$ is a selective Type I error controlling procedure:
    \begin{equation}
        \label{eq:tie_tool}
        P_{H_0}\left( \frac{p - (1-\frac{1}{N(Z)})(p - q^+(Z))_+ }{q^+(Z) + \frac{1}{N(Z)}(q(Z) - q^+(Z))}.  \leq \alpha \mid S = 1\right) \leq \alpha.  
    \end{equation} 
    
    Now, suppose we observe $X$ drawn from the exponential family \eqref{eq:exp_fam} and let $W$ be the index of the largest $X_i$ (with ties broken randomly). For $i \neq j$, \Cref{exm:exp_fam} tells us that the UMPU p-value $p = p_{ij}$ from \eqref{eq:umpu_rank_verification} for the null $H_{0, ij}: \theta_i \leq \theta_j $ is selectively dominant given the transformed nuisance statistics $Y_{-i}$ from \eqref{eq:reparam}. Taking $q_{ij}^+(Y_{-i})$, $q_{ij}(Y_{-i})$, $N_i(Y_{-i})$, and $f$ from  \eqref{eq:rank_verification_lower}, \eqref{eq:rank_verification_upper}, \eqref{eq:rank_verification_num_ties}, it is now easy to show that rejecting the data-dependent null $H_{0, Wj}$ for $j \neq W$ when 
    \begin{equation}
      \frac{p_{Wj} - (1-\frac{1}{N_W})(p_{Wj} - q_{Wj}^+)_+ }{q_{Wj}^+ + \frac{1}{N_W}(q_{Wj} - q^+_{Wj})} \leq \alpha
    \end{equation} 
    controls Type I error conditionally on $W$. 
    
    For $i \neq j$, if $H_{0, ij}$ is false, then trivially $P(\text{falsely reject } H_{0, Wj}| W= i)$. For the case that $H_{0, ij}$ is true, the event $W=i$ is the same event as selecting $p_{ij}$ for inference in \eqref{eq:tie_adj_p_val}, so 
    \begin{align*}
        P(\text{falsely reject } H_{0, Wj}| W= i) &= P\left(\frac{p_{Wj} - (1-\frac{1}{N_W})(p_{Wj} - q_{Wj}^+)_+ }{q_{Wj}^+ + \frac{1}{N_W}(q_{Wj} - q^+_{Wj})} \leq \alpha \;\middle|\; W= i\right) \\
        &= P\left(\frac{p_{ij} - (1-\frac{1}{N_i})(p_{ij} - q_{ij}^+)_+ }{q_{ij}^+ + \frac{1}{N_i}(q_{ij} - q^+_{ij})} \leq \alpha \;\middle|\; W = i\right), \\
        &\leq \alpha. 
    \end{align*}
    which is sufficient to imply conditional error control. 

    Now suppose we reject the data-dependent null $\cup_{j \neq W} H_{0, Wj}$ when we reject all of the individual nulls $H_{0, Wj}$ for $j \neq W$. It is straightforward to see that this will control Type I error both conditionally on $W$ and marginally. If $i$ is such that $\theta_i > \theta_j$ for all $j \neq i$, then trivially $ P(\text{falsely reject } \cup_{j \neq W} H_{0, Wj}| W= i) = 0 $. If this is not true, then there exists some $k \neq i$ such that $\theta_k \geq \theta_i$, so 
    \begin{align*}
        P(\text{falsely reject } \cup_{j \neq W} H_{0, Wj}| W= i) \leq P(\text{falsely reject } H_{0, Wk}| W= i) \leq \alpha.
    \end{align*} 
    Again, marginal error control follows from the usual law of total probability argument:
    \begin{align*}
        P(\text{falsely reject } \cup_{j \neq W} H_{0, Wj}) &= \sum_{i=1}^n P(W= i)P(\text{falsely reject } \cup_{j \neq W} H_{0, Wj} |W = i) \\
        &\leq \alpha \sum_{i=1}^n P(W=i)\\
        &=\alpha.
    \end{align*}
\end{example}

\section{Selective dominance and one-sided testing}
\label{sec:one_sided_appdx}

In this appendix, we establish the selective dominance property for UMP p-values in MLR families and UMPU p-values in exponential families. We also show that in these cases, the adjusted p-value from \Cref{thm:adjustment} is monotone in the parameter, under suitable conditions. Throughout the appendix, we draw from the discussion and proof strategy in Appendix B.1 of \cite{Lei}.

\subsection{MLR Families}
\label{sec:one_sided_mlr_appdx}

We consider a parametric family $P_{\theta}$ parameterized by a real parameter $\theta \in R$ such that each $P_{\theta}$ has density $p_{\theta}(x)$ with respect to some carrier measure $\mu$. We will suppose that these densities share support (i.e., for any two $\theta$ and $\theta'$ we have $p_{\theta}(x) > 0 \iff p_{\theta'}(x)> 0$). Further, we suppose that for any $\theta \leq \theta'$, the likelihood ratio $p_{\theta'}(x)/p_{\theta}(x)$ is a monotone non-decreasing function of some real-valued function $T(x)$ on this support (i.e., for any $x_1 \leq x_2$ with $p_{\theta}(x_1), p_{\theta}(x_2) > 0$, we have $p_{\theta'}(x_1)/p_{\theta}(x_1) \leq p_{\theta'}(x_2)/p_{\theta}(x_2)$). 

Recall that for a testing problem, the critical function $\phi(x)$ (see \cite[Section 3.1]{Lehmann}) tells us the probability of rejecting the null having observed data $x$, so $\phi(X) = P(\text{reject } H_0 | X)$. From Theorem 3.4.1 of \cite{Lehmann} we know the test governed by the critical function 
\begin{equation}
    \label{eq:mlr_test}
    \phi(x) = \begin{cases}
        1 &\text{if } T(x) > C  \\
        \gamma &\text{if } T(x) = C  \\
        0 & \text{otherwise }
    \end{cases}
\end{equation}
is UMP for testing $H_0 : \theta \leq \theta_0$ against the alternatives $H_a : \theta > \theta_0 $ so long $C$ and $\gamma$ satisfy
\begin{equation}
    \label{eq:constraint}
    P_{\theta_0}(T(X) > C) + \gamma P_{\theta_0}(T(X) = C) = \alpha.
\end{equation}
Denote the left-continuous survival function of $T(X)$ and its right-hand limit under $P_{\theta_0}$ as
\begin{equation*}
    G(t) = P_{\theta_0}(T(X) \geq t) \qquad G^+(t) = \lim_{u \downarrow t} G(u).
\end{equation*}
Since $G(t)$ is a monotone non-increasing function, we can also define its generalized inverse 

\begin{equation*}
    G^{-1}(z) = \inf \{ t : G(t) \leq z\},
\end{equation*}

\Cref{lem:setting_constants} gives a natural way to set $C$ and $\gamma$ in \Cref{eq:mlr_test} to get an UMP test. 

\begin{lemma}[An UMP test]
    \label{lem:setting_constants} Adopting the convention that $0/0 = 0$, taking $C = G^{-1}(\alpha)$ and $\gamma = (\alpha - G^+(C))/(G(C) - G^+(C))$
    in \Cref{eq:mlr_test} gives an UMP test.
\end{lemma}

\begin{proof}
    At continuity points $t$ of $G(\cdot)$, we have $G(G^{-1}(t)) = t$ and also $P(T(X) = t) = 0$. Thus, if $C = G^{-1}(\alpha)$ is a continuity point of $G(\cdot)$, then $G^+(C) = P_{\theta_0}(T(X) > C) = P_{\theta_0}(T(X) \geq C) = G(C) = \alpha$ and $\gamma=0$, so the constraint \eqref{eq:constraint} is immediately satisfied. 

    If $C = G^{-1}(\alpha)$ is not a continuity point of $G(\cdot)$, then $G(C) - G^+(C) > 0$ and the constraint \eqref{eq:constraint} is also immediately satisfied. To ensure that we still have a valid test, however, we need $\gamma \in [0, 1]$. This is true so long as $\alpha \in [G^{+}(C), G(C)]$. We know that $G(t) \leq \alpha$ for any $t > C$, so $G^{+}(C) \leq \alpha$. If $G(C) < \alpha$, then we could find some $t^{-} < C$ such that $G(C) < \alpha$ by left-continuity, but this would contradict that $C = G^{-1}(\alpha)$, finishing the proof. 
\end{proof}

Letting $U_{aux} \sim \text{Unif}([0, 1])$ be auxiliary randomness that is independent of $X$, a simple way to instantiate the test from \Cref{lem:setting_constants} is to reject whenever $T(X) > C$ or when $T(X) = C$ and $U_{aux} \leq \gamma$. \Cref{lem:fuzzy} explains how this is the same as rejecting when the p-value, termed as a fuzzy p-value in \cite{Geyer}, 
\begin{equation}
    \label{eq:mlr_p_val}
    p = G^+(T(X)) + U_{aux}(G(T(X)) - G^+(T(X))) 
\end{equation}
is at most $\alpha$. 

\begin{lemma}[Fuzzy p-value is UMP]
    \label{lem:fuzzy}
    Rejecting $H_0: \theta \leq \theta_0$ when the fuzzy p-value \eqref{eq:mlr_p_val} is at most $\alpha$ instantiates the test from \Cref{lem:setting_constants}, and is therefore UMP. 
\end{lemma}

\begin{proof} We rewrite 
    \begin{equation*}
        p = (1 - U_{aux})G^+(T(X)) + U_{aux} G(T(X))
    \end{equation*}
    and consider four  cases. 
    \begin{itemize}
        \item If $t < G^{-1}(\alpha)$ then we can find some $t^+ > t$ such that $G(t^+) > \alpha$. Thus $G(t) \geq G^{+}(t) > \alpha$. So $p > \alpha$ whenever $T(X) < G^{-1}(\alpha)$ 
        \item If $t = G^{-1}(\alpha)$ and $G^{-1}(\alpha)$ is a continuity point of $G(\cdot)$, then $G^{+}(t) = G(t) = \alpha$. Thus, in this case $p \leq \alpha$ whenever $T(X) = G^{-1}(\alpha)$ and $U_{aux} \leq  \frac{\alpha - G^+(G^{-1}(\alpha))}{G(G^{-1}(\alpha)) - G^+(G^{-1}(\alpha))} = \infty$. 
        \item If $t = G^{-1}(\alpha)$ and $G^{-1}(\alpha)$ is not a continuity point of $G(\cdot)$, then we must have that $G(t) - G^+(t)  > 0$. Also by right continuity we have $G(t) \geq \alpha$ and by how $G^{-1}(\cdot)$ is defined we have $G^+(t) \leq \alpha$. In this case $p \leq \alpha$ also whenever $T(X) = G^{-1}(\alpha)$ and $U_{aux} \leq  \frac{\alpha - G^+(G^{-1}(\alpha))}{G(G^{-1}(\alpha)) - G^+(G^{-1}(\alpha))}$. 
        \item If $t > G^{-1}(\alpha)$ then $G^+(t) \leq G(t) \leq \alpha$. So $p \leq \alpha$ whenever $T(X) > G^{-1}(\alpha)$. 
    \end{itemize}

    This implies the following set equalities:

    \begin{align*}
        \{p \leq \alpha \} &= \{T(X) > G^{-1}(\alpha) \} \cup \left\{T(X) =  G^{-1}(\alpha), U_{aux} \leq \frac{\alpha - G^+(G^{-1}(\alpha))}{G(G^{-1}(\alpha)) - G^+(G^{-1}(\alpha))} \right\}\\
                           &= \{T(X) > C \} \cup \left\{T(X) =  C, U_{aux} \leq \gamma \right\}\\
    \end{align*}

\end{proof}

\Cref{lem:uniform} shows that $p \sim \text{Unif}([0, 1])$ under $P_{\theta_0}$, which will be useful for us later. 

\begin{lemma}[Fuzzy p-value is uniform at null boundary]
    \label{lem:uniform}
    Under $P_{\theta_0}$, the p-value \eqref{eq:mlr_p_val} has a $\text{Unif}([0, 1])$ distribution.
\end{lemma}

\begin{proof}

Using the set equality from \Cref{lem:fuzzy} but replacing $\alpha$ with $z \in (0, 1)$, we find

\begin{align*}
        &P_{\theta_0}(G^+(T(X)) + U_{aux}(G(T(X)) - G^+(T(X))) \leq z)\\
        &= P_{\theta_0}(T(X) > G^{-1}(z)) + P_{\theta_0}\left(T(X) = G^{-1}(z),  U \leq \frac{z - G^+(G^{-1}(z))}{G(G^{-1}(z)) - G^{+}(G^{-1}(z))}\right)\\
        &= P_{\theta_0}(T(X) > G^{-1}(z)) + P_{\theta_0}(T(X) = G^{-1}(z)) P_{\theta_0}\left(U \leq \frac{z - G^+(G^{-1}(z))}{G(G^{-1}(z)) - G^{+}(G^{-1}(z))}\right)\\
        &=G^+(G^{-1}(z)) + (G(G^{-1}(z)) - G^{+}(G^{-1}(z))) \cdot \frac{z - G^+(G^{-1}(z))}{G(G^{-1}(z)) - G^{+}(G^{-1}(z))}\\
        &= z.
    \end{align*}

\end{proof}

Now we can show that $p$ is a selectively dominant p-value for testing the null $H_0: \theta \leq \theta_0$. In what follows, we consider some fixed $\theta \leq \theta_0$ and prove some facts that allow us to relate the distribution of $T(X)$ under $P_{\theta_0}$ to its distribution under $P_{\theta}$.

\begin{lemma}[Distribution of $T(X)$]
    \label{lem:stan_machine}
    Let $g_{\theta}(T(x))$ be a non-increasing function that equals the likelihood ratio $p_{\theta}(x)/p_{\theta_0}(x)$ on the support and 
    \begin{equation*}
        \nu(A) = \int I(T(x) \in A) p_{\theta_0}(x) \mu(dx)
    \end{equation*}
    be the measure of $T(X)$ under $X \sim P_{\theta_0}$. Then 
    \begin{equation*}
        P_{\theta}(T(X) \in A) = \int_A g_{\theta}(t) \nu(dt)
    \end{equation*}
\end{lemma}

\begin{proof}
We know that 
\begin{align*}
    P_{\theta}(T(X) \in A) = \int I(T(x) \in A) p_{\theta_0}(x)g_{\theta}(T(x)) \mu(dx), 
\end{align*}
so we need to show that 
\begin{equation}
    \label{eq:stan_machine_to_show}
    \int I(T(x) \in A) p_{\theta_0}(x) g_{\theta}(T(x)) \mu(dx) =  \int_A g_{\theta}(t) \nu(dt)
\end{equation}
If $g_{\theta}(T(x)) = I(T(x) \in A')$ happens to be an indicator then \eqref{eq:stan_machine_to_show} holds. Therefore, we can apply the standard machine (see the discussion after Equation 42 in \cite{Lei}) to show that \eqref{eq:stan_machine_to_show} holds for all non-negative functions $g_{\theta}(\cdot)$.  
\end{proof}

\begin{lemma}[Distribution of $(T(X), U_{aux})$]
    \label{lem:pi_lambda}
    If $\omega$ denotes the product measure of $\nu$ and Lebesgue measure $\lambda$ on $[0, 1]$, i.e., the distribution of $(T(X), U_{aux})$ under $P_{\theta_0}$, then 
    \begin{equation}
    \label{eq:joint_dist}
        P_{\theta}((T(X), U) \in B) = \int_{B} g_{\theta}(t) \omega(dt, du)
    \end{equation}
\end{lemma}

\begin{proof}
    We will first argue that \eqref{eq:joint_dist} holds for any $B$ which is a product set $A_1 \times A_2$. We can further reduce to the case that $g_{\theta}(T(x)) = I(T(x) \in A_1')$ is an indicator. Then we see using our previous lemma that 
    \begin{align*}
        P_{\theta}((T(X), U_{aux}) \in A_1 \times A_2) &= P_{\theta}(T(X) \in A_1)P(U_{aux} \in A_2)\\
                                          &= \int_{A_1} g_{\theta}(t) \nu(dt) \cdot \int_{A_2} \lambda(du)\\
                                          &= \int_{A_1 \cap A_1'} \nu(dt) \cdot \int_{A_2} \lambda(du)\\
                                          &= \int_{A_1 \cap A_1' \times A_2} \omega(dt, du)\\
                                          &= \int_{A_1 \times A_2} g_{\theta}(t)\omega(dt, du)
    \end{align*}
    To handle the case of general $g_{\theta}(\cdot)$ we can again simply apply the standard machine. 
    
    The full result then follows from an application of the $\pi-\lambda$ theorem: the set of $B$ for which \eqref{eq:joint_dist} holds is a $\lambda$-system, and \eqref{eq:joint_dist} holds for every set in the $\pi$ system of all product sets $B = A_1 \times A_2$. 
\end{proof}

Note that our p-value $p$ is a deterministic function of $T(X)$ and $U_{aux}$:
\begin{equation*}
    p = m(T(X), U_{aux}) \qquad m(t, u) = G^+(t) + u(G(t) - G^+(t)).
\end{equation*}
As such, we sometimes write our selection function as a function of $T(X)$ and $U_{aux}$:
\begin{equation*}
    s(t, u) = s(m(t, u)). 
\end{equation*}
We use this abuse of notation in our next lemma, which characterizes the conditional distribution of $T(X)$ given selection.

\begin{lemma}[Distribution of $(T(X), U_{aux})$ given selection]
    \label{lem:cond_dist}
    For any selection function $s(x)$ under which $p$ has a positive probability of selection under $P_{\theta}$, 
    \begin{equation*}
        P_{\theta}((T(X), U_{aux}) \in B| S = 1 ) = \frac{\int_{B} g_{\theta}(t)s(t, u) \omega(dt, du)}{ \int g_{\theta}(t) s(t, u) \omega(dt, du) }
    \end{equation*}
\end{lemma}

\begin{proof} First note that
    \begin{equation*}
        P_{\theta}( (T(X), U_{aux}) \in B | S= 1 ) = \frac{P_{\theta}((T(X), U_{aux}) \in B, S = 1) }{P_{\theta}(S = 1)}.
    \end{equation*} 
    Thus it suffices to show for any set $B$ that 
    \begin{equation*}
        P_{\theta}( (T(X), U_{aux}) \in B,  S= 1 ) = \int_B g_{\theta}(t) s(t, u) \omega(dt, du). 
    \end{equation*} 
    By the definition of conditional expectation
    \begin{align*}
        P_{\theta}( (T(X), U_{aux}) \in B,  S= 1 ) &= E_{\theta}[E_{\theta}[ I(S=1) \mid  T(X), U_{aux}] I((T(X), U_{aux}) \in B) ]\\
                                                   &= E_{\theta}[ s(T(X), U_{aux}) I((T(X), U_{aux}) \in B) ]
    \end{align*}
    If $s(t, u) = I((t, u) \in B)$ is an indicator function, then the result is implied by our previous lemma. We again get the result for general selection functions $s(t, u)$ by applying the standard machine. 
\end{proof}

With these lemmas under our belt, we can show \Cref{prop:mlr_sel_dom}, the main result of this sub-section. Since $p \sim_{P_{\theta_0}} \text{Unif}([0, 1])$ by \Cref{lem:uniform}, this proposition is sufficient to imply selective dominance. 

\begin{proposition}
    \label{prop:mlr_sel_dom}
    For any selection function $s(x)$ for which $p$ has positive probability of selection under $P_{\theta}$,  
    \begin{equation*}
        P_{\theta}(p \leq z | S = 1)  \leq P_{\theta_0}(p \leq z | S = 1).
    \end{equation*}
\end{proposition}

\begin{proof}

First we show that $P_{\theta}(S= 1) > 0 \implies P_{\theta_0}(S=1) > 0$. We show the contrapositive. If 
\begin{align*}
    0 = P_{\theta_0}(S=1) &= E_{\theta_0}[E_{\theta_0}[S | p]] \\
                          &= E_{\theta_0}[s(p)]\\
                          &= E_{\theta_0}[s(T(X), U_{aux})],
\end{align*} 
then $s(T(X), U_{aux}) = 0$ a.e. under $P_{\theta_0}$. Therefore,
\begin{align*}
    P_{\theta}(S=1) = E_{\theta}[s(T(X), U_{aux})] = E_{\theta_0}[g_{\theta}(T(X))s(T(X), U_{aux})] = 0.
\end{align*}

Now, fix $z \in (0, 1)$. If $z$ is such that $P_{\theta}(p \leq z | S = 1)  = 0$ then the desired inequality is trivial. To handle the non-trivial case, we recall \Cref{lem:fuzzy} and note three facts:
\begin{itemize}
    \item If $(t, u) \in m^{-1}([0, z])$ then $t \geq G^{-1}(z)$,
    \item If $(t, u) \in m^{-1}((z, 1])$ then $t \leq G^{-1}(z)$,
    \item The sets $m^{-1}([0, z])$ and $m^{-1}((z, 1])$ are disjoint.
\end{itemize}
Thus,
\begin{align*}
     \frac{1}{P_{\theta}(p \leq z | S = 1)} &= \frac{\int_{m^{-1}([0, 1])} g_{\theta}(t) s(t, u) \omega(dt, du) }{\int_{m^{-1}([0, z])} g_{\theta}(t) s(t, u) \omega(dt, du) }\\
                                            &= \frac{\int_{m^{-1}([0, z])} g_{\theta}(t) s(t, u) \omega(dt, du) + \int_{m^{-1}((z, 1])} g_{\theta}(t)  s(t, u)\omega(dt, du) }{\int_{m^{-1}([0, z])} g_{\theta}(t) s(t, u) \omega(dt, du) }\\
                                            &= 1 + \frac{\int_{m^{-1}((z, 1])} g_{\theta}(t) s(t, u) \omega(dt, du)}{\int_{m^{-1}([0, z])} g_{\theta}(t) s(t, u) \omega(dt, du)}\\
                                            &\geq 1 + \frac{g_{\theta}(G^{-1}(z))  \int_{m^{-1}((z, 1])}  s(t, u) \omega(dt, du)}{g_{\theta}(G^{-1}(z)) \int_{m^{-1}([0, z])} s(t, u) \omega(dt, du)} \\
                                            &= 1 + \frac{\int_{m^{-1}((z, 1])} s(t, u) \omega(dt, du)}{ \int_{m^{-1}([0, z])} s(t, u) \omega(dt, du)} \\
                                            &= \frac{1}{P_{\theta_0}(p \leq z | S = 1)}, 
\end{align*}
where to finish we have noted that $g_{\theta_0}(t) = 1$ almost everywhere in the measure $\omega$. 
\end{proof}

\subsection{Exponential families}
\label{sec:one_sided_exp_fam_appdx}

Suppose we observe data $X \in \R^m$ from an exponential family $P_{\theta}$ parameterized by $\theta \in \R^n$ i.e., under $P_{\theta}$ the data $X$ has density  
\begin{equation*}
    g_{\theta}(x) = \exp( \theta_1 T_1(x) + \dots + \theta_n T_n(x) - \psi(\theta) ) g(x) 
\end{equation*}
with respect to some carrier measure $\mu$. We consider the problem of testing $H_0: \theta_i \leq \theta_{0, i}$.

The UMPU test for $H_0: \theta_i \leq \theta_{0, i}$ is valid conditional on $T_{-i}(X)$. More specifically, Theorem 4.4.1 of \cite{Lehmann} tells us that any test of the form 

\begin{equation*}
    \label{eq:exp_fam_test}
    \phi(t_i, t_{-i}) = \begin{cases}
        1 &\text{if } t_i > C_0(t_{-i})  \\
        \gamma(t_{-i}) &\text{if } t_i = C_0(t_{-i})   \\
        0 & \text{otherwise }
    \end{cases}
\end{equation*}
where the functions $\gamma(\cdot)$ and $C_0(\cdot)$ satisfy 
\begin{equation*}
    E_{\theta_i = \theta_{0, i}}[\phi(T_i(X), t_{-i}) | T_{-i}(X)] = \alpha \text{ a.e. under } P_{\theta_i = \theta_{0, i}}
\end{equation*}
is UMPU for testing $H_0: \theta_i \leq \theta_{0, i}$. Lemma 2.7.2 of \cite{Lehmann} tells us that the conditional distribution of $T_{i}(X)$ given $T_{-i}(X) = t_{-i}$ admits a density 
\begin{equation*}
    g_{\theta_i, t_{-i}}(t_i) = \exp(\theta_i t_i - \tilde{\psi}(\theta_i))  
\end{equation*}
with respect to some base measure $\mu_{t_{-i}}$. This density has an MLR in $t_i$ (to be specific, we are imagining observing $T_i(X)$ from its conditional distribution $T_{-i}(X)$, and the map $T(\cdot)$ from the previous sub-section is actually the identity). Hence, a concrete UMPU test is to just run our UMP test from the previous section using the conditional distribution given $T_{-i}(X) = t_{-i}$. In particular, our work from the previous section implies that it is UMPU to reject using the p-value

\begin{equation}
    \label{eq:ump_exp_fam}
    p = G^{+}(T_i(X)|T_{-i}(X)) + U_{aux}(G(T_i(X)|T_{-i}(X)) - G^{+}_i(T_i(X)|T_{-i}(X))),
\end{equation}
where $U_{aux}$ is an uniform random variable independent of the data and 
\begin{equation*}
    G(t_i | t_{-i}) = P_{\theta_0}(T_i(X) \geq t | T_{-i}(X) = t_{-i}) \qquad G^{+}(t_i | t_{-i}) = \lim_{u \downarrow t_i} G(u| t_{-i}).
\end{equation*}

We argue that this p-value is selectively dominant given $Z$. Fix a distribution $P$ in the null. If we consider some selection function $s(x, z)$ such that $P(S= 1) = 0$, then \eqref{eq:selective_dominance} holds trivially. Instead consider the case that $P(S= 1) > 0$ and let $f_{z}(x)$ denote the conditional PDF of $p$ given $Z=z$. \Cref{prop:mlr_sel_dom} along with \Cref{lem:uniform} tells us exactly that for $z$ such that $\int_0^1 s(x, z) f_z(x) dx > 0$,
\begin{equation*}
    \frac{\int_0^t s(x, z) f_z(x)}{\int_0^1 s(x, z) f_z(x)} \leq \frac{\int_0^t s(x, z) }{\int_0^1 s(x, z)}  \text{ for all } t \in [0, 1].
\end{equation*}
Arguments in \Cref{sec:adjustment_proof} and \Cref{sec:density_proof} then tell us that the following three facts are true a.e. under $P(\cdot | S=1)$: (1) $\int_0^1 s(x, Z) f_Z(x) dx > 0$ , (2) the left-hand side of the above equals the left-hand side of \eqref{eq:selective_dominance}, and (3) the right-hand side of the above equals the right-hand side of \eqref{eq:selective_dominance}. Combined these facts imply that $p$ is selectively dominant given $Z$.

\subsection{Monotonicity of selective MLR p-values}
\label{sec:one_sided_monotone_appdx}

In this sub-section, we consider data $(X, Z)$ where the conditional distribution $X | Z= z \sim P_{\theta, z}$ is parameterized by $\theta \in \R$ and has an MLR in $T(x)$. Because we do everything conditional on $Z=z$, without loss of generality we can just work with $X$ and understand that the results will hold a.e. over $Z$. Letting
\begin{equation*}
    G^{\theta_0}(t) = P_{\theta_0}(T(X) \geq t) \qquad G^{\theta_0, +}(t) = \lim_{u \downarrow t} G^{\theta_0}(u)
\end{equation*}
we let 
\begin{equation*}
    p^{\theta_0} = G^{\theta_0, +}(T(X)) + U_{aux}(G^{\theta_0}(T(X)) - G^{\theta_0, +}(T(X)))
\end{equation*}
be the UMP p-value for testing $H_{0}: \theta \leq \theta_0$. Again, $U_{aux}$ is a uniform random variable independent of the data. Let 
\begin{equation*}
    m^{\theta_0}(t, u) = G^{\theta_0, +}(t) + u(G^{\theta_0}(t) - G^{\theta_0, +}(t))
\end{equation*}
be the map such that $m^{\theta_0}(T(X), U) = p^{\theta_0}$. 

Considering a class of selection functions $s^{\theta_0}(x)$ such that $\int_0^1 s^{\theta_0}(x) > 0$ (when $Z$ is present we need $\int_0^1 s^{\theta_0}(x, Z) > 0$ almost surely), we want to show that the selective p-values from \Cref{thm:adjustment}, 
\begin{equation*}
    p^{\theta_0}_{sel} = \frac{\int_{0}^{p^{\theta_0}} s^{\theta_0}(x) dx}{\int_{0}^{1} s^{\theta_0}(x) dx},
\end{equation*}
are monotone non-decreasing in $\theta_0$. Specifically, we will show that this is true when the selection function $s^{\theta_0}(x)$ is independent of $\theta_0$ once written in terms of the data (i.e., selection can be stated in terms of the data without reference to the null parameter being tested). Formally, we establish monotonicity whenever there exists $\tilde{s}(\cdot, \cdot)$, independent of $\theta_0$, such that 
\begin{equation*}
    s^{\theta_0}(x) = \tilde{s}(t, u) \text{ for all } t, u \text{ with } x = m^{\theta_0}(t, u), 
\end{equation*}

Like before, we now consider some fixed $\theta \leq \theta_0$ and let $g_{\theta}(T(x))$ be the non-increasing function that equals the likelihood ratio $p_{\theta}(x)/p_{\theta_0}(x)$ on the support. The next lemma allows us to rewrite the selective p-value in a useful way. 

\begin{lemma} Letting $E_{r, s} = \{(t, u) : t > r \text{ or } t = r \text{ and } u \leq s \}$. For any $t$ and $u$ such that $m^{\theta}(t, u) = y$,
\begin{equation*}
    \int_0^y s^{\theta}(x) dx = \int_{E_{t, u}} g_{\theta}(t) \tilde{s}(t, u) \omega(dt, du). 
\end{equation*}
\end{lemma}
\begin{proof}
   First we handle the case that $s^{\theta}(x) = I(x \in B_{\theta})$ is an indicator of membership to some set. Then the left-hand side of the equation is the Lebesgue measure of the set $\{x : x \leq y \}$ intersected with $B_{\theta}$. By \Cref{lem:uniform} this is the same as the probability under $P_{\theta}$ of the p-value $p^{\theta}$ being at most $y$ and in $B_{\theta}$. By our choice of $t$ and $u$, the set difference between $E_{t, u}$ and the set $\{(t, u): m^{\theta_0}(t, u) \leq y \}$ is measure zero under $P_{\theta}$. Thus by \Cref{lem:pi_lambda} and how $\tilde{s}$ is defined, the right-hand side of the equation is also the probability  probability under $P_{\theta}$ of the p-value $p^{\theta}$ being at most $y$ and in $B_{\theta}$. We again get the result for general selection functions $s^{\theta}(x)$ by applying the standard machine. 
\end{proof}

According the the previous lemma, 
\begin{equation*}
    p^{\theta}_{sel} = \frac{\int_{0}^{p^{\theta}} s^{\theta}(x) dx}{\int_{0}^{1} s^{\theta}(x) dx} = \frac{\int_{E_{T(X), U_{aux}}} g_{\theta}(t) \tilde{s}(t, u) \omega(dt, du) }{\int g_{\theta}(t) \tilde{s}(t, u) \omega(dt, du) }
\end{equation*}
Hence it would suffice to show that for all $r$ and $s$, 
\begin{equation*}
  \frac{\int_{E_{r, s}} g_{\theta}(t) \tilde{s}(t, u) \omega(dt, du) }{\int g_{\theta}(t) \tilde{s}(t, u) \omega(dt, du) }
\end{equation*}
is monotone in non-decreasing in $\theta$. This is the subject of our next lemma.  

\begin{lemma}
    \label{prop:monotone_adjustment}
    For any $r$ and $s$, the quantity   
    \begin{equation*}
        \frac{\int_{E_{r, s}} g_{\theta}(t) \tilde{s}(t, u) \omega(dt, du) }{\int g_{\theta}(t) \tilde{s}(t, u) \omega(dt, du) }
      \end{equation*}
    is monotone non-decreasing in $\theta$. 
\end{lemma}

\begin{proof}
    The proof strategy is the same as our earlier results. We would like to show that 

    \begin{equation*}
        \frac{\int_{E_{r, s}} g_{\theta}(t) \tilde{s}(t, u) \omega(dt, du) }{\int g_{\theta}(t) \tilde{s}(t, u) \omega(dt, du) } \leq  \frac{\int_{E_{r, s}} g_{\theta_0}(t) \tilde{s}(t, u) \omega(dt, du) }{\int g_{\theta_0}(t) \tilde{s}(t, u) \omega(dt, du) } 
    \end{equation*}

    If the numerator of the left-hand side is zero then the inequality must hold. In the other case we see that 

    \begin{align*}
        \frac{\int g_{\theta}(t) \tilde{s}(t, u) \omega(dt, du)}{\int_{E_{r, s}} g_{\theta}(t) \tilde{s}(t, u) \omega(dt, du) } &= 1 + \frac{\int_{E_{r, s}^c} g_{\theta}(t) \tilde{s}(t, u) \omega(dt, du)}{\int_{E_{r, s}} g_{\theta}(t) \tilde{s}(t, u) \omega(dt, du) } \\
        &\geq 1 + \frac{ g_{\theta}(r) \int_{E_{r, s}^c} \tilde{s}(t, u) \omega(dt, du)}{g_{\theta}(r)\int_{E_{r, s}}  \tilde{s}(t, u) \omega(dt, du) } \\
        &= 1 + \frac{  \int_{E_{r, s}^c} \tilde{s}(t, u) \omega(dt, du)}{\int_{E_{r, s}}  \tilde{s}(t, u) \omega(dt, du) }\\
        &= \frac{\int g_{\theta_0}(t) \tilde{s}(t, u) \omega(dt, du)}{\int_{E_{r, s}} g_{\theta_0}(t) \tilde{s}(t, u) \omega(dt, du) }
    \end{align*}

    where we have noted that if $(t, u) \in E_{r, s}$ then $t \geq r$ and if $(t, u) \in E_{r, s}$ then $t \leq r$, and also that $g_{\theta_0}(t) = 1$ almost everywhere in the measure $\omega$.  
\end{proof}

\section{Selecting multiple p-values for inference}
\label{sec:multiple_p_vals_appdx}

In this appendix, we generalize our selective dominance framework to allow us to select multiple p-values for inference. 

Suppose we have $k$ p-values for the nulls $H_{0, i}$ that are selectively dominant given some common $Z$. We require that these p-values are conditionally independent given $Z$. Define $k$ binary selection random variables $S_j \in \{0, 1\}$, where $S_j = 1$ when $p_j$ is selected. The relationship between $p_j$, $Z$, and $S_j$ is governed by the selection functions \footnote{Again, formally $s(p_j, Z) = E[S | p_j, Z]$.} 
\begin{equation*}
    s_j(x, z) = p(S_j = 1 | p_j = x, Z=z)
\end{equation*} 
Furthermore, we demand that the different selections happen independently, i.e., 
\begin{equation}
    \label{eq:mult_indep}
    P(S_1 = 1, \dots, S_k = 1 \mid p_1, \dots, p_k, Z) = \prod_{j = 1}^k P(S_j = 1 \mid p_j, Z) \text{ a.e. under } P 
\end{equation} 

Define the selective p-values from \Cref{thm:adjustment} $p_{sel, j} = \int_0^{p_j} s_j(x, Z)dx/\int_0^1 s_j(x, Z) dx$. Our goal is to show that these selective p-values remain independent and are also valid p-values after selection:

\begin{equation}
    \label{eq:mult_error_control}
    P_{\cap_{i=1}^k H_{0, i}}(p_{sel, 1} \leq t_1, \dots, p_{sel, j} \leq t_1 | S_1=1, \dots, S_j=1) \leq \prod_{j=1}^k t_j \text{ for all } t \in [0,1 ]^k 
\end{equation}

Fix a distribution $P$ in the global null. As this statement is vacuously true when $P( S_1=1, \dots, S_j=1) = 0$, it suffices to suppose that  $P( S_1=1, \dots, S_j=1) > 0$, in which case we can define the conditional probability measures 
\begin{equation*}
    Q(A) = P(A| S_1=1, \dots, S_k= 1) = \frac{P(A,  S_1=1, \dots, S_k=1)}{P( S_1=1, \dots, S_k=1)}
\end{equation*}
and 
\begin{equation*}
    Q_j(A) = P(A| S_j=1) = \frac{P(A, S_j=1)}{P( S_j=1)}
\end{equation*}
Note that $Q$ is absolutely continuous with respect to all the $Q_j$ and also $P$, so any statement which is true a.e. under any $Q_j$ or a.e under $P$ is also true a.e. under $Q$.

By taking expectations of \eqref{eq:mult_indep} conditional on $Z$ with respect to both sides, we find that the $S_j$ are independent given $Z$ a.e. under $P$ (and thus also a.e. under $Q$). The following equalities hold a.e. under $P$. 
\begin{align*}
    P(S_1 = 1, \dots, S_k = 1 \mid Z) &= E[P(S_1 = 1, \dots, S_k = 1 | p_1, \dots, p_k, Z) \mid Z] \\
    &= E \left[\prod_{j = 1}^k P(S_j = 1 | p_j, Z) \;\middle|\; Z\right] & \text{(\Cref{eq:mult_indep}) }\\
    &= \prod_{j=1}^k E[P(S_j = 1 | p_j, Z)] & \text{ ($p_j$ conditionally independent given $Z$)  } \\
    &= \prod_{j=1}^k P(S_j = 1 | Z)
\end{align*}
where we have used that the $p_j$ are conditionally independent given $Z$ to move the expectation inside the product. 

Now, we recall from \Cref{sec:adjustment_proof} that the selective p-value $p_{sel, j}$ is a continuous function of $p_j$ for a.e. $Z$ under $Q_j$, and therefore this is true for a.e. $Z$ under $Q$. Namely, we have for all $t \in [0, 1]^k$ that 

\begin{align*}
    Q(p_{sel, 1} \leq t_1,  \dots, p_{sel, k} \leq t_k| Z) &=  Q(p_1 \in A_{1, Z, t}, \dots, p_k \in A_{k, Z, t} |Z) \text{ a.e. under } Q
\end{align*}
To be specific, for $z$ such that $\int_0^1 s_j(x, z)dx > 0 $, the $A_{j, z, t}$ are Borel sets such that $\int_0^y s_j(x, z)/\int_0^1 s_j(x, z) \leq t_j \iff y \in A_z$. Now, the following equalities hold for all $t \in [0, 1]^k$ a.e. under $Q$:

\begin{align*}
    &E_Q[ I(p_1 \in A_{1, Z, t}) \dots I(p_k \in A_{k, Z, t}) \mid Z] \\
    &= \frac{E_P[ I(p_1 \in A_{1, Z, t})S_1 \dots I(p_k \in A_{k, Z, t})S_k \mid Z] }{E_P[ S_1\dots S_k \mid Z]} & \text{(\Cref{lem:ce})}\\
    &= \frac{E_P[I(p_1 \in A_{1, Z, t}) \dots I(p_k \in A_{k, Z, t}) E_P[ S_1\dots S_k \mid p_1, \dots, p_k, Z] | \mid Z ]  }{E_P[ S_1\dots S_k \mid Z]}  & \text{(LoTE)}\\
    &= \frac{E_P[I(p_1 \in A_{1, Z, t}) E_P[S_1= 1 | p_k, Z] \dots I(p_k \in A_{k, Z, t}) E_P[S_k= 1 | p_k, Z]  | \mid Z ]  }{E_P[ S_1\dots S_k \mid Z]} & \text{(\Cref{eq:mult_indep})}\\
    &= \prod_{i=1}^k \frac{E_P[ I(p_i \in A_{i, Z, t}) E_P[S_i= 1 | p_i, Z]| Z]}{E_P[S_{i, Z} | Z]} & \text{(earlier result)}\\\\
    &= \prod_{i=1}^k \frac{E_P[ I(p_i \in A_{i, Z, t})S_i| Z]}{E_P[S_i | Z]} & \text{(LoTE)} \\
    &=  \prod_{i=1}^k E_{Q_i}[I(p_i \in A_{i, Z, t})| Z] & \text{(\Cref{lem:ce})}\\
    &= \prod_{i=1}^k E_{Q_i}[I(p_{sel, i} \leq t_i)| Z] & \text{(earlier reasoning)}\\
    &\leq \prod_{i=1}^k t_i & \text{(\Cref{thm:adjustment})}
\end{align*}
Of course, the fact that 
\begin{equation*}
    P(p_{sel, 1} \leq t_1, \dots, p_{sel, k} \leq t_k | Z, S_1=1, \dots, S_j=1) \leq \prod_{j=1}^k t_j \text{ for all } t \in [0,1 ]^k  \text{ a.e. under } Q 
\end{equation*}
is sufficient to imply \eqref{eq:mult_error_control}. 

\subsection{Proof of \Cref{cor:cfisher}}

Suppose we have $n$ independent and selectively dominant p-values $p_i$ for the nulls $H_{0, i}$. We consider a subset $p_{i_1}, \dots, p_{i_k}$ of them to use for inference. Denoting the remaining as $Z = (p_{i_{k+1}}, \dots, p_{i_n})$, we have that $p_{i_1}, \dots, p_{i_k}$ are conditionally independent given $Z$. We select the $p_{i_j}$ for $j \leq k$ to use for inference when they are smaller than all the p-values in $Z$, i.e., we use the selection functions $s_{j}(x, z) = I(x \leq \min_k z_k)$. For these selection functions, the condition \Cref{eq:mult_indep} is obviously satisfied, and \Cref{eq:mult_error_control} tells us that, under the global null $\cap_{j=1}^k H_{0, i_j} $, the selective p-values $p_{i_j}/\min_{\ell \not \in \{i_1, \dots, i_k \} } p_{\ell}$ are independent and valid p-values after selection. In particular, if we denote the $1-\alpha$ quantile of the $\chi^2_{d}$ distribution as $Q_d$, a consequence of \Cref{eq:mult_error_control} is that 
\begin{equation}
    \label{eq:cfisher_error_control}
    P_{\cap_{j=1}^k H_{0, i_j} }(-2 \sum_{j=1}^k \log(p_{i_j}/ \min_{\ell \not \in \{i_1, \dots, i_k\}} p_{\ell} ) \geq Q_{2k}|S_{i_1} = 1, \dots,  S_{i_k} = 1)
\end{equation}

To show conditional error control, we first note that if $ \cap_{j = 1}^k H_{0, i_j}$ is false, then trivially
\begin{equation*}
    P(\text{falsely reject } \cap_{i=1}^k H_{0, (i)} |S_{i_1} = 1, \dots,  S_{i_k} = 1) = 0 \leq \alpha.
\end{equation*} 
For the case that $\cap_{j = 1}^k H_{0, i_j}$ is true, then the event $S_{i_1} = 1, \dots,  S_{i_k} = 1$ is the same as selecting the $p_{i_1}, \dots, p_{i_k}$ for inference in \eqref{eq:cfisher_error_control}, so 
\begin{align*}
    &P(\text{falsely reject } \cap_{i=1}^k H_{0, (i)} |S_{i_1} = 1, \dots,  S_{i_k} = 1)\\
    &=P(-2 \sum_{j=1}^k \log(p_{(j)}/p_{(k+1)}) \geq Q_{2k} |S_{i_1} = 1, \dots,  S_{i_k} = 1)\\
    &=  P(-2 \sum_{j=1}^k \log(p_{i_j}/ \min_{\ell \not \in \{i_1, \dots, i_k\}} p_{\ell} ) \geq Q_{2k}|S_{i_1} = 1, \dots,  S_{i_k} = 1)\\
    &\leq \alpha 
\end{align*}
Now, letting $J$ denote the set of indices of the bottom $k$ p-values, It follows from the definition of conditional expectation that $P(A| J) = P(A | S_j = 1 \text{ for all } j \in J)$ a.e. under $P$, and the above therefore implies that $P(\text{falsely reject } \cap_{i=1}^k H_{0, (i)} | J)  \leq \alpha$ a.e. under $P$. Marginal error control then follows from the law of total expectation.

To see that the same applies for the global null repeat the same argument and replace any mentions of $\cap_{j = 1}^k H_{0, i_j}$ and $\cap_{j=1}^k H_{0, (j)}$ with $\cap_{i=1}^n H_{0, i}$. 

\subsection{Proof of \Cref{cor:tfisher}}

Suppose we have $n$ independent and selectively dominant p-values $p_i$ for the nulls $H_{0, i}$. We select $p_j$ to use for inference when it is smaller than some pre-specified threshold $\tau_j$, i.e., we use the selection functions $s_{j}(x) = I(x \leq \tau_j)$. Let $S_j$ be the indicator that the $j$th p-value is selected and $J$ be the random set of indices of the selected p-values. For these selection functions, the condition \eqref{eq:mult_indep} is obviously satisfied, and \eqref{eq:mult_error_control} tells us that, under the global null $\cap_{j=1}^k H_{0, i_j} $, the selective p-values $p_{i_j}/\tau_j$ are independent and valid p-values after selection. In particular, if we denote the $1-\alpha$ quantile of the $\chi^2_{d}$ distribution as $Q_d$, a consequence of \Cref{eq:mult_error_control} is that 
\begin{equation*}
    \label{eq:tfisher_error_control}
    P_{\cap_{j=1}^k H_{0, i_j} }(-2 \sum_{j=1}^k \log(p_{i_j}/ \tau_j ) \geq Q_{2k}|S_{i_1} = 1, \dots,  S_{i_k} = 1)
\end{equation*}

To show conditional error control, we first note that if $ \cap_{j = 1}^k H_{0, i_j}$ is false, then trivially
\begin{equation*}
    P(\text{falsely reject } \cap_{j \in J}^k H_{0, j} |S_{i_1} = 1, \dots,  S_{i_k} = 1) = 0 \leq \alpha.
\end{equation*} 
For the case that $\cap_{j = 1}^k H_{0, i_j}$ is true, then the event $S_{i_1} = 1, \dots,  S_{i_k} = 1$ is the same as selecting the $p_{i_1}, \dots, p_{i_k}$ for inference in \eqref{eq:tfisher_error_control}, so 
\begin{align*}
    &P(\text{falsely reject } \cap_{j \in J}^k H_{0, j} |S_{i_1} = 1, \dots,  S_{i_k} = 1)\\
    &=P(-2 \sum_{j\in J} \log(p_{j}/ \tau_j ) \geq Q_{2k} |S_{i_1} = 1, \dots,  S_{i_k} = 1)\\
    &=  P(-2 \sum_{j=1}^k \log(p_{i_j}/ \tau_j ) \geq Q_{2k}|S_{i_1} = 1, \dots,  S_{i_k} = 1)\\
    &\leq \alpha 
\end{align*}
It follows from the definition of conditional expectation that $P(A| J) = P(A | S_j = 1 \text{ for all } j \in J)$ a.e. under $P$, and the above therefore implies that $P(\text{falsely reject } \cap_{j \in J} H_{0, j} | J)  \leq \alpha$ a.e. under $P$. Marginal error control then follows from the law of total expectation.

To see that the same applies for the global null repeat the same argument and replace any mentions of $\cap_{j = 1}^k H_{0, i_j}$ and $\cap_{j \in J} H_{0, j}$ with $\cap_{i=1}^n H_{0, i}$. 

\section{Proofs}
\label{sec:proofs_appdx}

\subsection{Proof of \Cref{thm:adjustment}}
\label{sec:adjustment_proof}

First we prove a useful lemma.

\begin{lemma}
    \label{lem:ce}
    For a measure $P$ and an event $B$ with $P(B)> 0$ define the conditional measure $Q(E) = P(E| B) = P(E \cap B)/P(B)$. Then, for a random variable $X$, 
    \begin{equation*}
        Q(A | X) = \frac{P(A, B | X) }{ P(B | X) } \text{ a.e. under } Q
    \end{equation*} 
\end{lemma}

\begin{proof}

    First note that $E_P[I( P(B| X) = 0 )  I(B)] = E_P[I( P(B| X) = 0 )  E_P[ I(B) | X] ] = 0$, so $P(B | X) > 0$ a.e. on the set $I(B)$, Thus it must be the case that $P(B | X) > 0 $ a.e. under $Q$. Now we can compute for any Borel $C$ that  
    \begin{align*}
    &E_Q\left[\frac{P(A, B | X) }{ P(B | X) } I(P(B|X) > 0) I(X \in C) \right] \\
    &= E_P\left[\frac{P(A, B | X) }{ P(B | X) } I(P(B| X) > 0) I(X \in C) I(B)\right] / P(B)  & \text{(standard machine)}\\
    &= E_P\left[\frac{P(A, B | X) }{ P(B | X) } I(P(B| X) > 0) I(X \in C) P(B \mid X)\right] / P(B) & \text{(tower property)}  \\
    &= E_P[P(A, B \mid X) I(P(B|X) > 0) I(X \in C)] / P(B) & \\
    &= P(A, B, P(B|X) > 0, X \in C) / P(B) & \text{(definition of conditional expectation)}\\
    &= Q(A, P(B|X) > 0, X \in C) & \text{ (definiton of $Q$)}\\
    &= E_Q[I(A) I(X \in C)] & \text{ ($P(B|X)$ positive $Q$ a.e.)}\\ 
    &= E_Q[Q(A \mid X) I(X \in C)]. & \text{(definition of conditional expectation)}
    \end{align*}

    Thus $\frac{P(A, B | X) }{ P(B | X) } I(P(B|X) > 0)$ satisfies the definition of the conditional expectation of $I(A)$ given $X$ under $Q$. Since $I(P(B|X) > 0) = 1$ a.e. under $Q$, we are done.  

\end{proof}

Equipped with this lemma, we can provide a proof. 

Fix a probability distribution $P$ in $H_0$ and let $f_z(x)$ be the conditional PDF of $p$ given $Z=z$. If $P(S=1) = 0$, then our convention is that $P(p_{sel} \leq \alpha | Z, S=1)$ and $P(p_{sel} \leq \alpha | S=1)$ are identically zero, so the result is immediate. 

Now consider the non-trivial case, where $P(S=1) > 0$ and the conditional probability measure $Q(A)= P(A | S=1) = P(A, S=1)/P(S=1)$ is well defined. We first show that \eqref{eq:adjusted_error_control_cond} holds. 

We start by showing that $P(S'=1) > 0$ also. First, we compute 
\begin{equation*}
    E_P[S | Z] = E_P[ E_P[S|p, Z]| Z] = E_P[s(p, Z)| Z] = \int_0^1  s(x, Z)f_Z(x) dx 
\end{equation*}
and likewise
\begin{equation*}
    E_P[S' | Z] = E_P[ E_P[S|U, Z]| Z] = E_P[s(U, Z)| Z] = \int_0^1 s(x, Z) dx. 
\end{equation*}
Whenever $\int_0^1 f_z(x) s(x, z) dx > 0$, it must be the case this $z$ that the function $s(x, z)$ is not zero Lebesgue almost everywhere in $x$. Therefore $\int_0^1 s(x, z) dx > 0 \implies \int_0^1 f_z(x) s(x, z) dx > 0 $. Thus we have both that 
\begin{equation*}
    E_P[S|Z] > 0 \implies E_P[S'| Z] > 0 \text{ everywhere } 
\end{equation*}
and the contrapositive
\begin{equation*}
    E_P[S'|Z] = 0 \implies E_P[S| Z] = 0 \text{ everywhere } P 
\end{equation*}

Now, since $P(S = 1) = E_P[E_P[S| Z]]$, it cannot be the case that $E_P[S|Z] = 0$ a.e. under $P$.  By the first implication above we must then also have $P(S' = 1) = E_P[E_P[S'| Z]] > 0$ as well. Thus we can also define the conditional probability measure $Q'(A) = P(A | S'=1) = P(A, S'=1)/P(S'=1)$. Now, with both these probability measures at our disposal, we prove two lemmas. 

\begin{lemma}
    \label{lem:pos_selection}
    $E_P[S|Z]> 0$ a.e. under $Q$ and $E_P[S'|Z] > 0$ a.e. under $Q'$.
\end{lemma}

\begin{proof}
    The fact that $E_P[S|Z]> 0$ a.e. under $Q$ follows from the fact that 
    \begin{align*}
        P(E_P[S| Z] = 0, S=1) &= E_P[ I(E_P[S| Z] = 0) S] \\
                            &= E_P[ I(E_P[S| Z] = 0) E_P[S |Z]] \\
                            &= 0
    \end{align*}
    The fact that $E_P[S'|Z] > 0$ a.e. under $Q'$ can be shown by an identical proof. 
    
\end{proof}

\begin{lemma}
    \label{lem:switch}
    $Q'(Z \in A) = 0 \implies Q(Z \in A) = 0$, and as a consequence  $Q'(Z \in A) = 1 \implies Q(Z \in A) = 1$
\end{lemma}

\begin{proof} 
    Since $Q'(Z \in A) = P(Z \in A, S'=1)/P(S'=1)$, we have
    \begin{align*}
        &Q'(Z \in A) = 0 \\
        &\iff P(Z \in A, S'=1) = E_P[ I(Z \in A) E_P[ S' |Z]] = 0\\
        &\iff I(Z \in A) E_P[ S' |Z] = 0 \text{ a.e. under } P \\
        &\iff P(Z \in A, E_P[S'|Z]>0) = 0 
    \end{align*} 
    By our earlier results, however,  $P(Z \in A, E_P[S|Z]>0) \leq P(Z \in A, E_P[S'|Z]>0) = 0$. Thus  $P(Z \in A, E_P[S|Z]>0) =0$, and the same set of equivalences tell us that $Q(Z \in A) = 0$. 

    The consequence follows from the fact that $Q'(Z \in A) = 1 \implies Q'(Z \in A^c) = 0 \implies Q(Z \in A^c) = 0 \implies Q(Z \in A) = 1$. 
\end{proof}

Now, define the function 
\begin{equation*}
    F_{U|z, S'=1}(t) = \frac{\int_0^t s(x, z) dx }{\int_0^1 s(x, z) dx}
\end{equation*}
so that $p_{sel} = F_{U| Z, S'}(p)$. For $z$ such that $\int_0^1 s(x, z) dx > 0$, the above is a monotone non-decreasing and continuous function of $t$. For such $z$, we can define its generalized inverse
\begin{equation*}
    F^{-1}_{U |z, S'=1}(x)  = \inf \{t : F_{U |z, S' = 1}(t) > x  \}
\end{equation*}
which satisfies $F_{U|z, S' =1}(F^{-1}_{U |z, S'=1}(x)) = x$ and $F_{U | z, S' = 1}(t) \leq x \iff t \leq  F^{-1}_{U | z, S' = 1}(x)$. By \Cref{lem:pos_selection} and the fact that $E_P[S' | Z] > E_P[S|Z]$ everywhere, $F_{U|Z S'=1}(t)$ will be a continuous function of $x$ with probability one both under $Q$ and $Q'$. The following string of equalities and inequalities then hold almost surely under $Q$:

\begin{align*}
    &Q(F_{U | Z, S' = 1}(p) \leq t | Z)\\
     &= Q(p \leq F_{U | Z, S' = 1}^{-1}(t) | Z) & \text{($F_{U|Z S'=1}(t)$ continuous a.e. under $Q$) } \\
    &\leq Q'(U \leq F_{U |Z, S' = 1}^{-1}(t) |Z) & \text{(selective dominance) } \\
    &= t & \text{ (\Cref{lem:ce}) }
\end{align*}
We justify the last equality. First we use \Cref{lem:ce} to note that a.e. under $Q'$ we have 
\begin{align*}
    Q'(U \leq t |Z) &= \frac{P(U \leq t, S' = 1 | Z )}{P( S' = 1 | Z)}\\
                    &= \frac{E_P[I(U \leq t)P(S' = 1 | U, Z ) | Z ] }{E_P[P(S' = 1 | U, Z)| Z]}\\
                    &= \frac{\int_0^t s(x, Z) dx }{\int_0^1 s(x, Z) dx}\\
                    &= F_{U| Z, S'=1}(t)
\end{align*}
Thus a.e. under $Q'$ it holds that
\begin{equation*}
    Q'(U \leq F_{U |Z, S' = 1}^{-1}(t) |Z) = F_{U| Z, S'=1}(F^{-1}_{U| Z, S'=1}(t)) = t
\end{equation*}
Writing the event $Q'(U \leq F_{U |Z, S' = 1}^{-1}(t) |Z) = t$ as $Z \in A$ for some Borel set $A$, \Cref{lem:switch} tells us that $Q'(Z \in A) = 1 \implies Q(Z \in A) = 1$, which establishes the claim. Finally, we \eqref{eq:adjusted_error_control_marg} from the law of total expectation.  

In the case that $Q(p \leq t | Z) = t$ a.e. under $Q$, then by the same computations as earlier we can directly compute that $Q(p \leq t | Z) = F_{U| Z, S'=1}(t)$ a.e. under $Q$, and we can directly compute that $Q(p \leq F^{-1}_{U | Z, S' = 1}(t) | Z) =  F_{U| Z, S'=1}(F^{-1}_{U | Z, S' = 1}(t)) = t $ a.e. under $Q$, i.e., we get exact equality a.e. under $Q$ as claimed in \eqref{eq:adjusted_error_control_equality}.

\subsection{Proof of \Cref{thm:density}}
\label{sec:density_proof}

Fix a probability distribution $P$ in $H_0$ and let $f_z(x)$ be the conditional PDF of $p$ given $Z=z$. If $P(S=1) = 0$, then our convention is that $P(p \leq \alpha | Z, S=1) = 0$, so \eqref{eq:selective_dominance} holds automatically. Thus we can consider only selection functions $s(x, z)$ for which $P(S=1) > 0$. We argued in \Cref{sec:adjustment_proof} that when $P(S=1) > 0$ we also have $P(S' = 1)> 0$. Correspondingly, we define the conditional measures $Q(A)= P(A | S=1) = P(A, S=1)/P(S=1)$ and $Q'(A)= P(A | S=1) = P(A, S=1)/P(S=1)$. 

For now, consider fixed $z$ for which $\int_0^1 s(x, z) f_z(x) dx >0$. For such $z$, we argued in \Cref{sec:adjustment_proof} that 
 $\int_0^1 s(x, z) f_z(x) dx >0$ also. For such $z$, we argue that the inequality 
 \begin{equation*}
    \frac{\int_{0}^{t} s(x, z) f_z(x) dx }{\int_{0}^{1} s(x, z) f_z(x) dx } \leq \frac{\int_{0}^{t} s(x, z) dx}{\int_{0}^{1} s(x, z) dx } \\
 \end{equation*}
 holds for all $t\in[0, 1]$. If $\int_{0}^{t} s(x, z) f_z(x) dx =0$ then the inequality trivially holds. Otherwise, we see that 
 \begin{align*}
    \frac{ \int_{0}^{1} s(x, z) f_z(x) dx }{\int_{0}^{t} s(x, z) f_z(x) dx } &= 1 + \frac{\int_{t}^{1} s(x, z) f_z(x) dx }{\int_{0}^{t} s(x, z) f_z(x) dx }\\
    &\geq 1 + \frac{f_z(t)\int_{t}^{1} s(x ,z) dx }{f_z(t)\int_{0}^{t} s(x, z)  dx }\\
    &= 1 + \frac{\int_{t}^{1} s(x, z) dx }{\int_{0}^{t} s(x, z)  dx }\\
    &= \frac{ \int_{0}^{1} s(x, z) dx }{\int_{0}^{t} s(x, z)  dx },
\end{align*}
which is sufficient to imply the claim.

To finish the argument, we note by \Cref{lem:pos_selection} that $E[S|Z] = \int_0^1 s(x, Z) f_Z(x) dx >0 $ a.e. under $Q$. Applying \Cref{lem:ce}, we find that a.e. under $Q$
\begin{align*}
    Q(p \leq t |Z) &= \frac{P(p \leq t, S = 1 | Z )}{P(S = 1 | Z)}\\
                    &= \frac{E_P[I(p \leq t)P(S = 1 | p, Z ) | Z ] }{E_P[I(p \leq t)P( S = 1 | p, Z)| Z]}\\
                    &= \frac{\int_0^t s(x, Z) f_Z(x) dx }{\int_0^1 s(x, Z) f_Z(x) dx}\\
\end{align*}
and a.e. under $Q'$ 
\begin{align*}
    Q'(U \leq t |Z) &= \frac{P(U \leq t, S' = 1 | Z )}{P(S' = 1 | Z)}\\
                    &= \frac{E_P[I(U \leq t) P( S' = 1 | U, Z ) | Z ] }{E_P[P(S' = 1 | U, Z)| Z]}\\
                    &= \frac{\int_0^t s(x, Z) dx }{\int_0^1 s(x, Z) dx}\\
\end{align*}
But, we can write the event that $Q'(U \leq t |Z) = \frac{\int_0^t s(x, Z) dx }{\int_0^1 s(x, Z) dx}$ as $Z \in A$ for some Borel set $A$, and \Cref{lem:switch} tells us that since $Z \in A$ is measure one under $Q'$ it is also measure one under $Q$. That is, we have a.e. under $Q$ that 
\begin{equation*}
    Q'(U \leq t |Z) = \frac{\int_0^t s(x, Z) dx }{\int_0^1 s(x, Z) dx}
\end{equation*}
Combined, we have a.e. under $Q$ that $ Q(p \leq t |Z) \leq Q'(U \leq t |Z)$, which is the desired claim. 

Now we show the converse part of the statement. Again fix a probability distribution $P$ in $H_0$, and suppose there is a set $B$ such that $P(Z \in B)> 0$ and for all $z \in B$ the conditional density $f_{z}$ is everywhere continuous and not non-decreasing. More specifically, all we require is that for all $z\in E$, there are two points $y_1(z) < y_2(z)$ such that $f_{z}(x)$ is strictly larger in a neighborhood around $y_1(z)$ than in a neighborhood around $y_2(z)$, and these neighborhoods are disjoint. In particular for $\epsilon > 0$ let $N_{\epsilon}(y) = (y - \epsilon, y + \epsilon )$ be a ball around $y$. Then we need there to be $y_1(z)$, $y_2(z)$, $\epsilon(z) > 0$, and some $\eta(z) > 0$ such that, for all $w_1 \in N_{\epsilon(z)}(y_1(z))$ and $w_2 \in N_{\epsilon(z)}(y_2(z))$, $w_1  < w_2$ but $f_{z}(w_2) + \eta(z) < f_{z}(w_1) $. If $f_z$ is everywhere continuous and not non-decreasing, this will be true. Define $H(z) = \inf \{f_{z}(w_1): w_1 \in N_{\epsilon(z)}(y_1(z))\}$ and  $L(z) = \sup \{f_{z}(w_2): w_2 \in N_{\epsilon(z)}(y_2(z))\}$ so $H(z)  > L(z)$. Then consider the selection function 
\begin{equation*}
s(x, z)= \begin{cases}
1 &\text{if } x \in N_{\epsilon(z)}(y_1(z)) \cup N_{\epsilon(z)}(y_2(z)) \text{ and } z \in B \\
0 &\text{otherwise }.
\end{cases}
\end{equation*}
We show that, for this selection function, selection happens with positive probability:
\begin{align*}
    P(S = 1) &= E_P[E_P[ P(S = 1 | p, Z) | Z ]] \\
             &= E_P[\int_0^1 s(x, Z)f_Z(x) dx]\\
             &\geq E_P[\int_0^1 s(x, Z) f_Z(x)dx I(Z \in B)]\\
             & \geq E_P[2\epsilon(Z)\eta(Z)I(Z \in B)] > 0
\end{align*}
Thus the conditional probability measure $Q$ and $Q'$ are well defined. Next we show that $Q(Z \in E) = 1$, so we know that $Z \in E$ a.e. under $Q$:
\begin{align*}
    Q(Z \in E) &= \frac{P(Z \in E, S=1)}{P(S=1)}\\
               &= \frac{E_P[I(Z \in E) E_P[ P(S=1 | p, Z)| Z]] }{E_P[E_P[P(S=1| p, Z) | Z] ]}\\
               &= \frac{E_P[ I(Z \in E) \int_0^1 s(x, Z) f_Z(x) dx ] }{E_P[\int_0^1 s(x, Z) f_Z(x) dx ]}\\
               &= \frac{E_P[ I(Z \in E) \int_0^1 s(x, Z) f_Z(x) dx ] }{E_P[ I(Z \in E) \int_0^1 s(x, Z) f_Z(x) dx ]}\\
               &= 1,
\end{align*}
where we have noted that $s(x, z)$ is zero everywhere if $z \not \in E$. 

Finally, for $z \in E$, define $t(z)$ to be a value such that $t(z) > w_1$ for all $w_1 \in N_{\epsilon(z)}(y_1(z))$ and $t(z) < w_2$ for all $w_2 \in  N_{\epsilon(z)}(y_2(z))$. We know from earlier that a.e. under $Q$ 
\begin{equation*}
    Q(p \leq t(Z) |Z) = \frac{\int_0^{t(Z)} s(x, Z) f_Z(x) dx }{\int_0^1 s(x, Z) f_Z(x) dx},
\end{equation*}
where the numerator is at least $2\epsilon(Z)\eta(Z) >0$, so also a.e. under $Q$
so, 
\begin{align*}
    \frac{1}{Q(p \leq t(Z) |Z)} &= 1 + \frac{\int_{t(Z)}^1 s(x, Z) f_Z(x) dx }{\int_0^{t(Z)} s(x, Z) f_Z(x) dx}\\
                                &\leq 1 + \frac{2 \epsilon(Z) L(Z)}{2 \epsilon(Z)H(Z)}\\
                                &< 2, \\
\end{align*}
and 
\begin{align*}
    Q(p \leq t(Z) |Z) > 1/2.
\end{align*}
On the other hand, we also have from above that a.e. under $Q$ 
\begin{align*}
    Q'(U \leq t(Z) |Z) &= \frac{\int_0^{t(Z)} s(x, Z) dx }{\int_0^1 s(x, Z) dx}\\
                       &= \frac{2\epsilon(Z)}{4 \epsilon(Z)}\\
                       &= 1/2.
\end{align*}
Thus we have shown that, a.e. under $Q$, there exists some $t \in [0, 1]$ such that $Q(p \leq t | Z) > Q'(U \leq t)$, which means that $p$ is not selectively dominant given $Z$.

\subsection{Proof of \Cref{cor:hyb}}
\label{sec:hyb_proof_appdx}

Suppose we have $n$ independent and selectively dominant p-values $p_i$ for the null hypotheses $H_{0, i}$. Suppose we use $p_j$ to test $H_{0, j}$ only when we observe that $p_j$ is strictly larger than $\beta_n$ but still the smallest of all the p-values. We can apply \Cref{sec:dominance}'s framework with $p=p_j$, $Z = p_{-j}$ and the selection function $s(x, z) = I(\beta_n < x < \min_{k} z_k)$. It is straightforward to see that \Cref{thm:adjustment}'s selective p-value is $p_{sel}$ is $(p_j - \beta_n)/(\min_{i \neq j} p_i - \beta_n)$, and \Cref{thm:adjustment} therefore tells us that 
\begin{equation*}
    P_{H_{0, j}}\left( \frac{p_j - \beta_n}{\min_{i \neq j} p_i - \beta_n} \leq \frac{\alpha -\beta}{1-\beta}  \;\middle|\;  S=1 \right) \leq \frac{\alpha - \beta}{1-\beta}.
\end{equation*}
Re-arranging things we get 
\begin{equation}
    \label{eq:hybrid_tool}
    P_{H_{0, j}}\left( p_j  \leq \frac{\alpha -\beta}{1-\beta}\min_{i \neq j} p_i  + \left( 1 - \frac{\alpha -\beta}{1-\beta}\right) \beta_n  \;\middle|\; S=1  \right) \leq \frac{\alpha - \beta}{1-\beta}.
\end{equation}

Letting $W$ be the index of the smallest p-value (with ties broken randomly), we can now prove the claim that rejecting $H_{0, W}$ when 
\begin{equation*}
    p_{(1)} \leq \frac{\alpha-\beta}{1-\beta} p_{(2)} + \left(1 - \frac{\alpha-\beta}{1-\beta} \right) \beta_n 
\end{equation*}
controls Type I error at level $\alpha$. Considering any distribution $P$, let $B$ be the event that a p-value corresponding to a true null is at most $\beta_n$.  Then $P(B) \leq \beta$ and $P(B^c) \geq 1-\beta$. If $H_{0, j}$ is false, we trivially have have $P(\text{falsely reject } H_{0, W} | W = j, B^c) = 0 \leq (\alpha-\beta)/(1-\beta)$. If $H_{0, j}$ is true, the event $W=j \text{and} B^c$ is the same event as selecting $p_j$ for inference in \eqref{eq:hybrid_tool}, so 
\begin{align*}
    P(\text{falsely reject } H_{0, W} | W = j, B^c) &= P\left(p_{(1)} \leq \frac{\alpha-\beta}{1-\beta} p_{(2)} + \left(1 - \frac{\alpha-\beta}{1-\beta} \right) \beta_n  \;\middle|\; W = j, B^c\right)\\
    &=  P\left(p_j \leq \frac{\alpha-\beta}{1-\beta} \min_{i \neq j}p_i + \left(1 - \frac{\alpha-\beta}{1-\beta} \right) \beta_n  \;\middle|\; W = j, B^c\right)\\
    &\leq \frac{\alpha - \beta}{1-\beta}.
\end{align*}
Our result then follows from law of total probability:
\begin{align*}
    &P(\text{falsely reject } H_{0, W} )\\
    &=P(\text{falsely reject } H_{0, W}, B ) + P(\text{falsely reject } H_{0, W}, B^c )\\
    &\leq P(B)  + P(B^c) P(\text{falsely reject } H_{0, W}| B^c )\\
    &\leq P(B)  + P(B^c)  \sum_{j=1}^n P(\text{falsely reject } H_{0, W}|W =j, B^c ) P(W=j | B^c)\\
    &= P(B)  + (1 - P(B)) \frac{\alpha-\beta}{1-\beta}  \sum_{j=1}^n  P(W=j | B^c)\\
    &\frac{\alpha -\beta}{1-\beta} + \frac{1-\alpha}{1-\beta} P(B)\\
    &\leq \frac{\alpha -\beta}{1-\beta} + \frac{1-\alpha}{1-\beta} \beta\\
    &= \alpha. 
\end{align*}

\subsection{Proof of \Cref{cor:hyb_closed} and  \Cref{cor:cond_closed}}

It suffices to argue that closing our hybrid global null testing procedure rejects $H_{0, (k)}$ if and only if 
\begin{equation*}
p_{(j)} \leq \frac{\alpha - \beta}{1-\beta} p_{(j + 1)} + \left(1 - \frac{\alpha - \beta}{1-\beta} \right)\beta_{n - j + 1} 
\end{equation*}
for all $j \leq k$. We will define $p_{(n+1)} = 1$ so that the right-hand side of the above equals $\alpha$ when $j = n$. Correspondingly, for subsets $I \subseteq [p]$ of size one, we define the hybrid procedure to reject the global null $H_{I, 0}$ when the lone p-value is at most $\alpha$. For subsets $I$ of size strictly more than one, supposing that the smallest p-value in $I$ is the $\ell$th smallest p-value and the second smallest p-value in $I$ is the $m$th smallest p-value, the hybrid procedure rejects the global null $H_{I, 0}$ when
\begin{equation*}
    p_{(\ell)} \leq \frac{\alpha-\beta}{1-\beta} p_{(m)} + \left(1 - \frac{\alpha - \beta}{1-\beta} \right)\beta_{|I|}.
\end{equation*}

\noindent \textbf{Necessity: } For $1 \leq j \leq k$, let $I_{n-j + 1}$ be the size $n - j + 1$ subset that excludes the $j - 1$ smallest p-values (when $j = 1$ then $I = [p]$). This subset includes the index of the $k$th smallest p-value, so we must reject $H_{0, I}$ to reject $H_{0, (k)}$. It rejects exactly when
\begin{equation*}
    p_{(j)} \leq \frac{\alpha - \beta}{1-\beta} p_{(j + 1)} + \left(1 - \frac{\alpha - \beta}{1-\beta} \right)\beta_{n - j + 1}
\end{equation*}
so our conditions are necessary. \newline 

\noindent \textbf{Sufficiency: } Consider a subset $I$ that contains the index of the $k$th smallest p-value. If it is size-one, then we reject because 
\begin{equation*}
    p_{(k)} \leq \frac{\alpha - \beta}{1-\beta}p_{(k + 1)} + \left(1 - \frac{\alpha - \beta}{1-\beta} \right) \beta_{n - k + 1} \leq \frac{\alpha - \beta}{1-\beta} + \left(1 - \frac{\alpha - \beta}{1-\beta} \right)\beta = \alpha. 
\end{equation*}
Now suppose that $I$ is size $n-j+1$ for some $j < n$.  Its smallest p-value is the $\ell$th smallest p-value for some $\ell \leq k$ and $\ell \leq j$, and its second smallest p-value is the $m$th smallest p-value for some $m > \ell$. We reject because 
\begin{align*}
    p_{(\ell)} &\leq \frac{\alpha-\beta}{1-\beta} p_{(\ell + 1)} + \left(1 - \frac{\alpha - \beta}{1-\beta} \right)\beta_{ n - \ell + 1}\\
    &\leq \frac{\alpha-\beta}{1-\beta} p_{(m)} + \left(1 - \frac{\alpha - \beta}{1-\beta} \right)\beta_{ n - j + 1} \\
    &= \frac{\alpha-\beta}{1-\beta} p_{(m)} + \left(1 - \frac{\alpha - \beta}{1-\beta} \right)\beta_{|I|}.
\end{align*}

\end{appendix}

\end{document}